\documentclass[11pt,reqno]{article}
\usepackage{amsmath,amsthm,amssymb}
\usepackage{graphicx}
\usepackage{bm}

\def\R{\mathbb{R}}

\def\P{\mathbb{P}}
\def\E{\mathbb{E}}
\def\N{\mathbb{N}}
\def\Z{\mathbb{Z}}
\def\T{\mathbb{T}}
\def\CC{\mathcal{C}}
\def\PP{\mathcal{P}}
\def\TT{\mathcal{T}}

\def\L{\mathcal{L}}
\def\A{\mathcal{A}}

\DeclareMathOperator{\var}{Var}

\DeclareMathOperator{\Bin}{Bin}

\newtheorem{theorem}{Theorem}[section]
\newtheorem{proposition}[theorem]{Proposition}
\newtheorem{lemma}[theorem]{Lemma}
\newtheorem{remark}{Remark}[section]
\newtheorem{corollary}[theorem]{Corollary}

\def\eps{\varepsilon}

\def\eqd{\,{\buildrel d \over =}\,}

\newcommand{\diam}{\text{diam}}

\newcommand\adj[1]{\sim_{#1}}

\newcommand\even{\text{even}}
\newcommand\odd{\text{odd}}

\newcommand\Hom{\operatorname{Hom}}
\newcommand\Lip{\operatorname{Lip}}
\newcommand\col{\operatorname{Col}}
\newcommand\MCut{\operatorname{MCut}}
\newcommand\OMCut{\operatorname{OMCut}}
\newcommand\LS{\operatorname{LS}}
\newcommand\PLS{\operatorname{PLS}}
\newcommand\comp{\operatorname{comp}}
\newcommand\subcut{\operatorname{subcut}}
\newcommand\VCut{\operatorname{VCut}}
\newcommand\GC{G^\boxtimes}
\newcommand\range{\operatorname{Range}}
\newcommand\unifin{\in_R}

\newcommand\Out{\text{Out}}
\newcommand\In{\text{In}}

\DeclareMathOperator{\vol}{Vol}
\DeclareMathOperator{\triv}{Triv}
\DeclareMathOperator{\shift}{Shift}
\newcommand\Ei{E_{\operatorname{in}}}
\newcommand\Eonee{E_{1,\operatorname{e}}}
\newcommand\Omegaw{\Omega_{\operatorname{w}}}

\title{High-Dimensional Lipschitz Functions are Typically Flat}

\author{Ron Peled\thanks{School of Mathematical Sciences, Tel Aviv University; e-mail: peledron@post.tau.ac.il; webpage: http://www.math.tau.ac.il/{\raise.17ex\hbox{$\scriptstyle\sim$}}peledron/\ . Research supported by NSF Grant OISE 0730136.}}

\begin{document}
\maketitle

\abstract{A \emph{homomorphism height function} on the $d$-dimensional torus $\Z_n^d$ is a function on the vertices of the torus taking integer values and constrained to have adjacent vertices
take adjacent integer values. A \emph{Lipschitz height function} is defined similarly but may also take equal values on adjacent vertices. For each of these models, we consider the uniform distribution over all such functions with predetermined values at some fixed vertices (boundary conditions).
Our main result is that in high dimensions and with zero boundary values, the random function obtained is typically very flat, having bounded variance at any fixed
vertex and taking at most $C(\log n)^{1/d}$ values with high probability. This result matches, up to constants, a lower
bound of Benjamini, Yadin and Yehudayoff. Our results extend to any dimension $d\ge 2$, if one replaces the torus $\Z_n^d$ by an enhanced version of it, the torus $\Z_n^d\times\Z_2^{d_0}$ for some fixed $d_0$. Consequently, we establish one side of a conjectured roughening transition in 2 dimensions. The full transition is established for a class of tori with non-equal side lengths, including, for example, the $n\times\lfloor\frac{1}{10}\log n\rfloor$ torus. In another case of interest, we find that when the dimension $d$ is taken to infinity while $n$ remains fixed, the random function takes at most $r$ values with high probability, where $r=5$ for the homomorphism model and $r=4$ for the Lipschitz model. Suitable generalizations are obtained when $n$ grows with $d$. Our results have consequences also for the related model of uniform 3-coloring and establish that for certain boundary conditions, a uniformly sampled proper 3-coloring of $\Z_n^d$ will be nearly constant on either the even or odd sub-lattice.

Our proofs are based on the construction of a combinatorial
transformation suitable to the homomorphism model and on a careful
analysis of the properties of a class of cutsets which we term
\emph{odd cutsets}. For the Lipschitz model, our results rely also
on a bijection of Yadin. This work generalizes results of Galvin and
Kahn, refutes a conjecture of Benjamini, Yadin and Yehudayoff and
answers a question of Benjamini, H\"aggstr\"om and Mossel.}

\renewcommand{\thefootnote}{\fnsymbol{footnote}}
\renewcommand{\thefootnote}{\arabic{footnote}}

\newpage
\tableofcontents

\newpage
\section{Introduction}
In this paper we study the models of homomorphism and Lipschitz height functions. Given a graph $G$ and function $f:V[G]\to\Z$, where $V[G]$ is the vertex set of $G$ and $\Z$ is the set of integers, we call $f$ a \emph{homomorphism height function} if $|f(v)-f(w)|=1$ whenever $v$ and $w$ are adjacent in $G$ (and thus $f$ is a graph homomorphism of $G$ to $\Z$). We call it a \emph{Lipschitz height function} if $|f(v)-f(w)|\le 1$ whenever $v$ and $w$ are adjacent in $G$. Homomorphism height functions are a sub-class of Lipschitz height functions and $G$ admits them if and only if it is bipartite. We are interested in the typical properties of random height functions chosen uniformly at random from the set of homomorphism, or Lipschitz, functions satisfying specified boundary conditions. This model was introduced by Benjamini, H\"aggstr\"om and Mossel \cite{BHM00} (when $G$ is a tree, the model was investigated earlier, see \cite{BP94}) and further investigated in \cite{BS00}, \cite{K01}, \cite{LNR03}, \cite{G03}, \cite{BYY07} and \cite{E09}. To define the model precisely, we assume that $G$ is finite, connected and bipartite and take a subset $\emptyset\neq B\subseteq V[G]$ and function $\mu:B\to\Z$. We then restrict attention to the sets $\Hom(G,B,\mu)$ and $\Lip(G,B,\mu)$ of homomorphism and Lipschitz height functions $f$, respectively, for which $f(b)=\mu(b)$ for all $b\in B$.
The pair $(B,\mu)$ is called the boundary condition, or BC. The special case when $B$ is a singleton and $\mu$ equals zero on $B$ is of particular interest and we term it - a \emph{one-point BC}. Assuming $\Hom(G,B,\mu)\neq\emptyset$, we denote by $f\unifin\Hom(G,B,\mu)$ a function sampled uniformly at random from $\Hom(G,B,\mu)$. Such an $f$ is called the \emph{random height function} for the homomorphism model with boundary condition $(B,\mu)$. Assuming $\Lip(G,B,\mu)\neq\emptyset$, we similarly define $f\unifin\Lip(G,B,\mu)$, the random height function for the Lipschitz model. Our main object of study are the fluctuations of the random height function $f$ around its mean, as realized for example by $\var(f(v))$ for vertices $v\in V[G]$, by the number of values $f$ takes, or by a global structure $f$ may exhibit.

\begin{figure}[p!]
\centering
{\includegraphics[width=\textwidth,viewport=245 80 805 665,clip]{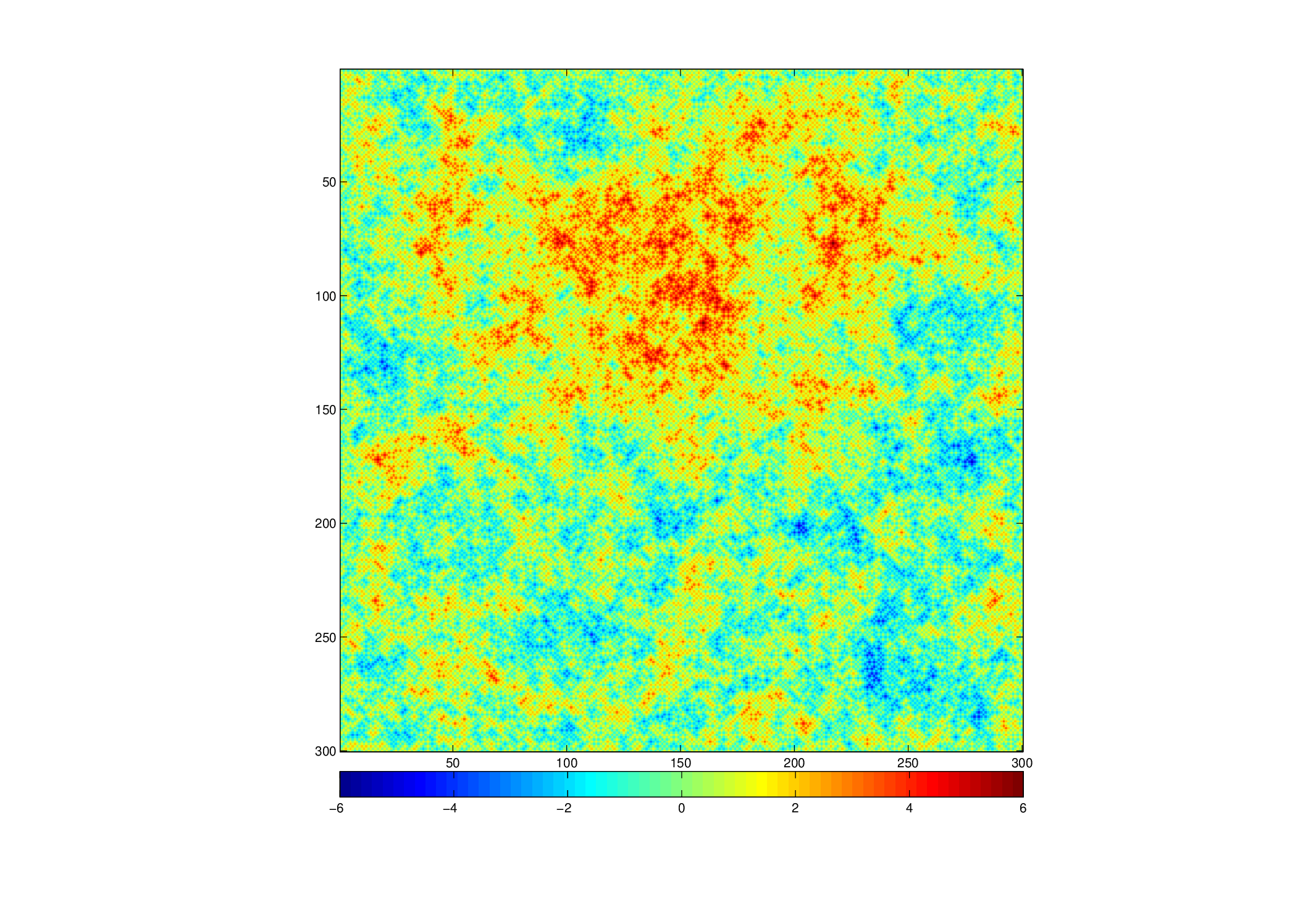}}
{\it
\caption{A sample of the random homomorphism height function on a $300\times300$ torus with boundary values set to 0 on every second vertex (see \eqref{zero_BC_def}). In this sample, values range from $-5$ to $6$. Sampled using coupling from the past \cite{PW96}.\label{2d_300_fig}}
}
\end{figure}

\begin{figure}[p!]
\centering
{\includegraphics[width=0.95\textwidth,viewport=185 11 860 713,clip]{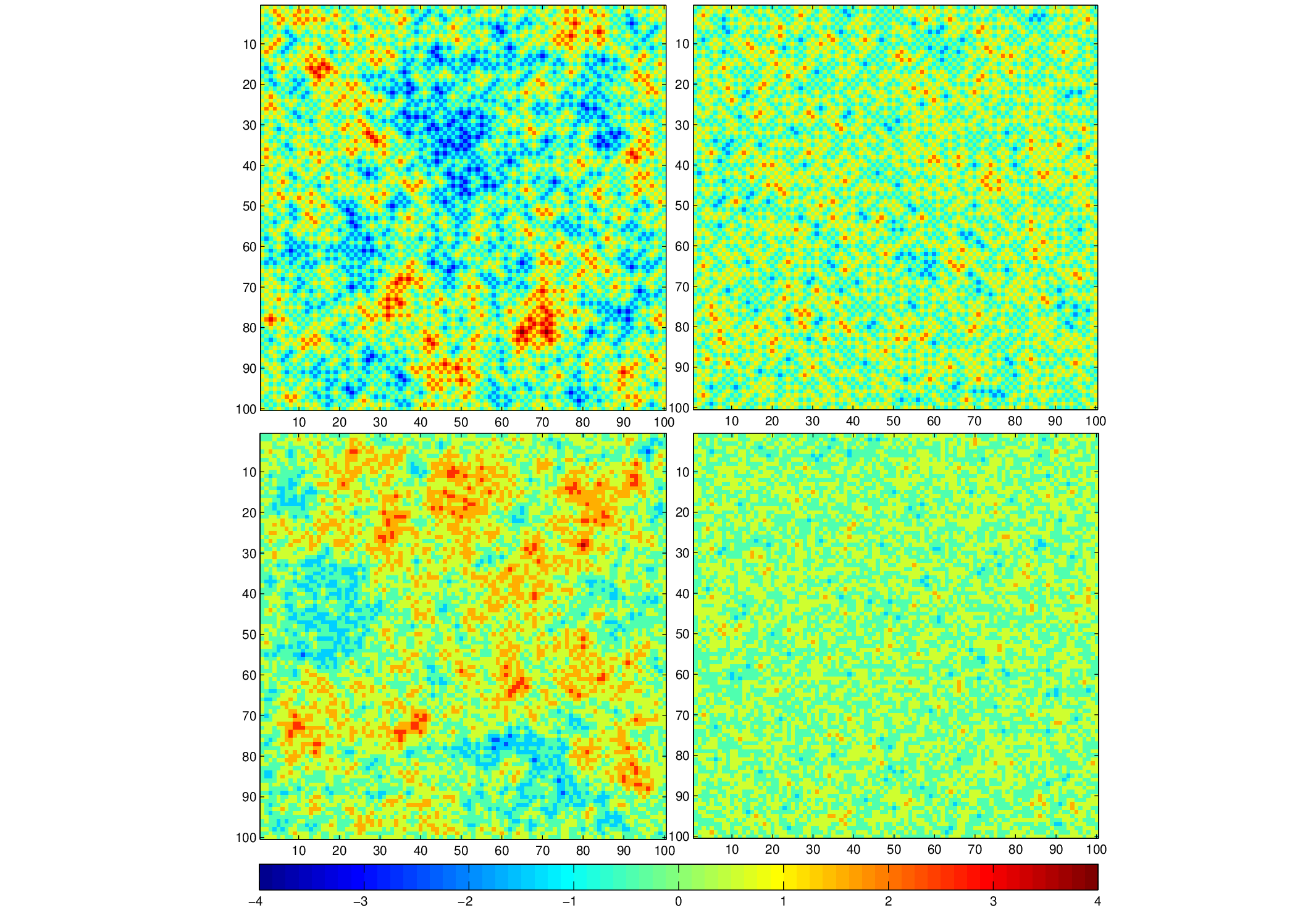}}
{\it
\caption{Top row: samples of the random homomorphism height function on a $100\times100$ torus (top left) and on the middle slice (at height 50) of a $100\times100\times100$ torus (top right), both with boundary values set to 0 on every second vertex (see \eqref{zero_BC_def}). Bottom row: samples of the random Lipschitz height function on a $100\times100$ torus (bottom left) and on the middle slice (at height 50) of a $100\times100\times100$ torus (bottom right), both constrained to have boundary values in the set $\{-\frac{1}{2},\frac{1}{2}\}$ (see zero-one BC in Section~\ref{Lipschitz_functions_section}) so that the values taken are in the set $\Z+\frac{1}{2}$ (the purpose of this shift is to obtain a more symmetric picture).
Sampled using coupling from the past \cite{PW96}.\label{2d_3d_100_fig}}
}
\end{figure}

We concentrate attention on the special case in which $G=\Z_n^d =
(\Z/n\Z)^d$, a cube with side length $n$ in the hyper-cubic lattice
$\Z^d$ with periodic boundary conditions (a torus). In this case,
the above height functions are strongly related to models of
statistical mechanics, e.g., simple random walk, the square ice
model and the uniform 3-coloring model (the anti-ferromagnetic
3-state Potts model at zero temperature). The height functions are
also examples of discrete surface models with nearest neighbor
interactions and it is of interest to compare them with other
surface models of this kind such as the discrete Gaussian free
field, lozenge and domino tilings and solid-on-solid models, see
,e.g., \cite{FS81}, \cite{G88}, \cite{NS97}, \cite{Ke01} and
\cite{S05} for details of these other models. By such comparison,
one may expect that the random height function (for both the
homomorphism and Lipschitz models) in dimension 2 will exhibit some
\emph{roughness}, meaning for example that when $G=\Z_n^2$ with the
one-point BC $(B,\mu)$, the variance of the height at a fixed vertex
$v$ will grow with the distance of $v$ from $B$. In contrast, when
the dimension $d$ is 3 or higher, one may expect that when
$G=\Z_n^d$ with the one-point BC, the random height function will be
\emph{localized}, having variance at each vertex bounded uniformly
in the side length of the torus. Numerical simulations appear to
support these expectations, see Figures~\ref{2d_300_fig} and
\ref{2d_3d_100_fig} for samples of the random height functions on
$\Z_{300}^2$, $\Z_{100}^2$ and $\Z_{100}^3$. However, none of these
predictions has been confirmed rigorously prior to this work. In
this paper we give a proof of the high-dimensional case of the above
predictions, when the dimension $d$ is above a certain threshold
$d_0$. Furthermore, we introduce the graph $G=\Z_n^d\times\Z_2^m$ (a
torus with $d$ sides of length $n$ and $m$ sides of length $2$)
which, for fixed $m$, is just an enhanced version of the torus
$\Z_n^d$, and prove that for a fixed large $m$ and any $d\ge 2$, the
random height function on $G$ is localized. More precisely, letting
$\range(f)$ be the number of values taken by $f$, we have the
following theorem.

\begin{theorem}\label{homogeneous_torus_thm}
There exist $d_0\in\N$, $C_d,c_d>0$ such that the following holds. If
\begin{itemize}
\item $G=\Z_n^d$ for even $n$ and $d\ge d_0$, or
\item $G=\Z_n^d\times\Z_2^{d_0}$ for even $n$ and $d\ge 2$,
\end{itemize}
then for all boundary conditions $(B,\mu)$ with non-positive boundary values, that is $\mu(b)\le 0$ for $b\in B$, if
\begin{itemize}
\item $\Hom(G,B,\mu)\neq \emptyset$ and $f$ is sampled uniformly from $\Hom(G,B,\mu)$, or
\item $\Lip(G,B,\mu)\neq\emptyset$ and $f$ is sampled uniformly from $\Lip(G,B,\mu)$,
\end{itemize}
then
\begin{equation*}
\P(f(x)\ge t)\le \exp\left(-c_d\, t^d\right)\qquad\forall t\ge 3\text{ and }x\in V[G].
\end{equation*}
If, in addition, we have zero boundary values, that is $\mu(b)=0$ for all $b\in B$, then
\begin{equation*}
\P(\range(f)\ge C_d\log^{1/d} n)\le n^{-4d}
\end{equation*}
and if $(B,\mu)$ is the one-point boundary condition then
\begin{equation}\label{range_two_sided_bound}
\P(c_d \log^{1/d} n\le \range(f)\le C_d\log^{1/d} n)\ge 1-n^{-3d}.
\end{equation}
\end{theorem}

Thus the situation resembles that of percolation and the lace
expansion \cite{S04}. One expects the results to hold starting from
a certain low dimension, but the proofs are available either for
large enough dimension, or in any dimension, but for an enhanced
version of the graph (in the case of percolation, the enhanced
version is the spread-out lattice). We remark that the lower bound
on the range in \eqref{range_two_sided_bound} follows from a theorem
of Benjamini, Yadin and Yehudayoff \cite{BYY07} (see
Theorem~\ref{BYY_thm} below) and our upper bound matches it up to
constants. We remark also that Yadin has found a bijection between
the Lipschitz model on a graph $G$ and the homomorphism model on
$G\times\Z_2$ (Theorem~\ref{Yadin_bijection_thm}). Our proof of
Theorem~\ref{homogeneous_torus_thm} uses this bijection by
establishing the theorem first for the homomorphism model and then
deducing the Lipschitz case via the bijection. Thus, although the
requirement that $n$ be even is essential only for the homomorphism
model (to make $G$ bipartite), we require it also for the Lipschitz
model for our proof to apply.

The careful reader may have noticed that while we expect the random height function to be rough in two dimensions, the theorem above states that it is localized for the \emph{enhanced} two-dimensional torus $\Z_n^2\times\Z_2^{d_0}$. Thus, if our expectation is true, the fluctuations of the random height function in two dimensions are quite sensitive to the local features of the graph; small enhancements may change the model from a rough to a localized regime. Analogous phenomenon (in terms of temperature) have been observed in Solid-On-Solid models taking integer values \cite{FS81} and are termed \emph{roughening transitions}. Our work establishes only one side of this transition since we do not show that the random height function in two dimensions is indeed rough, however, we are able to establish the full transition on a class of tori with non-equal side lengths including, for example, the $n\times\lfloor\frac{1}{10}\log n\rfloor$ torus. As a result we refute a conjecture of \cite{BYY07} and are able to answer a question of \cite{BHM00}. In \cite{BYY07} it was conjectured that on any graph $G$, the typical ranges of the random homomorphism and Lipschitz height functions are of the same order of magnitude. In \cite{BHM00} it was asked whether local changes to the graph (in the sense of rough isometries) can affect the typical range of the random height function by more than a constant factor. Thus the transition we establish provides, via the Yadin bijection, a refutation of the conjecture of \cite{BYY07} and an affirmative answer to \cite{BHM00}'s question. More details are provided in Section~\ref{roughening_transition_sec} below.

As mentioned above, the homomorphism model is strongly related to
the uniform 3-coloring model. Let us introduce this model in more
detail and explain how our results apply to it. For a graph $G$,
$\emptyset\neq B\subseteq V[G]$ and $\nu:B\to\{0,1,2\}$, we let
$\col(G,B,\nu)$ be the set of all proper 3-colorings (with colors
$0,1,2$) of $V[G]$ taking the values $\nu$ on $B$. We are interested
in the structure of a uniformly sampled coloring from
$\col(G,B,\nu)$. Suppose now that $f\in\Hom(G,B,\mu)$ for some BC
$(B,\mu)$. We note trivially that the map $f\mapsto (f\bmod 3)$
sends $\Hom(G,B,\mu)$ into $\col(G,B,\mu\bmod 3)$. Specializing to
the case $G=\Z_n^d$, it can be shown that this map becomes a
bijection for certain boundary conditions $(B,\mu)$. In these cases,
our results apply and give an understanding of the structure of the
uniform 3-coloring. We illustrate this here with one example (see
Section~\ref{3_coloring_ice_section} for more details). For
$G=\Z_n^d$, the zero BC are boundary conditions which, in some
coordinate system which turns $\Z_n^d$ into a box, put zero at every
second vertex on the boundary of this box. See \eqref{zero_BC_def}
for a precise definition and Figures~\ref{2d_300_fig} and
\ref{2d_3d_100_fig} for a sample from these boundary conditions. For
this BC, the set $B$ is contained in one of the two bipartite
classes of $G$, we call this class the even sub-lattice and denote
it by $V^\even$. We then find that in high dimensions, a uniformly
sampled 3-coloring with the zero BC will take the color zero on most
of the even sub-lattice, as follows.
\begin{theorem}\label{coloring_intro_thm}
There exist $d_0\in\N$, $c>0$ such that for all $d\ge d_0$, if $G=\Z_n^d$ for even $n$ and $g$ is a uniformly sampled coloring from $\col(G,B,\mu)$ with the zero BC $(B,\mu)$ then
\begin{equation*}
\frac{\E|\{v\in V^\even\ \big|\ g(v)\neq 0\}|}{|V^\even|}\le \exp\left(-\frac{cd}{\log^2 d}\right).
\end{equation*}
\end{theorem}
As in Theorem~\ref{homogeneous_torus_thm}, the theorem also applies to the graph $G=\Z_n^d\times\Z_2^{d_0}$, with appropriate BC, for any $d\ge 2$, sufficiently large $d_0$ and even $n$, see Section~\ref{3_coloring_ice_section} for more details.

One of the main existing results about the homomorphism model is the result of Galvin \cite{G03}, improving an earlier result of Kahn \cite{K01} who proved a conjecture of \cite{BHM00}. Galvin studied the model when $G=\Z_2^d$, the hyper-cube graph, for large dimensions $d$. He proved that with high probability, the random homomorphism height function, with the one-point BC, takes at most 5 values! He furthermore calculated the asymptotic (strictly positive) probabilities for taking exactly 3,4 and 5 values. We cite (the first part of) Galvin's result here.
\begin{theorem}(Galvin \cite{G03})\label{Galvin_thm}
If $G=\Z_2^d$ and $f$ is sampled uniformly from $\Hom(G,B,\mu)$, with the one-point boundary condition $(B,\mu)$, then
\begin{equation*}
\P(\range(f)>5)\le \exp(-\Omega(d))\qquad\text{as $d\to\infty$}.
\end{equation*}
\end{theorem}
Our techniques are flexible enough to provide a significant generalization of Galvin's theorem.
\begin{theorem}\label{homogeneous_tori_finite_range_thm}
For any integer $k\ge 2$, there exist $d_0(k)$ and $c_k>0$ such that the following holds. If $G=\Z_n^d$ with $d\ge d_0$ and even $n\le \exp\left(\frac{c_kd^{k-1}}{\log^2 d}\right)$ then:
\begin{itemize}
\item For $f$ sampled uniformly from $\Hom(G,B,\mu)$, with the one-point boundary condition $(B,\mu)$,
\begin{equation*}
\P(\range(f)>2k+1)\le \exp\left(-\frac{c_kd^k}{\log^2 d}\right).
\end{equation*}
\item For $f$ sampled uniformly from $\Lip(G,B,\mu)$, with the one-point boundary condition $(B,\mu)$,
\begin{equation*}
\P(\range(f)>2k)\le \exp\left(-\frac{c_kd^k}{\log^2 d}\right).
\end{equation*}
\end{itemize}
\end{theorem}
The case $k=2$, $G=\Z_2^d$ and $f\unifin\Hom(G,B,\mu)$ recovers the theorem of Galvin with an improved probability bound. Moreover, the theorem shows that the same phenomenon holds also when $G=\Z_n^d$ with $n\le \exp\left(\frac{c d}{\log^2 d}\right)$ and a similar phenomenon holds with $5$ replaced by $7, 9, 11$, etc., when the torus has larger side-length. Furthermore, we are able to treat random Lipschitz height functions and find that they exhibit even stronger concentration, taking at most $2k$ values with high probability in the situations when random homomorphism height functions take at most $2k+1$ values. Our results are in fact even more general, applying for more general tori and boundary conditions, see Theorems~\ref{finite_range_thm} and \ref{Lipschitz_finite_range_thm} below.

Our understanding of the typical structure of the random height
function in high dimensions extends beyond the understanding of its
height at fixed points and its range. Theorem~\ref{main_thm}, which
lies at the heart of all our other proofs, shows that for a random
homomorphism height function with, say, a one-point BC, the
probability that a level set of length $L$ surrounds a given vertex
is exponentially small in $L$ (see Figures~\ref{level_sets_fig} and
\ref{shift_transform_fig} for illustration of level sets). Thus,
with high probability, the height function will not have any level
sets longer than the logarithm of the size of the graph
(Corollary~\ref{level_set_length_cor}). We believe that the
structure of the typical homomorphism height function is that on
either the even or odd sub-lattice, it takes predominantly one
value. Furthermore, in places where this pattern is broken, an
occurrence which is exponentially rare in the boundary length of the
break-up, the function ``switches phase'' and predominantly takes a
different value on the other sub-lattice. This structure then
continues recursively inside each such break-up. We believe our
results can be used to make this picture precise, but do not pursue
this in this work. Instead, we content ourselves with proving
elements of the full picture such as the above-mentioned
level set theorem and such as showing that under certain boundary conditions, with high probability the function takes predominantly one value on one of the sub-lattices (Corollary~\ref{even_lattice_zero_cor}). We also believe that for certain (sequences of) boundary conditions, the homomorphism model has a thermodynamic limit and we indicate how our theorems may be used to prove this fact, see Section~\ref{thermodynamic_limit_sec} below.

We expect a similar typical structure for random Lipschitz height
functions. Indeed, this will follow from the Yadin bijection (see
Section~\ref{Lipschitz_functions_section}) by establishing the
typical structure of homomorphism height functions described above.
We expect that for, say, the one-point BC, the function takes
predominantly two consecutive values everywhere, where again, in
places where this pattern is broken, an occurrence which is
exponentially rare in the boundary length of the break-up, the
function switches to take predominantly two different consecutive
values and the structure repeats inside.

As indicated above, our theorems require the understanding of the random height functions (homomorphism or Lipschitz) on tori of varying side lengths such as $\Z_n^d$, $\Z_n^d\times\Z_2^m$ and $\Z_n\times\Z_{\lfloor\frac{1}{10}\log n\rfloor}\times\Z_2^m$. To be able to deal with all these cases under a unified framework, we shall consider in the sequel tori with general side lengths: $n_d\ge n_{d-1}\ge \ldots\ge n_1\ge 2$. However, as one may expect, the above picture, in which the random height function is localized, does not hold for all choices of side-lengths,
even when $d$ is large. For example, if $n_d=n$ and $n_i=2$ for all $1\le i\le d-1$, the torus is essentially one-dimensional and for large enough $n$, a random height function on it (with the one-point BC) will resemble a simple random walk bridge. We distinguish two cases: when $n_d\le \exp\left(c_d\prod_{i=1}^{d-1} n_i\right)$ and when $n_d\ge \exp\left(C_d\prod_{i=1}^{d-1} n_i\right)$ for some specific $C_d, c_d>0$ (see \eqref{non_linearity_cond} and \eqref{linear_cond} below), which we term a \emph{non-linear torus} and \emph{linear torus} respectively. We are then able to show that on non-linear tori in high dimensions (with, say, the one-point BC), the random height function is localized, having essentially the same features described above for $\Z_n^d$ in high dimensions, whereas on linear tori in all dimensions (with the one-point BC), the random height function is rough, resembling a simple random walk bridge. The results presented above are, perhaps, the most interesting special cases of these general results.

The main tool in our proofs is the analysis of a special class of
cutsets which we term \emph{odd cutsets}. These are minimal edge
cutsets on the torus which have all their interior vertex boundary
on the odd sub-lattice (see Section~\ref{preliminaries_sec} and
Definition \eqref{OMCut_def} for precise definitions). The cutsets
appear naturally in our model as the level sets of homomorphism
height functions (see Figures~\ref{level_sets_fig} and
\ref{shift_transform_fig} for examples). We find that such cutsets
have many special properties not shared by standard minimal cutsets
(see Sections~\ref{preliminaries_sec} and
\ref{odd_cutsets_structure_sec}) and believe that they may be of use
in the analysis of other models as well. Our main structure theorems
for odd cutsets, Theorems~\ref{count_cutsets_thm} and
\ref{interior_approximation_theorem}, provide information on the
regularity of their boundaries and on a certain way of approximating
them. Understanding such cutsets better and, in particular,
improving the bounds of these theorems (see also the open questions
in Section~\ref{open_questions_sec}) is the main ``bottleneck'' in
reducing the minimal dimension $d_0$ above which our theorems apply.

We end the introduction with some historical comments.
Theorem~\ref{coloring_intro_thm} on the existence of multiple Gibbs
states for the 3-coloring model on $\Z^d$ was conjectured by
Koteck\'y circa 1985, although the explicit conjecture seems not to
have appeared in print (see, e.g., \cite{K85} for context and
\cite{GKRS12} for additional details). The conjecture was made in
the stronger form that there are 6 distinct Gibbs states with
\emph{maximal entropy}, specified by a predominance of one color on
one of the sublattices. Theorem~\ref{coloring_intro_thm} was first
announced by Galvin, Kahn, Randall and Sorkin at a 2002 Newton
Institute programme and was later discussed by Kahn in talks and
communications with Koteck\'y and others. The theorem was first
mentioned in print in the 2007 work of Galvin and Randall
\cite{GR07}. The present author who was unaware of these
developments until late into his work published the present work on
the arXiv in May 2010. Galvin, Kahn, Randall and Sorkin published
their work on the arXiv in October 2012 \cite{GKRS12}, in which they
establish a version of Theorem~\ref{coloring_intro_thm}, showed in addition that the resulting Gibbs
states have maximal entropy and proved torpid mixing results for
related dynamics. Though similar in spirit, the approach of
\cite{GKRS12} is different from the present argument in that it
stays within the world of colorings and does not exploit the
connection with height functions. Finally, we remark that the ideas
of using cutsets with the ``odd'' property and approximating them
have been used in several previous works, e.g., in~\cite{BC+99, G08,
GK04, S87}.

\section{Results and discussion}
We begin this section with several definitions which are required for the statement of our main theorems. We then state our main theorems for homomorphism and Lipschitz height functions, together with a discussion of the above-mentioned roughening transition, the relation of the homomorphism model with proper 3-colorings and square ice, and the thermodynamic limit for the homomorphism model. We conclude this section with proof sketches for our main theorems, a reader's guide and acknowledgments.
\subsection{Definitions}\label{def_section}
For integer $n\ge 2$, let $\Z_n$ be the $n$-cycle graph. In our
convention, $\Z_n$ is a simple graph with vertices $\{0,1,\ldots,
n-1\}$ such that $i$ is adjacent to $i+1$ and $i-1$ modulo $n$. For
\emph{even} integers
\begin{equation}\label{torus_side_lengths}
n_d\ge n_{d-1}\ge \cdots \ge n_1\ge 2
\end{equation}
we let $G:=\Z_{n_1}\times\cdots\times \Z_{n_d}$ be the
$d$-dimensional torus with side lengths $n_1,\ldots, n_d$ (our
$\times$ refers to the Cartesian product of graphs, also denoted
$\square$ in certain literature). Henceforth, a torus will always
refer to a graph $G$ as above (and in particular, we will always
assume that the $n_i$ are even).
When needed, we shall assume a bipartition $(V^{\even}, V^{\odd})$ is chosen on $G$ and a natural coordinate system placed on it, using its product structure, so that
\begin{equation}
V[G]=\{(x_1,\ldots, x_d)\ |\ 0\le x_i\le n_i-1\text{ for $1\le i\le d$}\}\label{G_coordinates}.
\end{equation}

For an integer $r\ge 0$, we define the volume of a ball of radius $r$ by
\begin{equation*}
  \vol(r):=\left|\left\{w\in V[G]\ \big|\ d_G(v,w)\le r\right\}\right|
\end{equation*}
where $d_G$ is the graph distance in $G$ ($\vol(r)$ does not depend on the choice of $v\in V[G]$).

As explained in the introduction, we distinguish two types of tori. We call $G$ \emph{non-linear} if
\begin{equation}\label{non_linearity_cond}
n_d\le \exp\left(\frac{1}{d\log^3 d}\prod_{i=1}^{d-1} n_i\right)
\end{equation}
and we call it \emph{$\lambda$-linear} for some $\lambda>0$ if
\begin{equation}\label{linear_cond}
n_d\ge \exp\left(\frac{1}{\lambda}\prod_{i=1}^{d-1} n_i\right).
\end{equation}

Recalling that a pair $(B,\mu)$ with $\emptyset\neq B\subseteq V[G]$ and $\mu:B\to\Z$ is called a boundary condition, we say that $\mu$ is \emph{non-positive} if $\mu(b)\le 0$ for all $b\in B$ and $\mu$ is \emph{zero} if $\mu(b)=0$ for all $b\in B$. We call $(B,\mu)$ a \emph{legal (homomorphism) boundary condition} if $\Hom(G,B,\mu)\neq\emptyset$ and $\mu$ takes even values on $V^\even$ and odd values on $V^\odd$. We call it a \emph{legal Lipschitz boundary condition} if $\Lip(G,B,\mu)\neq\emptyset$.

We remark that our theorems below apply also to a slightly weaker definition of non-linear torus, when the $\frac{1}{\log^3 d}$ is replaced by $\frac{c}{\log^2 d}$ for a small enough $c>0$. Definition \eqref{non_linearity_cond} was chosen to simplify some of the notation. However, we note that in this definition and all our theorems below where a power of $\log d$ appears, it may well be the case that these $\log$ factors are an artifact of our proof and the theorems remain true without them.

\subsection{Homomorphism Height Functions}\label{hom_main_results_sec}
In this section we concentrate our attention on the homomorphism
height function model and its properties. The results will then be
extended to the Lipschitz height function model via the Yadin
bijection in Section~\ref{Lipschitz_functions_section}.

\subsubsection{Height and Range}\label{height_and_range_thm_sec}

We say that a set $B\subseteq V[G]$ has \emph{full projection} if,
in the coordinate system \eqref{G_coordinates}, there exists $1\le
i_0\le d$ such that every line of the form $\{(x_1,\ldots, x_d)\ |\
0\le x_{i_0}\le n_{i_0}-1\}$, for fixed $x_1,\ldots,
x_{i_0-1},x_{i_0+1},\ldots, x_d$, intersects $B$. Our next theorem
shows that on non-linear tori in high dimensions, the height of a
uniform homomorphism at a fixed vertex has very light tails.
\begin{theorem}\label{height_at_point_thm}
There exist $d_0\in\N$, $c>0$ such that for all $d\ge d_0$, non-linear tori $G$, legal boundary conditions $(B,\mu)$ with non-positive $\mu$ and $x\in V[G]$, if $f\unifin\Hom(G,B,\mu)$ then
\begin{equation*}
\P(f(x)\ge t)\le \exp\left(-\frac{c\vol(\lceil t/2\rceil-1)}{\min(t,d)\log^2 d}\right)\qquad\text{for all $t\ge 3$}.
\end{equation*}
Furthermore, if $t\ge 3$ satisfies $\vol(\lceil t/2\rceil-1)\le\frac{1}{3}n_d$ then
\begin{equation*}
\P(f(x)\ge t)\le \exp\left(-\frac{c\vol(\lceil t/2\rceil-1)}{\log^2 d}\right).
\end{equation*}
Finally, if $B$ has full projection then
\begin{equation*}
\P(f(x)\ge t)\le \exp\left(-\frac{c\vol(t-1)}{\log^2 d}\right)\qquad\text{for all $t\ge 2$}.
\end{equation*}
\end{theorem}

As an immediate corollary of the third part of the theorem, we obtain that if our boundary condition has full projection and zero $\mu$, then the random height function is zero on most of the even sub-lattice (see also Section~\ref{thermodynamic_limit_sec}).
\begin{corollary}\label{even_lattice_zero_cor}
Under the assumptions of Theorem~\ref{height_at_point_thm}, there exists $c>0$ such that if $B$ has full projection and $\mu$ is zero then
\begin{equation*}
\frac{\E|\{v\in V^\even\ |\ f(v)\neq 0\}|}{|V^\even|}\le \exp\left(-\frac{c d}{\log^2 d}\right).
\end{equation*}
\end{corollary}

A particularly important example of a full projection BC with zero $\mu$ is the \emph{zero BC}:
\begin{equation}\label{zero_BC_def}
B:=\{(x_1,\ldots, x_d)\in V^\even\ |\ \exists i\text{ s.t. $x_i\in\{0,n_i-1\}$}\},\quad \mu(b):=0\text{ for all }b\in B.
\end{equation}
Uniformly sampled homomorphisms with this boundary condition on $\Z_{300}^2$, $\Z_{100}^2$ and on $\Z_{100}^3$ are depicted in Figures~\ref{2d_300_fig} and \ref{2d_3d_100_fig} (only a slice of the torus is depicted in the 3-dimensional case) and suggest that the corollary holds in dimension 3 and fails in dimension 2.

We proceed to analyze the range of the uniform homomorphism on high-dimensional non-linear tori.

\begin{theorem}\label{range_thm}
There exist $d_0\in\N$, $C,c>0$ such that for all $d\ge d_0$, non-linear tori $G$ and legal boundary conditions $(B,\mu)$ with zero $\mu$, if we set
\begin{equation*}
k:=\min\left\{ m\in\N\ \big|\ \vol(m)\ge C\log^2
d\cdot\log|V[G]|\right\}
\end{equation*}
and let $f\unifin\Hom(G,B,\mu)$, then
\begin{equation*}
\P(\range(f)>2k+1)\le \exp\left(-\frac{c\vol(k)}{\log^2 d}\right)\le \frac{1}{|V[G]|^4}.
\end{equation*}
\end{theorem}

We remark that the theorem remains true if we change the power of $|V[G]|$ in the probability bound to any larger power; the current statement was chosen for simplicity.
We note also that the conclusion of the theorem implies $\E\range(f)\le 4k$, say, since $\range(f)$ is deterministically bounded by $|V[G]|$.

A result of an opposite nature was obtained in \cite{BYY07}. The result there is for an arbitrary graph $G$ and we present below a version of it specialized to tori (this is the line before last in the proof of Theorem 2.1 there).
\begin{theorem}(Benjamini, Yadin, Yehudayoff \cite{BYY07})\label{BYY_thm}
For a torus $G$, if $f\unifin\Hom(G,B,\mu)$ with a one-point BC $(B,\mu)$ and if $r\ge0$ is an integer for which $\vol(r)\le \eps\log_2 |V[G]|$ for some $0<\eps<1$ then
\begin{equation*}
\P(\range(f)\le r)\le e^2\exp\left(-\frac{|V[G]|^{1-\eps}}{\eps^2\log_2^2|V[G]|}\right).
\end{equation*}
\end{theorem}

Comparing Theorems~\ref{range_thm} and \ref{BYY_thm} we see that the bound on the range given by the former is of the right order of magnitude for one-point BC.
\begin{corollary}\label{range_sharpness_cor}
There exist $d_0\in\N$, $C_d,c_d>0$ such that for all $d\ge d_0$, non-linear tori $G$ and the one-point BC $(B,\mu)$, if $f\unifin\Hom(G,B,\mu)$ then
\begin{equation*}
\P(c_d r \le \range(f) \le C_d r)\ge 1-\frac{1}{|V[G]|^3},
\end{equation*}
where $r:=\min\left\{ m\in\N\ \big|\ \vol(m)\ge \log|V[G]|\right\}$.
\end{corollary}

As noted in the introduction, our techniques are sufficiently flexible to recover and extend the result of Galvin, Theorem~\ref{Galvin_thm} above. For homomorphism height functions, we have the following result.

\begin{theorem}\label{finite_range_thm}
For any integer $k\ge 2$, there exist $d_0(k)$ and $c_k>0$ such that for all $d\ge d_0(k)$, non-linear tori $G$ and legal boundary conditions $(B,\mu)$ with zero $\mu$, if $|V[G]|\le \exp\left(\frac{c_kd^{k}}{\log^2 d}\right)$ and $f\unifin\Hom(G,B,\mu)$ then
\begin{equation*}
\P(\range(f)>2k+1)\le \exp\left(-\frac{c_kd^k}{\log^2 d}\right).
\end{equation*}
\end{theorem}

The case $k=2$, $G=\Z_2^d$ and the one-point BC recovers the theorem of Galvin with an improved probability bound and shows that when $G$ is the hyper-cube, the range is at most 5 with high probability as $d\to\infty$. Moreover, the theorem shows that the same phenomenon holds for any boundary condition with zero $\mu$ and in any non-linear torus in which the side lengths are at most $2^{d^{1-\eps}}$, say, for some fixed $\eps>0$.
Furthermore, a similar phenomenon holds with $5$ replaced by $7, 9, 11$, etc.

The results presented above show that on non-linear tori in high dimensions, the random homomorphism height function is very localized. In contrast, the following theorem shows that for linear tori, the situation is drastically different and the fluctuations of the random height function resemble more those of a simple random walk - the one-dimensional case.

\begin{theorem}\label{linear_torus_thm}
For all $0<\lambda<\frac{1}{2\log 2}$, there exist $\alpha=\alpha(\lambda)>0$ and $C=C(\lambda)>0$ such that for all dimensions $d\ge 2$ and all $\lambda$-linear tori $G$, if $f\unifin\Hom(G,B,\mu)$ with the one-point BC $(B,\mu)$ then
\begin{equation}\label{linear_tori_thm_bound}
\P(\range(f)\le |V[G]|^\alpha)\le \frac{C}{|V[G]|^{\alpha}}.
\end{equation}
\end{theorem}

As a final remark to this section, we note that not all possible tori fall under our definitions of non-linear and linear tori. The remaining cases are left as open questions, see Section~\ref{open_questions_sec}.

\subsubsection{Level sets}\label{level_set_structure_sec}

Our understanding of the height and range of the random homomorphism height function on a non-linear torus stems from a detailed analysis of the level sets of such functions. To explain this further, we introduce a few more definitions.
Fixing a legal boundary condition $(B,\mu)$ for a non-positive $\mu$, $x\in V[G]$ and $f\in\Hom(G,B,\mu)$, we denote by $A$ the union of those connected components of $\{v\in V[G]\ |\ f(v)\le 0\}$ which contain points of $B$, and by $A^c_x$ the connected component of $x$ in $V[G]\setminus A$ (defined to be empty if $x\in A$). We then define
\begin{equation*}
\LS(f,x,B):=\begin{cases}\text{set of all edges between $A$ and $A^c_x$}&x\notin A\\\emptyset&x\in A\end{cases}.
\end{equation*}
$\LS(f,x,B)$ is the outermost height 1 level set of $f$ around $x$ when coming from $B$. The level sets (around all vertices $x$) are depicted in Figure~\ref{level_sets_fig} for height functions on two-dimensional tori. For an integer $L\ge 1$, we let $\Omega_{x,L}$ (implicitly $\Omega_{x,L,B,\mu}$) be the set of $f\in \Hom(G,B,\mu)$ for which $|\LS(f,x,B)|=L$. Similarly, for $x_1,\ldots, x_k\in V[G]$ and integers $L_1,\ldots, L_k\ge 1$ we let $\Omega_{(x_1,\ldots, x_k),(L_1,\ldots, L_k)}$ be the set of $f\in\cap_{i=1}^k \Omega_{x_i,L_i}$ satisfying that $\LS(f,x_i,B)\cap \LS(f,x_j,B)=\emptyset$ for all $i\neq j$ (one can show that these level sets are either identical or disjoint). The following theorem is at the heart of our analysis of random homomorphism height functions.
\begin{theorem}\label{main_thm}
There exist $d_0\in\N$, $c>0$ such that for all $d\ge d_0$,
$k\in\N$, non-linear tori $G$, legal boundary conditions $(B,\mu)$
with non-positive $\mu$, vertices $x_1,\ldots,x_k\in V[G]$ and
integers $L_1,\ldots, L_k\ge 1$ we have that if
$f\unifin\Hom(G,B,\mu)$ then
\begin{equation*}
\P(f\in\Omega_{(x_1,\ldots, x_k),(L_1,\ldots, L_k)})\le d^k\exp\left(-\frac{c\sum_{i=1}^k L_i}{d\log^2 d}\right).
\end{equation*}
\end{theorem}

\begin{figure}[t!]
\centering
{\includegraphics[width=\textwidth,viewport=30 70 990 575,clip]{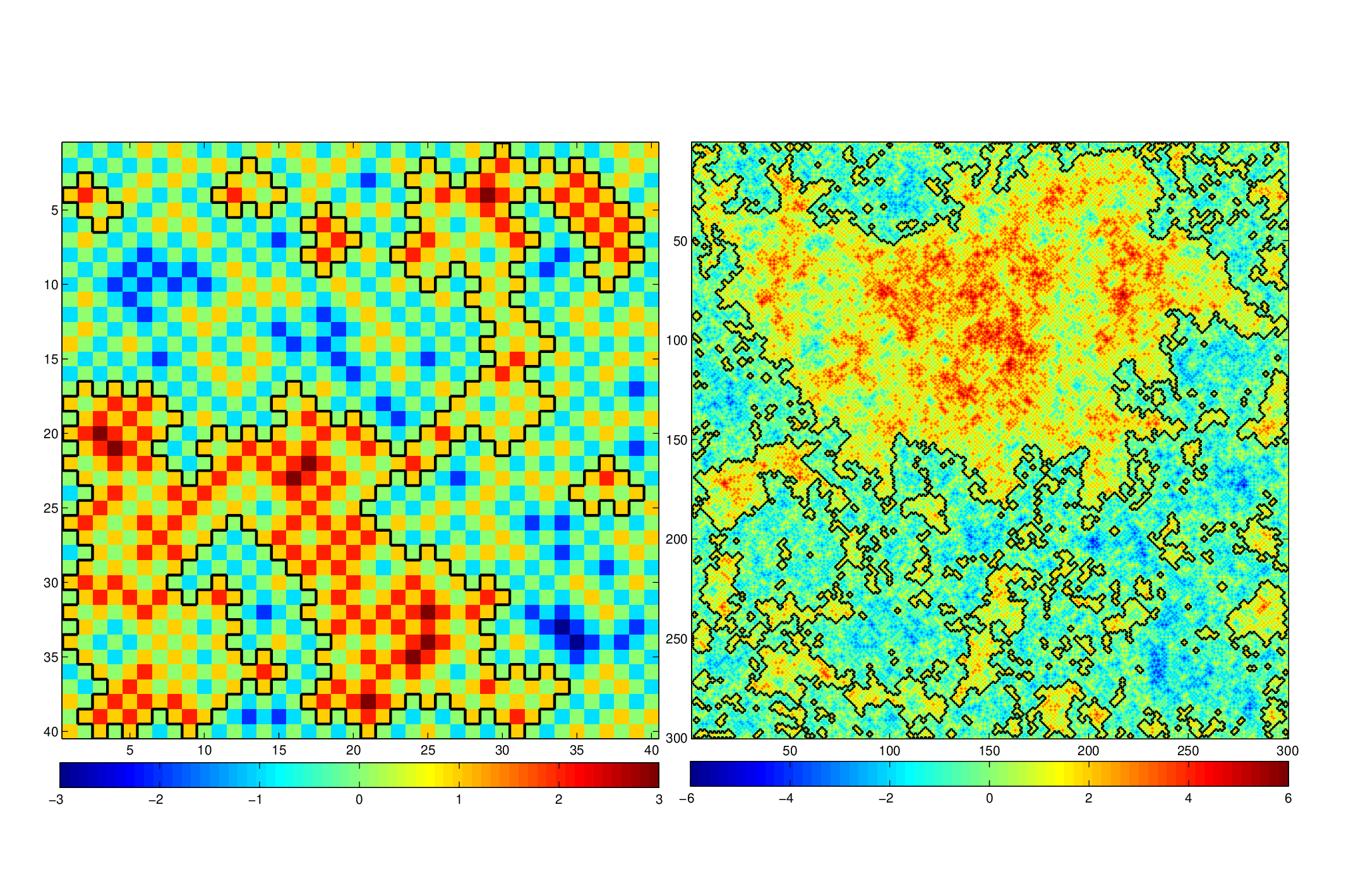}}
{\it
\caption{The outermost height 1 level sets of two samples of homomorphism height functions, the left on a 40 x 40 torus and the right on a 300 x 300 torus, both with zero boundary conditions (dual edges to the level sets are marked in black). Trivial level sets - those surrounding a single vertex - have been removed to obtain a less cluttered picture. Unlike these pictures, it is expected that in 3 dimensions and higher, the length of the longest level set is only logarithmic in the side of the torus. This is proven in sufficiently high dimensions in Corollary~\ref{level_set_length_cor}. Picture produced by Steven M. Heilman. \label{level_sets_fig}}
}
\end{figure}

This theorem is used in Section~\ref{isoperim_height_range_Lip_sec} below to prove the height and range theorems of Section~\ref{height_and_range_thm_sec}. The underlying idea is that we may define, in an analogous way to $\LS(f,x,B)$, also the outermost height $i$ level set of $f$ around a point. Then one can apply the above theorem inductively and conclude that the chance that surrounding a given point, for each $i$, the outermost level set of height $i$ has length $L_i$ is exponentially small in the \emph{sum} of these $L_i$'s. Thus, one may conclude that it is very unlikely that $f$ is large at any given point. See the proof sketches in Section~\ref{proof_sketch_readers_guide_ack_sec} for more details.

As a corollary, we obtain that the largest level set of a random homomorphism height function is at most logarithmic in size with high probability.
\begin{corollary}\label{level_set_length_cor}
Under the assumptions of Theorem~\ref{main_thm}, there exists $C>0$ such that
\begin{equation*}
\P\left(\max_{x\in V[G]} |\LS(f,x,B)|>Cd\log^2 d\cdot
\log|V[G]|\right)\le\frac{1}{|V[G]|^4}.
\end{equation*}
\end{corollary}
The corollary follows directly from Theorem~\ref{main_thm} by a union bound. Figure~\ref{level_sets_fig} presents some evidence that the corollary is false on $\Z_n^2$, but we expect it to hold on $\Z_n^d$ for all $d\ge 3$, as Figure~\ref{2d_3d_100_fig} suggests.

\subsubsection{Roughening transition}\label{roughening_transition_sec}

As explained in the introduction, we expect the random homomorphism height function on $\Z_n^2$ to be rough. Indeed, as is the case for some similarly defined models (e.g., the height function of the dimer model, see \cite{Ke01}), we expect that if $f\unifin\Hom(\Z_n^2,B,\mu)$ for a one-point BC $(B,\mu)$, then $f$ converges weakly to the Gaussian free field, and has $\var(f(v))=\Theta(\log n)$ for generic vertices $v$ and $\E(\range(f))=\Theta(\log n)$ as $n\to\infty$. In contrast, if we take $f\unifin\Hom(\Z_n^2\times\Z_2^{d_0},B,\mu)$ for some large, but fixed, $d_0$ and the one-point BC, then Theorem~\ref{homogeneous_torus_thm} implies that $\var(f(v))=O(1)$ and $\E(\range(f))=\Theta(\sqrt{\log n})$ as $n\to\infty$. Thus we expect a transition in the roughness of the random height function on the graphs $\Z_n^2\times\Z_2^{m}$ as $m$ increases from $0$ to the fixed value $d_0$. Analogous transitions (in terms of temperature) have been observed in Solid-On-Solid models taking integer values (see \cite{FS81}) and are termed \emph{roughening transitions}. We emphasize that we view the passage from the graph $\Z_n^2$ to the graph $\Z_n^2\times\Z_2^{d_0}$ as being a finite enhancement, replacing each vertex of the graph $\Z_n^2$ by a fixed-dimensional hypercube, which leaves the graph essentially two-dimensional as $n$ grows. Analogous enhancements have been used in the study of the mean-field behavior of statistical physics models where one considers the spread-out lattice \cite{S04}, the lattice $\Z^d$ with added long-range connections up to a fixed distance.

Our work establishes only one side of the above transition as we do not show that the random homomorphism height function on $\Z_n^2$ is indeed rough, however, we are able to establish the full transition on a class of tori with non-equal side lengths. Indeed, we may take as our starting point any sequence of $\lambda$-linear tori $G_n$ (see \eqref{linear_cond}) for $\lambda<\frac{1}{2\log 2}$ and side lengths satisfying $n_d=n$ and $\prod_{i=1}^{d-1} n_i\ge c\log n$ for some $c>0$. As a concrete example, one may take $G_n=\Z_n\times\Z_{\lfloor\frac{1}{10}\log n\rfloor}$.
We then note that by Theorem~\ref{linear_torus_thm}, we have some $\alpha,C>0$ such that if $f_1\unifin\Hom(G_n,B,\mu)$ for the one-point BC $(B,\mu)$ then
\begin{equation}\label{f_1_range}
\P(\range(f_1)\le n^\alpha)\le Cn^{-\alpha}\qquad\text{for all $n$.}
\end{equation}
We now let $G_{n,m}:=G_n\times \Z_2^m$ for $m\ge 0$ (so that $G_{n,0}=G_n$) and observe that for some large $m_0$, fixed and independent of $n$, we have that $G_{n,m}$ is non-linear (see \eqref{non_linearity_cond}) for all $m\ge m_0$. Thus, fixing a sufficiently large $m\ge m_0$, still independent of $n$, we may apply Corollary~\ref{range_sharpness_cor} to $f_2\unifin\Hom(G_{n,m},B,\mu)$ with the one-point BC $(B,\mu)$ and obtain
\begin{equation}\label{f_2_range}
\P(c_m \sqrt{\log n} \le \range(f_2) \le C_m \sqrt{\log n})\ge 1-\frac{1}{(2^m n)^3}\qquad\text{for all $n$.}
\end{equation}
Putting together \eqref{f_1_range} and \eqref{f_2_range} we obtain
\begin{equation}\label{poly_range_ratio}
\frac{\E\range(f_1)}{\E\range(f_2)}\ge c n^\beta\qquad\text{for all $n$}
\end{equation}
and some $\beta,c>0$. We call this transition a roughening transition.

We conclude this section by observing that the roughening transition just described
answers a question posed in \cite{BHM00} and refutes a conjecture of \cite{BYY07}. We first define the concept of rough-isometry (or quasi-isometry) of graphs. We say that two graphs $H$ and $H'$ are rough-isometric with constant $C>0$ if there exists $T:V[H]\to V[H']$ such that
\begin{equation}\label{rough_isometry_def}
\frac{1}{C}d_H(v,w)-C\le d_{H'}(T(v),T(w))\le Cd_H(v,w)+C
\end{equation}
for every two vertices $v,w\in V[H]$ and $d_{H'}(v,T(V[H]))\le C$ for every $v\in V[H']$. It was asked in \cite{BHM00} whether there exists a pair of sequences $H_n$ and $H_n'$ of (finite, connected and bipartite) graphs such that $H_n$ is rough-isometric to $H_n'$ with some constant $C>0$, independent of $n$, and $\frac{\E(\range(f_1))}{\E(\range(f_2))}\to\infty$ as $n\to\infty$, where $f_1\unifin\Hom(H_n,B,\mu)$ and $f_2\unifin\Hom(H_n',B,\mu)$ for the one-point BC. Noting that for $G_{n,m}$ defined above, $G_n$ is rough-isometric to $G_{n,m}$ with some constant $C_m$, we may fix an $m$ for which \eqref{poly_range_ratio} holds and obtain an affirmative answer to the question of \cite{BHM00} with a polynomial (in the size of the graphs) rate of convergence to infinity.

Lastly, in \cite{BYY07} it was conjectured that for any sequence of
(finite, connected and bipartite) graphs $H_n$ having maximal degree
$C$ (independent of $n$) and $|V[H_n]|\to\infty$, we have
$\frac{\E(\range(f_1))}{\E(\range(f_2))}=\Theta(1)$ where now
$f_1\unifin\Hom(H_n,B,\mu)$ and $f_2\unifin\Lip(H_n,B,\mu)$, both
with a one-point BC, where the $\Theta(1)$ may depend on $C$. We
note that \eqref{poly_range_ratio} implies that for some $m_1$
(independent of $n$), if we take $f_1\unifin\Hom(G_{n,m_1},B,\mu)$
and $g_1\unifin\Hom(G_{n,m_1+1},B,\mu)$, both with a one-point BC,
we have
\begin{equation}\label{poly_range_ratio_m_1}
\frac{\E\range(f_1)}{\E\range(g_1)}\ge c n^\gamma\qquad\text{for
infinitely many $n$}
\end{equation}
and some $\gamma,c>0$. Here, we need to restrict to infinitely many
$n$ since \eqref{poly_range_ratio} does not guarantee that the
change of behavior between the ranges of $G_{n, k}$ and $G_{n, k+1}$
occurs at the same $k$ for all $n$, only that such a $k$ exists and
is at most some $m$ which is independent of $n$. The Yadin bijection
implies (see Section~\ref{Lipschitz_functions_section} and
Corollary~\ref{Lip_one_point_cor} below) that if we define
$f_2\unifin\Lip(G_{n,m_1},B,\mu)$, with a one-point BC, then
$\E\range(f_2)=\E\range(g_1)-1$. Thus \eqref{poly_range_ratio_m_1}
shows that a subsequence of $H_n:=G_{n,m_1}$ refutes the conjecture,
giving a polynomially large (in $|V[H_n]|$) ratio between the
expected ranges. We remark that it may still be true that this ratio
of expected ranges is uniformly bounded \emph{below} for every
sequence of graphs $H_n$, as in the conjecture.

\subsubsection{Relation to the 3-coloring and square ice models}\label{3_coloring_ice_section}
For a graph $G$, $\emptyset\neq B\subseteq V[G]$ and $\nu:B\to\{0,1,2\}$, let $\col(G,B,\nu)$ be the set of all proper 3-colorings (with colors $0,1,2$) taking the values $\nu$ on $B$. Suppose now that $f\in\Hom(G,B,\mu)$ for some BC $(B,\mu)$. We note trivially that
\begin{equation}\label{mod3_map}
f\mapsto f\bmod 3
\end{equation}
sends $\Hom(G,B,\mu)$ into $\col(G,B,\mu\bmod 3)$. The situation becomes more interesting when $G$ is a box in $\Z^d$ (with non-periodic boundary). I.e., letting $P_n$ be the path graph on $n$ vertices, $G=P_{n_1}\times\cdots\times P_{n_d}$ for some $(n_i)\subseteq\N$. In this case, one may check that the above mapping is in fact a bijection between $\Hom(G,B,\mu)$ and $\col(G,B,\mu\bmod 3)$ for the one-point BC $(B,\mu)$ (In \cite{G03} this observation, for $G=\Z_2^d$, is attributed to Randall, but it may well go back farther). From this fact, it follows directly that for general BC $(B,\mu)$, $\Hom(G,B,\mu)$ is in bijection with $\col(G,B,\mu\bmod 3)$ by \eqref{mod3_map} if and only if
\begin{equation}\label{3col_graph_condition}
\begin{split}
&\text{For any $\mu':B\to\Z$ satisfying $\mu-\mu'\equiv 0\bmod3$ we either}\\
&\text{have $\mu-\mu'$ constant or $\Hom(G,B,\mu')=\emptyset$.}
\end{split}
\end{equation}

In our theorems, however, the graph $G$ is always a torus (that is, with periodic boundary) with even side-lengths. If it were the case that $\Hom(G,B,\mu)$ was in bijection with $\col(G,B,\mu)$ via \eqref{mod3_map} then we could apply our theorems to obtain information about a uniformly sampled coloring in $\col(G,B,\mu)$. However, even in very simple examples this may fail. Indeed, taking $G=\Z_6$ with the one-point BC $(B,\mu)$, the coloring $(0,1,2,0,1,2)$ does not correspond to any function in $\Hom(G,B,\mu)$ via \eqref{mod3_map}. We do not attempt here to find conditions under which \eqref{mod3_map} is a bijection and instead give just one example. Letting $G'$ be the box in $\Z^d$ with the same dimensions as $G$, we note that for the the zero BC $(B,\mu)$ (defined in \eqref{zero_BC_def}) we have that $\Hom(G,B,\mu)=\Hom(G',B,\mu)$, $\col(G,B,\mu)=\col(G',B,\mu)$ and condition \eqref{3col_graph_condition} holds for $\Hom(G',B,\mu)$. Thus, the map \eqref{mod3_map} is a bijection of $\Hom(G,B,\mu)$ and $\col(G,B,\mu)$ for the zero BC $(B,\mu)$. As one application of this fact, we deduce from Corollary \ref{even_lattice_zero_cor} that under some conditions, a uniformly chosen 3-coloring takes the same color on most of the even sub-lattice. The following theorem makes this statement precise.
\begin{theorem}\label{coloring_thm}
There exist $d_0\in\N$, $c>0$ such that for all $d\ge d_0$ and non-linear tori $G$, if $g$ is a uniformly sampled coloring from $\col(G,B,\mu)$ with the zero BC (defined in \eqref{zero_BC_def}) then
\begin{equation*}
\frac{\E|\{v\in V^\even\ \big|\ g(v)\neq 0\}|}{|V^\even|}\le \exp\left(-\frac{cd}{\log^2 d}\right).
\end{equation*}
\end{theorem}
We remark that the above theorem is meaningless for a torus for which one of the side-lengths is 2, since then the zero BC already assigns the value 0 to all vertices in $V^\even$ of such a torus. However, one can check simply that one may modify the zero BC to exclude those dimensions for which the side-length is $2$ and still deduce from the above discussion that $\Hom(G,B,\mu)$ and $\col(G,B,\mu)$ are in bijection and thus the above theorem holds. Explicitly, this modified BC will be $(B,\mu)$ with $B:=\{(x_1,\ldots, x_d)\in V^\even\ |\ \exists i\text{ s.t. $n_i\neq 2$ and $x_i\in\{0,n_i-1\}$}\}$ and $\mu\equiv 0$.

We also discuss briefly the square ice model. Square ice, also called the 6-vertex model, is a model defined on a 2-dimensional torus (or on a square in $\Z^2$ with some boundary conditions). A configuration of square ice is a choice of orientation for each edge satisfying that every vertex has exactly 2 incoming edges and 2 outgoing edges (so that each vertex is in one of 6 states). One then assigns weights to each of the 6 states and samples a configuration from a Gibbs measure with the assigned weights, see e.g. \cite{B82} for details. In particular, if all the weights are equal, one samples a configuration uniformly at random. We call this case \emph{uniform square ice}. It is well known that square ice configurations are in bijection with proper 3-colorings of the underlying torus (where in the bijection, one colors the dual torus). Hence, for certain boundary conditions, they correspond to homomorphism height functions by the bijection described above. Unfortunately, our work does not apply to the most interesting case of the $n\times n$ torus and hence does not shed further light on uniform square ice on it.

\subsubsection{Thermodynamic Limit}\label{thermodynamic_limit_sec}

Consider $G_n:=\Z_n^d$ (for even $n$), with the zero BC $(B_n,\mu_n)$ (see \eqref{zero_BC_def}) and let $f\unifin\Hom(G_n,B_n,\mu_n)$. We think of $G_n$ as embedded in $\Z^d$ as $[-n/2,n/2-1]^d$ (with the zero BC on the boundary of this box) and say that the distribution of $f$ converges weakly as $n\to\infty$ if for every finite $S\subseteq\Z^d$, the distribution of $f$ restricted to $S$ converges. In this case, we call the limiting measure the \emph{thermodynamic limit} of the homomorphism model with zero BC in dimension $d$. We believe, but do not prove, that for sufficiently high dimension, the homomorphism model has a thermodynamic limit with zero BC. We next outline a strategy which can possibly be used to prove this claim. Making this strategy rigorous is left for future research. Fix a dimension $d$ large enough for the following arguments and a finite set $S\subseteq \Z^d$. Consider $f\unifin\Hom(G_{n_1},B_{n_1},\mu_{n_1})$ and independently $g\unifin\Hom(G_{n_2},B_{n_2},\mu_{n_2})$ for some $n_1\ge n_2$ with $n_2$ large enough so that $S\subseteq V[G_{n_2}]$. Let $Z_f:=\{v\in V^\even\ |\ f(v)=0\}$ and $Z_g:=\{v\in V^\even\ |\ g(v)=0\}$. Let also $Z:=Z_f\cap Z_g\subseteq V[G_{n_2}]$. Finally let $\Omega$ be the event that every path from $S$ to the boundary of the cube $[-n_2/2,n_2/2-1]^d$ intersects $Z$. We observe that conditioned on $\Omega$, the distribution of $f$ restricted to $S$ coincides with the distribution of $g$ restricted to $S$ (see Lemma~\ref{level_set_cond_lemma} for a similar statement). Hence, the total variation distance of the distribution of $f$ restricted to $S$ from the distribution of $g$ restricted to $S$ is at most $\P(\Omega^c)$. Thus it will be sufficient to show that as $n_1,n_2\to\infty$, $\P(\Omega)\to 1$. This can be seen as a percolation question, in which $Z$ is the set of closed sites (explicitly, all sites in $V^\odd$ are open and a site in $V^\even$ is open if and only if $f(v)\neq 0$ or $g(v)\neq 0$). In this terminology, what we need to show is that the probability that the set $S$ is connected to distance $n$ (taking $n_1,n_2$ much larger than $n$) by a path of open sites decays to $0$ with $n$. The reason for this is heuristically clear, Theorem~\ref{height_at_point_thm} shows us that for $v\in V^\even$ we have $\P(v\text{ is open})\le \exp(-cd/\log^2 d)$ for some $c>0$ (since the zero BC has full projection) whereas the critical probability for independent percolation on $\Z^d$ is only polynomially small in $d$. The main difficulty in completing this argument is to show that this percolation model is indeed subcritical although there are dependencies between the different sites.

We now turn to the case of $G_n$ with the one-point BC $(B,\mu)$, embedded in $\Z^d$ as before with $B=\{\vec{0}\}$ (where $\vec{0}$ is the origin of $\Z^d$). We believe, but do not prove, that in sufficiently high dimension the homomorphism model also has a thermodynamic limit with the one-point BC. We expect this thermodynamic limit to have the following form: There is a distribution $\L$ on the integers, symmetric around zero with rapidly decaying tails ($\L$ is the ``average height'' of the limiting distribution) such that in order to obtain a sample from the thermodynamic limit with the one-point BC, one samples an integer $h$ from $\L$ and a height function $f$ from the thermodynamic limit with zero BC conditioned to have $f(\vec{0})=-h$, and then returns $f+h$ as the sample from the thermodynamic limit with the one-point BC.

\subsection{Lipschitz Height Functions}\label{Lipschitz_functions_section}
In this section we show how to extend the results described in the previous section to Lipschitz height functions. The possibility and ease of this extension are a direct consequence of a bijection discovered by Ariel Yadin \cite{Y09}. We start by describing this bijection.

Let $G$ be a finite, connected and bipartite graph. For $\emptyset\neq B\subseteq V[G]$ and $\mu:B\to\Z$, we recall that $(B,\mu)$ is a \emph{Lipschitz legal boundary condition} if $\Lip(G,B,\mu)\neq\emptyset$. We let $G_2:=G\times\Z_2$. We note that $G_2$ is also bipartite and fix on it a bipartition $(V^\even_2, V^\odd_2)$. We think of $G_2$ as two copies of the graph $G$ with edges between the two copies of each vertex and denote the two vertices in $G_2$ corresponding to the vertex $v\in G$ by $(v,0)$ and $(v,1)$. The labeling is chosen so that $(v,0)\in V^\even_2$ and $(v,1)\in V^\odd_2$. Note that if $v,w\in V[G]$ and $v\adj{G} w$ then $(v,i)\adj{G_2}(w,1-i)$ for $i\in\{0,1\}$. We remind that for $\emptyset\neq B'_2\subseteq V[G_2]$ and $\mu'_2:B'_2\to\Z$, the pair $(B'_2,\mu'_2)$ is called a (homomorphism) legal boundary condition if $\Hom(G_2,B'_2,\mu'_2)\neq\emptyset$ and $\mu'_2$ takes even values on $V^\even_2$ and odd values on $V^\odd_2$. Finally, fixing a boundary condition $(B,\mu)$ on $G$, we set $B_2:=\{(v,i)\ |\ v\in B, i\in\{0,1\}\}$ and define $\mu_2:B_2\to\Z$ by
\begin{equation*}
\mu_2(v,i):=\begin{cases} \mu(v)& i=\mu(v)\bmod 2\\ \mu(v)-1&i\neq \mu(v)\bmod 2\end{cases}.
\end{equation*}
\begin{theorem}(Yadin Bijection \cite{Y09})\label{Yadin_bijection_thm}
\begin{enumerate}
\item $(B,\mu)$ is a Lipschitz legal boundary condition if and only if $(B_2,\mu_2)$ is a homomorphism legal boundary condition.
\item If $(B,\mu)$ is a Lipschitz legal boundary condition then the mapping $T:\Hom(G_2,B_2,\mu_2)\to\Lip(G,B,\mu)$ defined by
\begin{equation}\label{Yadin_bijection_def}
T(f)(v):=\max(f((v,0)), f((v,1)))
\end{equation}
is a bijection. Furthermore, in this case
\begin{equation}\label{Yadin_bijection_range_prop}
\range(T(f))=\range(f)-1
\end{equation}
for all $f\in\Hom(G_2,B_2,\mu_2)$.
\end{enumerate}
\end{theorem}

We note that there is no boundary condition $(B,\mu)$ on $G$ (with $B\neq\emptyset)$ for which the corresponding $(B_2,\mu_2)$ has $\mu_2(b)=0$ for all $b\in B_2$. To remedy this, we generalize slightly the definition of $\Lip(G,B,\mu)$. Given $\emptyset\neq B\subseteq V[G]$ and a set $\Psi$ of functions $\mu:B\to\Z$ we let $\Lip(G,B,\Psi):=\cup_{\mu\in\Psi} \Lip(G,B,\mu)$. We say that $\Psi$ is \emph{zero-one} if $\Psi$ is the set of all functions of the form $\mu:B\to\{0,1\}$ and we say that $(B,\Psi)$ is a \emph{Lipschitz legal BC} if $\Lip(G,B,\Psi)\neq\emptyset$. As usual, we write $g\unifin\Lip(G,B,\Psi)$ when $g$ is sampled uniformly from $\Lip(G,B,\Psi)$. We then obtain the following corollaries from the Yadin bijection.
\begin{corollary}\label{Lip_zero_one_cor}
For every Lipschitz legal BC $(B,\Psi)$ with zero-one $\Psi$, the Yadin bijection $T$ defined in \eqref{Yadin_bijection_def} maps $\Hom(G_2,B_2',\mu_2')$, where $B_2':=\{(v,0)\ |\ v\in B\}$ and $\mu_2'$ is zero, bijectively to $\Lip(G,B,\Psi)$, with the relation \eqref{Yadin_bijection_range_prop} holding for all $f\in\Hom(G_2,B_2',\mu_2')$.
\end{corollary}
\begin{corollary}\label{Lip_one_point_cor}
Let $g\unifin\Lip(G,B,\mu)$ and $f\unifin\Hom(G_2,B'_2,\mu'_2)$ where $(B,\mu)$ and $(B'_2,\mu'_2)$ are one-point BCs on $G$ and $G_2$ respectively. Then $\range(g)\eqd\range(f)-1$.
\end{corollary}

Using the bijection and its corollaries, we deduce analogs of the theorems of Section~\ref{hom_main_results_sec}. We start with an analogue of Theorem~\ref{height_at_point_thm} on the height of a uniform height function.
\begin{theorem}\label{Lipschitz_height_at_point_thm}
There exist $d_0\in\N$, $c>0$ such that for all $d\ge d_0$, non-linear tori $G$, Lipschitz legal boundary conditions $(B,\mu)$ with non-positive $\mu$ and $x\in V[G]$, if $g\unifin\Lip(G,B,\mu)$ then
\begin{equation*}
\P(g(x)\ge t)\le \exp\left(-\frac{c\vol(\lceil t/2\rceil-1)}{\min(t,d)\log^2 d}\right)\qquad\text{for all $t\ge 3$}.
\end{equation*}
Furthermore, if $t\ge 3$ satisfies $\vol(\lceil t/2\rceil-1)\le\frac{1}{6}n_d$ then
\begin{equation*}
\P(g(x)\ge t)\le \exp\left(-\frac{c\vol(\lceil t/2\rceil-1)}{\log^2 d}\right).
\end{equation*}
Finally, if $B$ has full projection then
\begin{equation*}
\P(g(x)\ge t)\le \exp\left(-\frac{c\vol(t-1)}{\log^2 d}\right)\qquad\text{for all $t\ge 2$}.
\end{equation*}
\end{theorem}

Corollary~\ref{even_lattice_zero_cor} and the bijection now imply that for special boundary conditions, the random Lipschitz function is highly concentrated, taking only two values on most of the torus. We demonstrate this for one specific BC (see \eqref{G_coordinates} for the coordinate system),
\begin{equation}\label{Lipschitz_zero_BC}
 B^\square:=\{(x_1,\ldots, x_d)\in V[G]\ \big|\ \exists i\text{ s.t. $x_i\in\{0,n_i-1\}$}\}.
\end{equation}

\begin{corollary}\label{two_values_Lip_cor}
There exist $d_0\in\N$, $c>0$ such that for all $d\ge d_0$, non-linear tori $G$, if $g\unifin\Lip(G,B^\square,\Psi)$ with zero-one $\Psi$ then
\begin{equation*}
\frac{\E|\{v\in V[G]\ |\ g(v)\notin \{0,1\}\}|}{|V[G]|}\le \exp\left(-\frac{c d}{\log^2 d}\right).
\end{equation*}
\end{corollary}
This phenomena can be observed in Figure~\ref{2d_3d_100_fig} where a slice of a sample of a Lipschitz function on $\Z_{100}^3$ with these boundary conditions (shifted by $\frac{1}{2}$ for symmetry) is depicted.

We continue with a theorem about the range of a random Lipschitz function.
\begin{theorem}\label{Lipschitz_range_thm}
There exist $d_0\in\N$, $C>0$ such that for all $d\ge d_0$ and non-linear tori $G$, if we set
\begin{equation*}
k:=\min\left\{ m\in\N\ \big|\ \vol(m)\ge C\log^2 d\log|V[G]|\right\}
\end{equation*}
and let $g\unifin\Lip(G,B,\Psi)$ for Lipschitz legal BC $(B,\Psi)$ with zero-one $\Psi$, or let $g\unifin\Lip(G,B,\mu)$ for a one-point BC $(B,\mu)$, then
\begin{equation*}
\P(\range(g)>2k)\le \frac{1}{|V[G]|^4}.
\end{equation*}
\end{theorem}

Corollary~\ref{range_sharpness_cor} and the bijection show that our range bounds are sharp for one-point BCs.
\begin{corollary}\label{two_sided_Lipschitz_range_cor}
There exist $d_0\in\N$, $C_d,c_d>0$ such that for all $d\ge d_0$, non-linear tori $G$ and the one-point BC $(B,\mu)$, if $g\unifin\Lip(G,B,\mu)$ then
\begin{equation*}
\P(c_d r \le \range(g) \le C_d r)\ge 1-\frac{1}{|V[G]|^3},
\end{equation*}
where $r:=\min\left\{ m\in\N\ \big|\ \vol(m)\ge \log|V[G]|\right\}$.
\end{corollary}

We also obtain an analogue of Theorem~\ref{finite_range_thm}.
\begin{theorem}\label{Lipschitz_finite_range_thm}
For any integer $k\ge 2$, there exist $d_0(k)$ and $c_k>0$ such that for all $d\ge d_0(k)$ and non-linear tori $G$ with $|V[G]|\le \exp\left(\frac{c_kd^{k}}{\log^2 d}\right)$, if we let $g\unifin\Lip(G,B,\Psi)$ for Lipschitz legal BC $(B,\Psi)$ with zero-one $\Psi$, or let $g\unifin\Lip(G,B,\mu)$ for a one-point BC $(B,\mu)$, then
\begin{equation*}
\P(\range(g)>2k)\le \exp\left(-\frac{c_kd^k}{\log^2 d}\right).
\end{equation*}
\end{theorem}

Note that the range of $g$ is one less than the range of $f$ in the corresponding Theorem~\ref{finite_range_thm}. This follows from \eqref{Yadin_bijection_range_prop}.

Perhaps surprisingly, if we take in the above theorem $g\unifin\Lip(G,B,\mu)$ for Lipschitz legal BC $(B,\mu)$ with zero $\mu$ (instead of $(B,\Psi)$ with zero-one $\Psi$), then we do not expect the theorem to remain true in general. Indeed, if $G=\Z_2^d$ and $B=\{(x_1,\ldots, x_d)\ \big|\ \sum_{i=1}^d x_i = \lfloor \frac{d}{2}\rfloor\}$, say, then the boundary conditions divide the torus into two, roughly equal, connected components. Now if, as the theorem suggests, the typical random function will take 4 values on each of these components, then, by symmetry of the distribution of the function on each component under taking negations, there would be positive probability (bounded away from $0$ with $d$) that these 4 values would not be the same on both components, thus leading to the function taking at least 5 values overall (with probability bounded away from $0$ with $d$). In other words, we expect that for some boundary sets $B$, a random Lipschitz function can be more concentrated when its boundary values consist of zeros and ones than when they consist only of zeros.

Theorem~\ref{linear_torus_thm} concerning the behavior of the random height function on linear tori also has an analogue for Lipschitz functions.
\begin{theorem}\label{Lipschitz_linear_torus_thm}
For all $0<\lambda<\frac{1}{4\log 2}$, there exist $\alpha=\alpha(\lambda)>0$ and $C=C(\lambda)>0$ such that for all dimensions $d\ge 2$ and all $\lambda$-linear tori $G$, if $g\unifin\Lip(G,B,\mu)$ with the one-point BC $(B,\mu)$ then
\begin{equation*}
\P(\range(g)\le |V[G]|^\alpha)\le C|V[G]|^{-\alpha}.
\end{equation*}
\end{theorem}
Consequently, the roughening transition discussed in Section~\ref{roughening_transition_sec} occurs also for the Lipschitz height function model.

\subsection{Proof Sketches, Reader's Guide and Acknowledgments}\label{proof_sketch_readers_guide_ack_sec}

As explained in the introduction and previous sections, we first prove our theorems for the homomorphism model and then use the Yadin bijection (Theorem~\ref{Yadin_bijection_thm}) to transfer our results to the Lipschitz model. For simplicity, we will assume throughout this sketch (except in the section on linear tori) that $G=\Z_n^d$ (for even $n\ge 4$ and large $d$), but our proofs remain essentially unchanged for more general non-linear tori in high dimensions. The main ingredient in proving our results for the homomorphism model on non-linear tori is to prove the level set theorem (Theorem~\ref{main_thm}). We start by explaining some key ideas which go into the proof. We will then explain how these ideas are put together. Related ideas have appeared in the work of Galvin and Kahn \cite{GK04}.

{\bf Expanding Transformation:} Given a finite set $U$, a subset $\Omega\subseteq U$ and a transformation $T:\Omega\to \PP(U)$ (where $\PP(U)$ is the power set of $U$), we define two parameters
\begin{equation*}
\label{expanding_transform_param}
\begin{split}
&\Out(T):=\min_{f\in \Omega} |T(f)|,\\
&\In(T):=\max_{g\in U} |\{f\in\Omega\ \big|\ g\in T(f)\}|.
\end{split}
\end{equation*}
We call $\tau(T):=\frac{\Out(T)}{\In(T)}$ the \emph{expansion
factor} of $T$ and call $T$ an \emph{expanding transformation} if
$\tau>1$. It is not difficult to verify that the mere existence of
an expanding transformation $T$ implies that
$\frac{|\Omega|}{|U|}\le \frac{1}{\tau(T)}$. We will apply this idea
to the space $U=\Hom(G,B,\mu)$ (for some BC $(B,\mu)$) in order to
deduce that certain sets of homomorphism height functions
$\Omega\subseteq U$ have small probability.

{\bf Odd Cutsets:} Given $\emptyset\neq B\subseteq V[G]$ and $x\in V[G]$, a minimal edge cutset $\Gamma$ separating $x$ and $B$ is a set of edges of $G$ which separate $x$ and (every vertex of) $B$ and have the property that if any edge is removed from $\Gamma$ then they no longer separate $x$ and $B$. We denote the set of such cutsets by $\MCut(x,B)$. The \emph{interior (vertex) boundary} of such a cutset $\Gamma$ is the set of vertices incident to $\Gamma$ and connected to $x$ in $G$ by a path which does not cross $\Gamma$. We distinguish a special sub-class of $\MCut(x,B)$, the \emph{odd (minimal edge) cutsets}, which we denote by $\OMCut(x,B)$, which are those $\Gamma\in\MCut(x,B)$ whose interior boundary lies completely in the odd bi-partition class of $G$. Such cutsets arise naturally as the level sets of homomorphism height functions (see Figures~\ref{level_sets_fig} and \ref{shift_transform_fig} for an illustration) and their understanding is fundamental to our analysis.

It will be important to distinguish a subset of the interior boundary of an odd cutset $\Gamma$ by the following definition. We say that a vertex in the interior boundary is \emph{exposed} if it is incident to at least $2d-\sqrt{d}$ edges of $\Gamma$. Thus exposed vertices ``see'' the cutset in nearly all directions.

We do not address the question of the number of odd cutsets in this work (see also the open questions in Section~\ref{open_questions_sec}), but use related facts and hence remark to the reader (this fact is neither proved nor used) that in the whole of $\Z^d$, the number of odd cutsets separating the origin from infinity and having at least $L$ edges is at least $2^{(1+\eps_d)L/2d}$ for $d\ge 2$, some $\eps_d>0$ and large $L$. This can be seen by counting those odd cutsets which approximate closely the boundary of a large cube with sides orthogonal to the axes of $\Z^d$.

We shall need two structural results on odd cutsets which we now explain.

{\bf Odd Cutsets with Rough Boundary:} For an odd cutset $\Gamma$ (fixing some $x$ and $B$), we introduce the parameter $R_\Gamma$ to be $\sum_v \min(P_\Gamma(v), 2d-P_\Gamma(v))$ where the sum is over all vertices in the interior boundary of $\Gamma$ and $P_\Gamma(v)$ is the number of $\Gamma$ edges incident to $v$ (the $R$ is for regularity and the $P$ is for plaquette). This parameter is a measure of the regularity of $\Gamma$, with a value significantly smaller than $d$ times the size of the interior boundary indicating some roughness of $\Gamma$. In the first of our structural results, Theorem~\ref{count_cutsets_thm}, we prove that
\begin{equation*}
|\OMCut(x,B,R)|\le \exp\left(\frac{C\log^2 d}{d} R\right)
\end{equation*}
for some $C>0$, where $\OMCut(x,B,R)$ is the set of odd cutsets $\Gamma\in\OMCut(x,B)$ having $R_\Gamma=R$.

We shall not sketch in detail here the way this theorem is proved, but only mention that it proceeds roughly by describing an odd cutset by a ``skeleton'' of it (which, in a certain graph, is a dominating set for the interior boundary of the cutset), and showing that the number of such skeletons is not too large. The odd property of the cutset is fundamentally used (and indeed, the analogous theorem for general cutsets, those in $\MCut$, may well be false).

We will use this theorem in the following way. Consider an odd cutset $\Gamma$ having exactly $L$ edges and at least $\left(1-\frac{\lambda}{\log^2 d}\right)\frac{L}{2d}$ exposed vertices, for some $\lambda>0$. Since an exposed vertex is incident to at least $2d-\sqrt{d}$ edges of $\Gamma$, and these edges are distinct from one exposed vertex to the other, it follows that the exposed vertices alone are ``responsible'' for $\left(1-\frac{\lambda}{\log^2 d}\right)\frac{(2d-\sqrt{d})L}{2d}$ of the $L$ edges of $\Gamma$. Thus the boundary of $\Gamma$ is, in a sense, quite rough, and we may suspect that there are not that many odd cutsets with this property ($L$ edges and $\left(1-\frac{\lambda}{\log^2 d}\right)\frac{L}{2d}$ exposed vertices). Indeed, from the above theorem it is not difficult to deduce that their number is at most $\exp\left(\frac{C\lambda}{d} L\right)$ for some $C>0$, which is the estimate we shall use in the sequel.

{\bf Interior Approximation to Odd Cutsets:} The second of our two structure theorems for odd cutsets, Theorem~\ref{interior_approximation_theorem}, shows that odd cutsets may be approximated well in a certain sense. To explain this, we let $\Gamma$ be an odd cutset (fixing some $x$ and $B$) and say that a set of vertices $E$ is an \emph{interior approximation} to $\Gamma$ if it is contained in the interior of $\Gamma$ (those vertices reachable from $x$ by a path which does not cross $\Gamma$) and contains all the non-exposed vertices in the interior boundary of $\Gamma$. Theorem~\ref{interior_approximation_theorem} then shows that, when $B$ is a singleton, there exists a family of subsets of $V[G]$ of size at most $2\exp\left(\frac{C\log^2 d}{d^{3/2}} L\right)$ for some $C>0$, which contains an interior approximation to every odd cutset of size $L$. Thus, while the total number of such odd cutsets may exceed $2^{(1+\eps_d)L/2d}$ (as remarked above), they may be grouped into sets having the same interior approximation with the number of such sets not exceeding $2\exp\left(\frac{C\log^2 d}{d^{3/2}} L\right)$.

As before, we shall not sketch in detail the proof of this theorem, but mention that it proceeds roughly similar to the proof of our first structural result, by describing an odd cutset by a ``skeleton'' of it, and showing that the number of such skeletons is not too large. The main added ingredient is a classification of the interior boundary of $\Gamma$ into three types of vertices: the exposed vertices, the vertices incident to at most $\sqrt{d}$ edges of $\Gamma$ and the vertices incident to between $\sqrt{d}$ and $2d-\sqrt{d}$ edges of $\Gamma$. It turns out that, compared to our first structural theorem, a much smaller skeleton suffices in this theorem since one is only interested in recovering the vertices of the second and third type (with the third type being much easier to handle than the second). Again, the odd property of the cutset is fundamentally used.

{\bf The Level Set Theorem:} We now explain how the previous ingredients are put together to prove the level set theorem, Theorem~\ref{main_thm}. We shall explain only the case of one level set, $k=1$. The cases in which $k>1$ follow in a simple manner from the proof of this case (see Section~\ref{expanding_transformation_sec}). Fixing a graph $G$, boundary conditions $(B,\mu)$, $x\in V[G]$ and $L$, we aim to show that the set $\Omega_{x,L}$, of homomorphism height functions having a level set of length $L$ around $x$, is very small compared to the whole of $\Hom(G,B,\mu)$. We will do so using the concept of expanding transformation described above. We will construct a $T:\Omega_{x,L}\to\PP(\Hom(G,B,\mu))$ and show that there exists a partition of $\Omega_{x,L}$ into (not too many) subsets such that $T$ is expanding (with a large expansion factor) on each of these subsets.

\begin{figure}[t!]
\centering
{\includegraphics[width=\textwidth,viewport=70 127 515 300,clip]{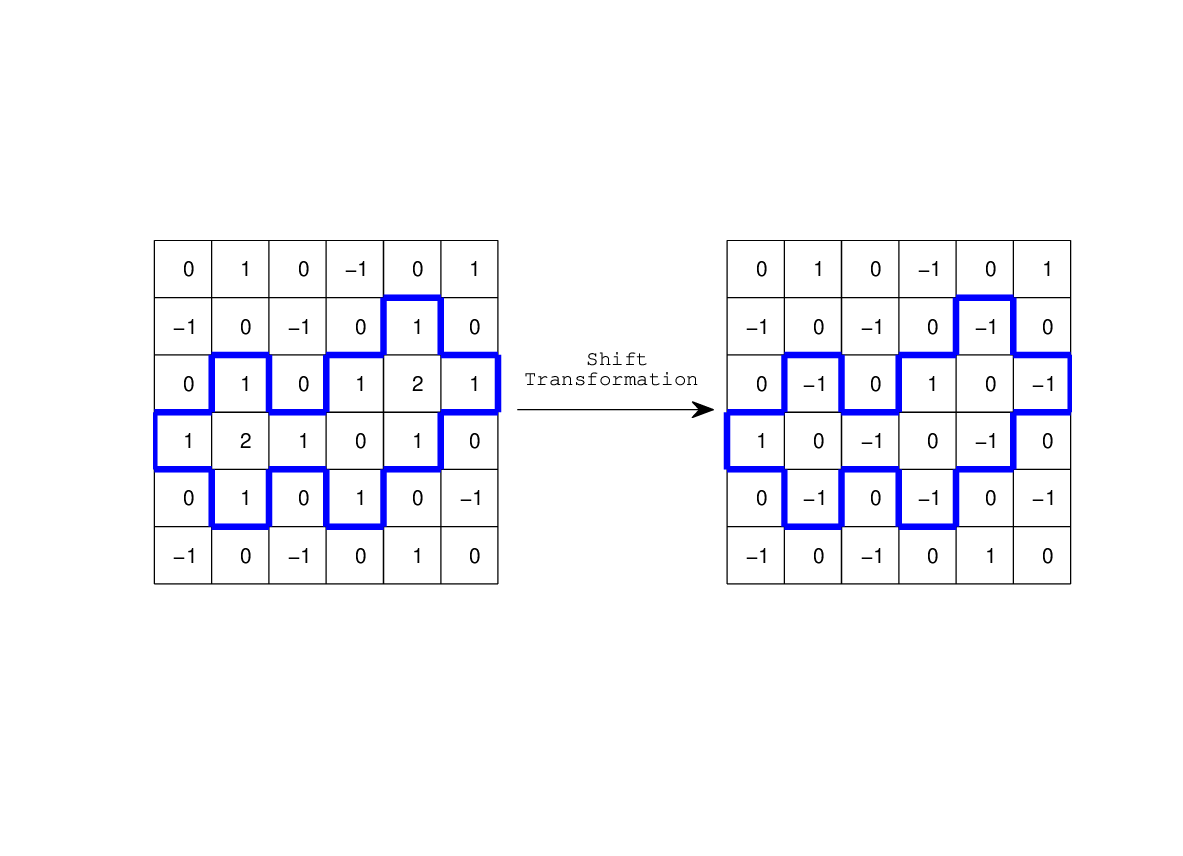}}
{\it
\caption{An illustration of the shift transformation. The function on the left is a homomorphism height function on a 6 x 6 torus with zero boundary conditions and the shaded blue line is a level set of it on which the shift transformation is applied.  The transformation replaces the value at each vertex inside the level set by the value at its neighbor to the right, minus one. The resulting function, depicted on the right, is again a homomorphism height function, and has the property that each vertex which is inside and immediately to the left of the level set is surrounded by zeros. The shaded blue line is drawn on the right function for convenience only, it is not a level set of that function. \label{shift_transform_fig}}
}
\end{figure}

We start the construction of $T$ by defining $\shift:\Omega_{x,L}\to\Hom(G,B,\mu)$, the shift transformation (see Figure~\ref{shift_transform_fig}). For $f\in \Omega_{x,L}$, we denote its level set $\LS(f,x,B)$ by $\Gamma$ and recall that it is an odd cutset separating $x$ and $B$. We let $\CC_1$ be the set of vertices in the interior of $\Gamma$ (so that $x\in\CC_1$) and define $\shift(f)(v)$ to equal $f(v)$ for $v\notin\CC_1$ and to equal $f(v+e_1)-1$ for vertices $v\in\CC_1$, where $v+e_1$ is the vertex located one unit from $v$ in the first coordinate direction (see Figure~\ref{shift_transform_fig}). Informally, on vertices of $\CC_1$, $\shift$ shifts the function $f$ by one lattice space (in the first coordinate direction) and subtracts one from its values. One can then verify that $\shift(f)$ is indeed in $\Hom(G,B,\mu)$ for $f\in\Omega_{x,L}$. The next step is to define the set $E_{1,1}$ of vertices $v$ in the interior of $\Gamma$ for which $(v,v+e_1)\in\Gamma$, and to check that if $v\in E_{1,1}$, then necessarily $\shift(f)(w)=0$ for all neighbors $w$ of $v$ (as in Figure~\ref{shift_transform_fig}). In other words, denoting $g:=\shift(f)$, we observe that for vertices $v\in E_{1,1}$, placing either $+1$ or $-1$ in $g(v)$ results in a valid homomorphism height function. This leads naturally to the definition of $T_1:\Omega_{x,L}\to\PP(\Hom(G,B,\mu))$ as the transformation which replaces each $f$ by the set of all functions formed from $\shift(f)$ by placing $\pm 1$ at the points of $E_{1,1}$. We have that for each $f$, $|T_1(f)|=2^{|E_{1,1}|}$ where $E_{1,1}$ potentially depends on $f$. However, somewhat curiously, odd cutsets have the additional property that exactly $\frac{1}{2d}$ of their edges are of the form $(v,v+e_1)$ for vertices $v$ in their interior. Thus for all $f\in\Omega_{x,L}$, $|T_1(f)|=2^{L/2d}$.

The transformation $T_1$ is a good candidate for our expanding transformation since, as we have just explained, $\Out(T_1)=2^{L/2d}$. However, it is not so simple to bound the parameter $\In(T_1)$ of this transformation. One approach is to note that the transformation $\shift$ is invertible given the level set $\Gamma$, that is, for any $f\in\Omega_{x,L}$, one can reconstruct $f$ from knowing $\shift(f)$ and $\LS(f,x,B)$. The same then holds for $T_1$ for any $f\in\Omega_{x,L}$, one can reconstruct $f$ from knowing any $g\in T_1(f)$ and $\LS(f,x,B)$. Thus $\In(T_1)$ is bounded by the number of possibilities for $\LS(f,x,B)$, which is itself bounded by the number of odd cutsets of length $L$ surrounding $x$. This approach amounts, more or less, to a Peierls-type argument. Unfortunately, as we remarked above, the number of odd cutsets may well exceed $2^{L/2d}$ and thus this approach fails to show that $T_1$ is expanding. Instead, we deduce a more modest result. First, we recall from our first structure theorem for odd cutsets that the number of odd cutsets having $L$ edges and at least $\left(1-\frac{\lambda}{\log^2 d}\right)\frac{L}{2d}$ exposed vertices in their interior boundary is bounded by $\exp\left(\frac{C\lambda}{d} L\right)$ for some $C>0$. Then, we define $\Omega_{x,L,1}\subseteq\Omega_{x,L}$ to be those $f\in \Omega_{x,L}$ whose level set $\LS(f,x,B)$ has at least $\left(1-\frac{\lambda}{\log^2 d}\right)\frac{L}{2d}$ exposed vertices in its interior boundary. We conclude that $\In(T_1)\le \exp\left(\frac{C\lambda}{d} L\right)$ on $\Omega_{x,L,1}$ and thus if $\lambda$ is chosen small enough then $T_1$ is expanding on $\Omega_{x,L,1}$  and $\P(\Omega_{x,L,1})\le 2^{-\frac{cL}{d}}$ for some $c>0$.

It remains to bound $\P(\Omega_{x,L,2})$ for $\Omega_{x,L,2}:=\Omega_{x,L}\setminus\Omega_{x,L,1}$. Our second structure theorem for odd cutsets, Theorem~\ref{interior_approximation_theorem}, motivates a change in the definition of $T_1$. We define the transformation $T_2:\Omega_{x,L}\to\PP(\Hom(G,B,\mu))$ as follows: For each $f$, $T_2(f)$ is the set of functions obtained by modifying each $g\in T_1(f)$ to equal $1$ on all the exposed vertices of $\LS(f,x,B)$. The modification is achieved by noting that $g$ must equal either $1$ or $-1$ on each exposed vertex, identifying for each exposed vertex $v$ for which $g(v)=-1$ the component of it in $G\setminus\{v\ |\ g(v)=0\}$ and negating the values of $g$ on this component. The advantage of $T_2$ over $T_1$ is that it preserves more information on the positions of the exposed vertices of the level set of its input. Its disadvantage is that $|T_2(f)|$ can be much smaller than $2^{L/2d}$ if $f$ has many exposed vertices.

Next, we observe that if $f\in\Omega_{x,L}$ and $h\in T_2(f)$, then knowledge of $h$ and an \emph{interior approximation to $\LS(f,x,B)$} is sufficient to recover $\LS(f,x,B)$ completely. This follows directly from the fact that $\LS(f,x,B)$ is defined solely in terms of the union of components of $B$ in $G\setminus \{v\ |\ f(v)=1\}$. By definition of $T_2$ and interior approximations, this union of components is the same as the union of components of $B$ in $G\setminus (\{v\ |\ h(v)=1\} \cup E)$ where $E$ is an interior approximation to $\LS(f,x,B)$. We would like to use this to bound $\In(T_2)$ by the bound on the number of interior approximations given by Theorem~\ref{interior_approximation_theorem}. However, unlike $T_1$, it is not true that for any $f\in\Omega_{x,L}$, one can reconstruct $f$ from knowing any $h\in T_2(f)$ and $\LS(f,x,B)$. To recover $f$, we need to recover $\LS(f,x,B)$ and, in addition, enumerate on which negations (of the values of $h$ on exposed vertices) were performed in the definition of $T_2(f)$. Potentially, this enumeration factor is as large as $2$ to the power of the number of exposed vertices.

The above discussion shows that the expansion properties of $T_2$ improve when restricted to subsets of $\Omega_{x,L}$ on which the level set of the function has few exposed vertices in its interior boundary. Indeed, we can show that when restricted to (suitable partitions of) the subset of functions having exactly $m$ such exposed vertices, then the expansion factor of $T_2$ is at least $2^{L/2d-m-1}\exp\left(-\frac{C\log^2 d}{d^{3/2}} L\right)$ (this is slightly worse on general non-linear tori). Recalling that $m\le \left(1-\frac{\lambda}{\log^2 d}\right)\frac{L}{2d}$ on $\Omega_{x,L,2}$, we deduce that $T_2$ is expanding on (suitable subsets of) $\Omega_{x,L,2}$ and conclude that $\P(\Omega_{x,L,2})\le d^C\exp\left(-\frac{cL}{d\log^2d}\right)$ for some $C,c>0$, proving the level set theorem.

{\bf Height and Range:} We now explain how
Theorem~\ref{homogeneous_torus_thm} (for the homomorphism model) and
its more general versions, Theorems~\ref{height_at_point_thm} and
\ref{range_thm}, follow from the level set theorem. Fix $t\ge 1$.
Assuming that our boundary values $\mu$ are non-positive, we note
that if our random height function $f$ has $f(x)\ge t$ then, since
the function changes by one between adjacent vertices, we must have
$f(v)\ge 1$ for all vertices $v$ whose (graph) distance from $x$ is
at most $t-1$. Thus the level set $\LS(f,x,B)$ must surround a
(graph) ball of radius $t-1$. If we could deduce from this fact that
$|\LS(f,x,B)|$ is large, when $t$ is large, then we would deduce
from our level set theorem that the event $\{f(x)\ge t\}$ has small
probability. However, $\LS(f,x,B)$ need not be large. For example,
if our boundary set $B$ is a singleton $\{b\}$ then it is possible
that the level set contains only the $2d$ incident edges to $b$. To
overcome this difficulty, we define for each $i\ge 1$ the level set
$\LS_i(f,x,B)$: the outermost height $i$ level set of $f$ around
$x$, which is defined analogously to $\LS(f,x,B)$ (in fact, it
equals $\LS(\tilde{f},x,B)$ for $\tilde{f}:=f-(i-1)$). We then
observe, again using that our boundary values $\mu$ are non-positive
and that the function changes by one between adjacent vertices, that
if $f(x)\ge t$ then $\LS_i(f,x,B)$ must separate a (graph) ball of
radius $i-1$ from a ball of radius $t-i$. In
Section~\ref{general_isoperimetry_section} we develop isoperimetric
estimates which show that these conditions (and a technical
assumption involving $n$, the side-length of the torus) imply that
$\LS_i(f,x,B)$ is at least as large as the size of the boundary of a
ball of radius $\min(i-1, t-i)$. Thus we finally obtain, by taking
$i=\lceil \frac{t}{2}\rceil$ (and assuming that $t\ge 3$), that $f$
has a level set of length at least $c_d t^{d-1}$. Combined with the
level set theorem, this implies that the probability of the event
$\{f(x)\ge t\}$ is at most $\exp(-c_d t^{d-1})$.

Theorem~\ref{homogeneous_torus_thm} states an even stronger fact,
that $\P(f(x)\ge t)\le \exp(-c_d t^d)$. This implies the estimate on
$\P(\range(f)\ge C_d\log^{1/d} n)$ by a union bound. As mentioned in
the introduction, in the case of a one-point BC, the matching lower
bound on $\range(f)$ follows from Theorem~\ref{BYY_thm} of
\cite{BYY07}. To obtain this stronger estimate on $\P(f(x)\ge t)$,
we observe that the level set $\LS_i(f,x,B)$ is defined solely in
terms of the values of $f$ on the exterior of the level set and on
the interior vertex boundary of the level set. Thus, given
$\LS_i(f,x,B)$, the distribution of $f$ in the interior of the level
set equals the distribution of a random homomorphism height
function, on this interior, with boundary values $i$ on the interior
boundary of $\LS_i(f,x,B)$ (this fact is formalized in
Lemma~\ref{level_set_cond_lemma}). This implies that the level set
theorem may be applied inductively, first to $\LS_1(f,x,B)$, then to
$\LS_2(f,x,B)$ given $\LS_1(f,x,B)$ and so on, until applying it to
$\LS_t(f,x,B)$ given $\LS_i(f,x,B)$ for all $1\le i<t$. We conclude
that the probability that $f(x)\ge t$ and, for $1\le i< t$,
$|\LS_i(f,x,B)|=L_i$ is at most $\exp(-c_d\sum_{i=1}^t L_i)$. But
the isoperimetric estimates mentioned in the previous paragraph
imply that if $f(x)\ge t$, then necessarily at least order $t$ of
the level sets $\LS_i(f,x,B)$ have size of order $t^{d-1}$, thus
giving the required estimate $\P(f(x)\ge t)\le \exp(-c_d t^d)$.

{\bf Linear Tori:} Finally, we explain the ideas behind the proof of
Theorem~\ref{linear_torus_thm}, which shows that random homomorphism
height functions on $\lambda$-linear tori, with
$\lambda<\frac{1}{2\log 2}$ and the one-point BC, have large range
with high probability. For concreteness, we focus on the case that
$G=\Z_n\times\Z_{\lfloor \lambda \log n\rfloor}$ for some
$\lambda<\frac{1}{2\log 2}$, but the general case follows similarly.
For such a torus, we introduce the notion of a ``wall'' in the
homomorphism function $f$. A wall consists of two adjacent, roughly
vertical, lines of vertices (crossing the torus in the ``short''
direction) on which $f$ is constant (a different constant on each of
the lines). Intuitively, such walls form since the chance that they
occur for any particular horizontal coordinate is of order
$2^{-2\lambda\log n}$ (since the function $f$ has to change in a
prescribed way on, approximately, $2\lambda\log n$ edges), but there
are $n$ possibilities for this coordinate and hence many walls will
form if $\lambda<\frac{1}{2\log 2}$. Our proof formalizes this
argument. The proof then concludes by comparing the behavior of $f$
on these walls to the behavior of a random walk bridge. Since such
bridges have large range with high probability, we are able to
deduce that $f$ does as well.

{\bf Reader's guide:} The rest of the paper is structured as follows. In Section~\ref{preliminaries_sec}, definitions and preliminary results which will be needed throughout the paper are given. The proof of the level set theorem, Theorem~\ref{main_thm}, is given in Section~\ref{level_set_thm_sec}, which is divided into several parts: Section~\ref{expanding_transformation_sec} introduces the notion of expanding transformation and the properties required of it for our proof. Section~\ref{T_def_section} defines the expanding transformation $T$ we will use. In Section~\ref{odd_cutsets_structure_sec} we state and prove our structure theorems for odd cutsets. Finally, Section~\ref{proof_of_transform_theorem_sec} puts together the previous ingredients to deduce that the transformation $T$ has the required expansion properties. In Section~\ref{isoperim_height_range_Lip_sec} we deduce our theorems for the height and range of homomorphism and Lipschitz height functions from the level set theorem. To this end, isoperimetric estimates for cutsets on tori are developed in Section~\ref{general_isoperimetry_section}. Section~\ref{linear_tori_sec} proves Theorem~\ref{linear_torus_thm} on the typical range of values taken by random homomorphisms on linear tori. Finally, in Section~\ref{open_questions_sec} we conclude with a list of open questions.

{\bf Acknowledgments:} I thank Fedor Nazarov for initial discussions
of the problem, Senya Shlosman for many discussions and much
encouragement, Ariel Yadin for useful discussions and for permission
to use his bijection, Gil Kalai and Michael Khanevsky for useful
discussions on the topology of the torus, G\'abor Pete for referring
me to Tim\'ar's paper \cite{T07}, Steven M. Heilman for producing
Figure~\ref{level_sets_fig} and Omer Angel, Itai Benjamini, Sourav
Chatterjee, Amir Dembo, Aernout van Enter, Frank den Hollander, Dima
Ioffe, Bo'az Klartag, Roman Koteck\'y, Gady Kozma, Nir Lev, Andrea
Montanari, Elchanan Mossel, Pierre Nolin, Scott Sheffield, Allan
Sly, Sasha Sodin, Yvan Velenik and Amir Yehudayoff for useful
conversations on this work. Special thanks is extended to Yinon Spinka and to two anonymous referees for a careful reading of the manuscript and for many suggestions which greatly improved the presentation.

This work began while I was on a Clay liftoff fellowship. It benefited from discussions with participants of the XII'th Brazilian school of probability, and was partially completed during stay at the Institut Henri Poincar\'e - Centre Emile Borel - in its thematic program on Interacting Particle Systems, Statistical Mechanics and Probability Theory. I thank all of these for their hospitality and support.

\section{Preliminaries}\label{preliminaries_sec}
In this section we introduce notation used throughout the paper and
prove some preliminary results that we will need. The first time a
notation is introduced it is highlighted in boldface.

{\bf The torus $G$:} For a torus $G$, with even side lengths
$(n_i)_{i=1}^d$ as in \eqref{torus_side_lengths}, we denote by
$\bm{\Delta(G)}$ the degree of (any vertex in) $G$. We have
$\Delta(G):=2d-\sum_{i=1}^d 1_{(n_i=2)}$ and we will frequently use
that
\begin{equation}\label{Delta_G_bounds}
d\le \Delta(G)\le 2d.
\end{equation}
We shall denote
\begin{equation}\label{alpha_def}
\alpha:=\prod_{i=1}^{d-1} n_i,
\end{equation}
the size of the smallest ``section'' of the torus. We let $\bm{d_G}$
stand for the graph distance and for $v,w\in V[G]$ write
$\bm{v\adj{G} w}$ if $d_G(v,w)=1$. Denote by $\bm{S(v)}:=\{w\in
V[G]\ |\ w\adj{G} v\}$, the set of neighbors of $v$ in $G$ and
$\bm{S(E)}:=\cup_{v\in E} S(v)$ for $E\subseteq V[G]$. By definition
$|S(v)|=\Delta(G)$. As in \eqref{G_coordinates}, we fix a coordinate
system for the torus $G$ so that $V[G]=\{(x_1,\ldots, x_d)\ |\ 0\le
x_i\le n_i-1\}$. For $v=(v_1,\ldots, v_d)\in V[G]$ and $1\le i\le
d$, we denote by $\bm{v+e_i}$ the vertex whose coordinates are
$(v_1, \ldots, v_{i-1}, v_i+1 \pmod{n_i}, v_{i+1},\ldots, v_d)$ and
by $\bm{v-e_i}$ the vertex whose coordinates are $(v_1, \ldots,
v_{i-1}, v_i-1 \pmod{n_i}, v_{i+1},\ldots, v_d)$. We similarly
define $\bm{v+ke_i}$ for all $k\in\Z$. We note that $v+e_i=v-e_i$
iff $n_i=2$. Letting $k=\max\{i\ |\ n_i=2\}$ ($k=0$ if $n_1>2$) we
define $\bm{(f_i)_{i=1}^{\Delta(G)}}$ by $f_i=e_i$ for $1\le i\le d$
and $f_i=-e_{i-d+k}$ for $d<i\le \Delta(G)$. By our definitions
$\{v+f_i\ |\ 1\le i\le\Delta(G)\}=S(v)$.

We note a simple expansion property of $G$.
\begin{proposition}\label{2nd_neighborhood_prop}
Let $v\in V[G]$ and $Q:=\{(i,j)\ |\ 1\le i,j\le\Delta(G)\text{, }f_i\neq -f_j\}$. Then for any $(i,j)\in Q$ we have $|\{(k,\ell)\in Q\ |\ v+f_i+f_j = v+f_k+f_\ell\}|\le 2$.
\end{proposition}
\begin{proof}
Let $(i,j), (k,\ell)\in Q$ and suppose that $(i,j)\neq (k,\ell)$ and
\begin{equation}\label{two_step_equality}
v+f_i+f_j=v+f_k+f_\ell.
\end{equation} Suppose first that $f_i=f_j$. If $f_i=e_m$ or $f_i=-e_m$, we must have $n_m=4$ and $f_k=f_\ell = -f_i$ for \eqref{two_step_equality} to hold, proving the proposition in this case. If $f_i\neq f_j$, then $v+f_i+f_j$ differs from $v$ in two coordinates and we must have $k=j$ and $\ell=i$ for \eqref{two_step_equality} to hold, proving the proposition in this case as well.
\end{proof}

We let $\bm{G^{\otimes r}}$ for integer $r>0$ be the graph with the
same vertex set as $G$ and with $u,v\in V[G]$ adjacent if and only
if $1\le d_G(u,v)\le r$. With this notation $G^{\otimes 1}=G$. Note
also that the degree of the vertices in $G^{\otimes r}$ is bounded
above by $\sum_{i=1}^r (2d)^i\le (2d)^{r+1}-1$. We shall need the
following standard counting lemma:
\begin{proposition}\label{span_tree_prop}
Given $v\in V[G]$ and integers $M,r>0$, the number of sets $E\subseteq V[G]$ with $|E|=M$ and $E\cup\{v\}$ connected in $G^{\otimes r}$ does not exceed $(2d)^{2(r+1)M}$.
\end{proposition}

\begin{proof}
To avoid dealing separately with the cases where $v\in E$ and $v\notin E$, let $G_r$ be the graph $G^{\otimes r}$ with the vertex $v$ doubled in the following sense: $G_r$ has as vertex set the vertex set of $G$ union one additional vertex called $v'$, and has as edges the edges that $G^{\otimes r}$ has, an edge from $v'$ to each of the neighbors of $v$ and an edge between $v$ and $v'$. Note that the maximal degree in $G_r$ is bounded by $(2d)^{r+1}$.

For every $E$ as in the proposition, we note that $E\cup\{v'\}$ is connected in $G_r$ and we fix a spanning tree $T_E$ for it. Starting from $v'$, we can perform a depth first search of $T_E$, starting and ending at $v'$ and passing through each edge exactly twice. Since $T_E$ has exactly $M$ edges, we obtain that the number of possibilities for $T_E$ (and hence for $E$) is upper bounded by the number of walks of length $2M$ in $G_r$ which start at $v'$. This gives the required bound.
\end{proof}

In this paper, a cycle is a closed walk having no repeated vertices
(besides its starting and ending point). An edge cycle is the set of
edges of a cycle. A basic 4-cycle is a cycle of the form $v, v+f_i,
v+f_i+f_j, v+f_j, v$ for some $v\in V[G]$, $f_i\neq f_j$ and
$f_i\neq -f_j$. We let $\bm{{\GC}}$ be the graph with the same
vertex set as $G$ and with $u,v\in V[G]$ adjacent if and only if
they lie on some basic 4-cycle. Let $k:=\min\{1\le i\le d\ |\
n_i>2\}$ ($k=\infty$ if $n_d=2$, i.e., on the hypercube) and for
each $k\le i\le d$ and $v\in V[G]$, let $P_i(v)$ be the cycle
$v,v+e_i,v+2e_i,\ldots,v+n_i e_i$ which starts at $v$ and wraps
around the torus once in the $e_i$ direction. We use without proof
the fact that on the torus, for any $(v_i)_{i=k}^d\subseteq V[G]$,
the edge sets of basic 4-cycles and the edge sets of
$(P_i(v_i))_{i=k}^d$ (these are not needed on the hypercube)
generate the cycle space of $G$ over $Z_2$, i.e., any edge cycle can
be written as the exclusive or of some subset of these edge cycles
(this can be seen by taking the tree whose root is at $O=(0,\ldots,
0)$ and in which the parent of $x\in V[G]\setminus\{O\}$ is $x-e_m$
where $m=\min\{1\le i\le d\ |\ x_i>0\}$ and observing that its
fundamental cycles are in the span of the given generating set). Let
$G^+((v_i)_{i=k}^d)$ be the graph with the same vertex set as $G$
and in which $u,v$ are adjacent if they are adjacent in $\GC$ or
both lie on $P_i(v_i)$ for some $i$ ($G^+=\GC$ on the hypercube). A
clever result of Tim\'ar \cite{T07} showing connectivity of
boundaries of connected sets implies
\begin{theorem}[Special case of Lemma 2 in \cite{T07}]\label{Timar_result}
Letting $k=\min\{1\le i\le d\ |\ n_i>2\}$, for any $(v_i)_{i=k}^d\subseteq V[G]$, $x\in V[G]$ and $G$-connected $\CC\subseteq V[G]$, the set
\begin{equation*}
E_1:=\{\text{connected component of $x$ in $V[G]\setminus\CC$}\}\cap\{v\in V[G]\ |\ d_G(v,\CC)=1\}
\end{equation*}
(i.e., the outer boundary of $\CC$ visible from $x$), is connected in $G^+((v_i)_{i=k}^d)$.
\end{theorem}

{\bf Vertex Cutsets:} For $x,y\in V[G]$, let $\bm{{\VCut(x,y)}}$ be
the set of all vertex cutsets (not necessarily minimal) separating
$x$ and $y$. I.e., the set of all $E\subseteq V[G]$ such that any
path from $x$ to $y$ must intersect $E$ (possibly at $x$ or $y$).
Recalling the definition of $\alpha$ from \eqref{alpha_def}, we will
need
\begin{proposition}\label{cutset_start_prop}
Let $x,y\in V[G]$ and $M>0$ an integer. If $M<2\alpha$ then there exists a set $A=A(x,y,M)\subseteq V[G]$ with $|A|\le 30M$ such that every $E\in\VCut(x,y)$ with $|E|\le M$ intersects $A$. If $M\ge 2\alpha$, the same is true with a set $A$ satisfying $|A|\le 31M+n_d$.
\end{proposition}
We use the following lemmas:
\begin{lemma}\label{cutset_lower_bound_lem}
Let $x,y\in V[G]$ and $B_x,B_y\subseteq V[G]$ be connected sets with $x\in B_x$ and $y\in B_y$. Suppose there exist $k$ paths between $B_x$ and $B_y$, pairwise disjoint in their interior. Then every $E\in\VCut(x,y)$ either intersects $B_x\cup B_y$ or has $|E|\ge k$.
\end{lemma}
\begin{proof}
Let $P_1,\ldots, P_k$ be paths between $B_x$ and $B_y$, pairwise disjoint in their interior. Let $Q_j$ be a walk from $x$ to $y$ which travels inside $B_x$ to the starting point of $P_j$, then travels along $P_j$ and finally travels inside $B_y$ to $y$. All the $Q_j$ must intersect $E$ by its definition. Hence if $E$ does not intersect $B_x\cup B_y$ then it intersects each $P_j$ in its interior and hence has at least $k$ points.
\end{proof}
\begin{lemma}\label{minimal_cutset_lem}
Let $x,y\in V[G]$. Every $E\in\VCut(x,y)$ satisfies either $E\cap\{x,y\}\neq \emptyset$ or $|E|\ge d$.
\end{lemma}
\begin{proof}
The lemma is standard, but we give a proof for completeness. Suppose $E\cap\{x,y\}=\emptyset$, then by the previous lemma it is enough to exhibit $d$ paths from $x$ to $y$, disjoint in their interior. By applying translations and reflections to the torus, we may assume without loss of generality that $x=(0,0,\ldots,0)$ and $y=(a_1,a_2,\ldots, a_d)$ with $0\le a_j\le \frac{n_j}{2}$. For each $1\le j\le d$, if $a_j\neq 0$, define the path $P_j$ as the path from $x$ to $y$ going from $x$ to $x+a_j e_j$ by adding $e_j$ each step, then to $x+a_j e_{j}+a_{j+1}e_{j+1}$ by adding $e_{j+1}$, then to $x+a_je_j+a_{j+1}e_{j+1}+a_{j+2}e_{j+2}$ by adding $e_{j+2}$ and so on until $y$, where all subscripts are interpreted cyclically (so that $e_{d+1}=e_1$, $a_{d+2}=a_2$, etc.). If $a_j=0$, we define the path $P_j$ as going from $x$ to $x+e_j$ then to $x+e_j+a_{j+1}e_{j+1}$ and so on until $x+e_j+\sum_{k=1}^{d-1}a_{j+k}e_{j+k}$ and finally to $y$ (by subtracting $e_j$). It is straightforward to verify that these paths are all disjoint in their interiors.
\end{proof}
\begin{proof}[Proof of Proposition~\ref{cutset_start_prop}]
By applying translations and reflections to the torus, we may assume without loss of generality that $x=(0,0,\ldots,0)$ and $y=(a_1,a_2,\ldots, a_d)$ with $0\le a_j\le \frac{n_j}{2}$. Let $P$ be the path from $x$ to $y$ which goes in straight lines, in the positive coordinate directions, from $(0, 0, 0, \ldots, 0)$ to $(a_1, 0, 0, \ldots, 0)$ to $(a_1, a_2, 0, \ldots, 0)$ and so on up to $y$. We start by supposing that $|E|=M'$ for some $M'\le M$ and divide into cases:
\begin{enumerate}
\item $M'<d$. By Lemma~\ref{minimal_cutset_lem}, letting $A_{M'}^1:=\{x,y\}$ we have $E\cap A^1\neq\emptyset$ and $|A_{M'}^1|=2\le 10M'$.
\item $d\le M'<\frac{1}{4}\sum_{j=1}^{d-1} n_j$. Define
\begin{equation*}
B'=\{z\in V[G]\ |\ \exists 1\le j\le d-1\text{ and }0\le i\le \frac{n_j}{2}-1\text{ s.t. }z=x-ie_j+ie_d\}.
\end{equation*}
We have $|B'|=\frac{1}{2}\sum_{j=1}^{d-1} n_j - (d-2)$ and we check that for any $z_1,z_2\in B'$, $z_1\neq z_2$, the paths $P+z_1$ and $P+z_2$ are disjoint. Indeed, the last statement is equivalent to saying $(P-P)\cap(B'-B')=\{(0,\ldots,0)\}$, but  if we write $z_1-z_2=(i_1,\ldots, i_d)$ if $z_1 = z_2 + \sum_{j=1}^d i_je_j$ and $-\frac{n_j}{2}+1\le i_j\le \frac{n_j}{2}$, then each point in $B'-B'$ has sum of coordinates 0 (using the fact that $n_d\ge n_j$ for all $j$) and cannot have its $j$'th coordinate equal $\frac{n_j}{2}$ for any $j$, while each point in $P-P$ either has its $j$'th coordinate equal to $\frac{n_j}{2}$ for some $j$, or all its coordinates are simultaneously non-negative or non-positive.

Continuing, for any $1\le a\le |B'|$, we may find a connected set $B$ with $x\in B$ such that $|B|\le 2a$ and $|B\cap B'|=a$. Taking such a set for $a=M'+1$ (using that $\frac{1}{4}\sum_{j=1}^{d-1} n_j>d$ by the assumption of this item) and letting $A_{M'}^2:=B\cup (B+y)$, by Lemma~\ref{cutset_lower_bound_lem}, $E$ intersects $A_{M'}^2$ and $|A_{M'}^2|\le 10M'$.

\item $\frac{1}{4}\sum_{j=1}^{d-1} n_j \le M'< 2\prod_{j=1}^{d-1} n_j$. In this case we may find a connected set $B\subseteq \{x\in V[G]\ |\ x_d=0\}$ which contains a path from $x$ to $(a_1, a_2, \ldots, a_{d-1}, 0)$ and such that $\lceil\frac{M'+1}{2}\rceil\le |B|\le 2M'+1$. This set is connected by $2|B|$ disjoint paths to the set $B+(0, \ldots, 0, a_d)$ (the paths are simply straight lines along the last direction, going in both directions around the torus). Letting $A_{M'}^3:=B\cup (B+(0, \ldots, 0, a_d))$, by Lemma~\ref{cutset_lower_bound_lem}, $E$ intersects $A_{M'}^3$ and $|A_{M'}^3|\le 10M'$.

\item $M'\ge 2\prod_{j=1}^{d-1} n_j$. Letting $A_{M'}^4:=P$, the path $A_{M'}^4$ must intersect $E$ by its definition and its length is $\sum_{j=1}^d a_j\le \frac{1}{2}\sum_{j=1}^d n_j\le \prod_{j=1}^{d-1} n_j + n_d \le M+n_d$ (using that $n_j\ge 2$).
\end{enumerate}
Next, for $M'\le M$, let $1\le j(M')\le 4$ be the ``case'' above in which $M'$ is treated. We note that we may choose the $(A_{M'}^{j(M')})_{M'=1}^M$ so that $A_{M''}^{j(M'')}\subseteq A_{M'}^{j(M')}$ whenever $M''\le M'$ and $j(M'')=j(M')$. Hence we may define $A_M:=\cup_{M'=1}^M A_{M'}^{j(M')}$ and have $E\cap A_M\neq\emptyset$ and $|A_M|\le 30M$ if $M<2\prod_{j=1}^{d-1} n_j$ and $|A_M|\le 31M+n_d$ if $M\ge2\prod_{j=1}^{d-1} n_j$, as required.
\end{proof}

{\bf Minimal Edge Cutsets:} For non-empty $X,Y\subseteq V[G]$, let
$\bm{{\MCut(X,Y)}}$ be the set of all minimal edge cutsets
separating $X$ and $Y$. I.e., the set of all $\Gamma\subseteq E[G]$
such that any path from some $x\in X$ to some $y\in Y$ must cross an
edge of $\Gamma$ and any strict subset $\Gamma'\subset \Gamma$ does
not share this property. Note that $\MCut(X,Y)=\MCut(Y,X)$ and that
$\MCut(X,Y)\neq\emptyset$ if and only if $X\cap Y=\emptyset$. For
$x,y\in V[G]$, we shall write $\MCut(x,Y), \MCut(X,y)$ and
$\MCut(x,y)$ instead of $\MCut(\{x\},Y), \MCut(X,\{y\})$ and
$\MCut(\{x\},\{y\})$.

For $\Gamma\in\MCut(X,Y)$ and $v\in V[G]$, define
$\bm{{\comp(\Gamma,v)}}$ to be the connected component of $v$ in $G$
when removing the edges of $\Gamma$, $\bm{P_\Gamma(v)}$ to be the
number of edges in $\Gamma$ incident to $v$ and
$\bm{\Ei(\Gamma,v)}:=\comp(\Gamma,v)\cap\{w\ |\ P_\Gamma(w)>0\}$,
the inner boundary of $\comp(\Gamma,v)$. By definition, for any
$v_1,v_2\in V[G]$ we have that $\comp(\Gamma,v_1)$ and
$\comp(\Gamma,v_2)$ are either disjoint or identical. We define
$\bm{{\subcut(\Gamma, v)}}$ to be all edges between
$\comp(\Gamma,v)$ and its complement. We have
\begin{proposition}\label{subcut_identical_prop}
For any non-empty $X,Y\subseteq V[G]$, $\Gamma\in\MCut(X,Y)$ and $x\in X$ we have $\subcut(\Gamma,x)\subseteq\Gamma$ and $\subcut(\Gamma,x)\in\MCut(x,Y)$. In addition, if $x_1,x_2\in X$ then $\subcut(\Gamma,x_1)$ and $\subcut(\Gamma,x_2)$ are either disjoint or identical.
\end{proposition}
\begin{proof}
Let $\Gamma_x:=\subcut(\Gamma,x)$ and $C_x:=\comp(\Gamma,x)$. By definition of $\Gamma_x$ and $C_x$ we have $\Gamma_x\subseteq \Gamma$. Furthermore, since $\Gamma\in\MCut(X,Y)$, any path from $x$ to a vertex in $Y$ must pass through an edge of $\Gamma_x$. To show that $\Gamma_x$ is minimal, fix $e=\{v,w\}\in\Gamma_x$ with $v\in C_x$. We need to show that there exists a path $P$ from $x$ to some $y\in Y$ whose only intersection with $\subcut(\Gamma,x)$ is at $e$. Since $\Gamma_x\subseteq\Gamma$ and $\Gamma\in\MCut(X,Y)$, there exists $x'\in X$ and a path $P'$ from $x'$ to some $y\in Y$ which only intersects $\Gamma_x$ at $e$. It is not possible that $P'$ crosses $e$ from $w$ to $v$ since by definition of $C_x$, any path from $v$ to some $y\in Y$ must cross $\Gamma_x$ (so $P'$ will have crossed $\Gamma_x$ at least twice). Hence we may take $P$ to be a path from $x$ to $v$ which avoids $\Gamma_x$ and then continues along $P'$ to $y$. This shows $\Gamma_x\in\MCut(x,Y)$.

Now let $x_1,x_2\in X$, $C_{x_1}:=\comp(\Gamma,x_1)$, $C_{x_2}:=\comp(\Gamma,x_2)$, $\Gamma_{x_1}:=\subcut(\Gamma,x_1)$ and $\Gamma_{x_2}:=\subcut(\Gamma,x_2)$. As remarked before the lemma, $C_{x_1}$ and $C_{x_2}$ are either identical or disjoint. If they are identical, then $\Gamma_{x_1}=\Gamma_{x_2}$. We will show that $\Gamma_{x_1}\cap\Gamma_{x_2}\neq \emptyset$ implies $C_{x_1}=C_{x_2}$. Indeed, suppose, to get a contradiction, that $e=\{v,w\}\in(\Gamma_{x_1}\cap \Gamma_{x_2})$, but $C_{x_1}\neq C_{x_2}$. Since $C_{x_1}\cap C_{x_2}=\emptyset$ we have WLOG that $v\in C_{x_1}$ and $w\in C_{x_2}$. But since $\Gamma\in\MCut(X,Y)$, there exists a path $P$ from $x_1$ to some $y\in Y$ intersecting $\Gamma_{x_1}$ only at $e$, and crossing $e$ from $v$ to $w$. Hence we may walk from $x_2$ to $w$ and then along $P$ to $y$ without crossing $\Gamma$ at all, contradicting that $\Gamma\in\MCut(X,Y)$.
\end{proof}
The following proposition puts in a convenient form the simple fact that if a vertex is completely surrounded by a cutset then it forms its own component with respect to it.
\begin{proposition}\label{one_point_cutset_prop}
For any non-empty $X,Y\subseteq V[G]$, $\Gamma\in\MCut(X,Y)$ and $v\in V[G]$ we either have $\Ei(\Gamma,v)=\{v\}$ or $1\le P_\Gamma(w)\le \Delta(G)-1$ for all $w\in \Ei(\Gamma,v)$.
\end{proposition}
\begin{proof}
Let $w\in \Ei(\Gamma,v)$ and note that by definition $P_\Gamma(w)\ge 1$. If $P_\Gamma(w)=\Delta(G)$ we must have $w=v$ since otherwise any path from $w$ to $v$ will cross $\Gamma$ contradicting the fact that $w\in\comp(\Gamma,v)$.
\end{proof}

The next proposition discusses the connectivity properties of cutsets on the torus.
\begin{proposition}\label{full_proj_prop}
For any $x,y\in V[G]$ and $\Gamma\in\MCut(x,y)$, we have that either $\Ei(\Gamma,x)$ has a unique $\GC$-connected component, or each of its $\GC$-connected components has full projection on at least one direction.
\end{proposition}
Here we mean that the projection of $E'\subseteq V[G]$ on direction
$1\le i\le d$ is $\{(v_1,\ldots, v_{i-1},v_{i+1},\ldots, v_d)\ |\
\exists v=(v_1,\ldots, v_{i-1},v_i,v_{i+1},\ldots, v_d)\in E'\}$. We
remark that it seems that $\Ei(\Gamma,x)$ as in the proposition may
have at most $2$ $\GC$-connected components. However, this seems
more difficult to prove and we do not need it in the sequel.
\begin{proof}
Set $\CC=\comp(\Gamma,y)$ and $E=\Ei(\Gamma,x)$. Then
\begin{equation*}
E=\{\text{connected component of $x$ in $V[G]\setminus\CC$}\}\cap\{v\in V[G]\ |\ d_G(v,\CC)=1\},
\end{equation*}
by minimality of $\Gamma$. As in Theorem~\ref{Timar_result}, let
$k=\min\{1\le i\le d\ |\ n_i>2\}$. Let $E'$ be a $\GC$-connected
component of $E$. Suppose that $E'$ does not have full projection on
any direction. Then we can pick $(v_i)_{i=k}^d$ such that $E'$ does
not intersect $P_i(v_i)$ for any $i$. Theorem~\ref{Timar_result}
implies that $E$ is connected in $G^+((P_i(v_i)_{i=k}^d))$, but by
our assumption, the connected component of $E'$ in
$G^+((P_i(v_i)_{i=k}^d))$ is $E'$ itself. Hence $E$ has a unique
$\GC$-connected component.
\end{proof}

The next proposition allows to find a point in each $\GC$-connected component of a cutset with relative ease.
\begin{proposition}\label{point_in_cutset_comp_prop}
For any $x,y\in V[G]$ and integer $M>0$, there exists $A=A(x,y,M)\subseteq V[G]$ with $|A|\le 40Mn_d^{1(M\ge\alpha)}$ such that for any $\Gamma\in\MCut(x,y)$ and $\GC$-connected component $E'$ of $\Ei(\Gamma,x)$ with $|E'|\le M$, we have $E'\cap A\neq\emptyset$.

\end{proposition}
We remark that the proof gives the stronger conclusion that if $|E'|\le M$ then \emph{all} $\GC$-connected components of $\Ei(\Gamma,x)$ intersect $A$, but we shall not need this.
\begin{proof}
Note that $\Ei(\Gamma,x)\in\VCut(x,y)$. We divide into two cases
\begin{enumerate}
\item $M<\alpha$. In this case, we take $A$ to be the set $A(x,y,M)$ of Proposition~\ref{cutset_start_prop}. For any $\Gamma\in\MCut(x,y)$ with $|\Ei(\Gamma,x)|\le M$, by that Proposition, $\Ei(\Gamma,x)\cap A\neq\emptyset$ and $|A|\le 30M$. By Proposition~\ref{full_proj_prop}, $\Ei(\Gamma,x)$ can have at most one $\GC$-connected component since otherwise each of its connected components would have at least $\alpha$ vertices.
\item $M\ge \alpha$. Writing $x=(x_1,\ldots, x_d)$, we set, for $1\le i\le d$, $P_i:=\{(v_1,\ldots, v_d)\ |\ v_j=x_j\text{ for all $j\neq i$}\}$. We then take $A$ to be the union of $A(x,y,M)$ of Proposition~\ref{cutset_start_prop} and $\cup_{i=1}^d P_i$. Note that $|A|\le 31M+n_d+\sum_{i=1}^d n_i\le 32M+2n_d\le 40Mn_d$ (using that $n_i\ge 2$). For any $\Gamma\in\MCut(x,y)$ with $|\Ei(\Gamma,x)|\le M$, we have $\Ei(\Gamma,x)\cap A\neq\emptyset$ by Proposition~\ref{cutset_start_prop}. If $\Ei(\Gamma,x)$ has a unique $\GC$-connected component we are done. Otherwise, by Proposition~\ref{full_proj_prop}, each of its $\GC$-connected components intersects $\cup_{i=1}^d P_i$.\qedhere
\end{enumerate}
\end{proof}

{\bf Odd Minimal Edge Cutsets:} For non-empty sets $X,B\subseteq
V[G]$, we define $\bm{{\OMCut(X,B)}}$, the set of \emph{odd minimal
edge cutsets}, to be those
\begin{equation}\label{OMCut_def}
\text{$\Gamma\in\MCut(X,B)$ satisfying that for any $x\in X$, $\Ei(\Gamma,x)\subseteq V^{\odd}$.}
\end{equation}
Note that that it follows that for any $b\in B$, $\Ei(\Gamma,b)\subseteq V^\even$ and that unlike $\MCut(X,B)$, we generally have $\OMCut(X,B)\neq \OMCut(B,X)$. We remark that ``oddness'' is preserved under taking subcut, that is, if $x\in X$ then $\subcut(\Gamma,x)\in\OMCut(x,B)$ and if $b\in B$ then $\subcut(\Gamma,b)\in\OMCut(X,b)$. This follows simply using Proposition~\ref{subcut_identical_prop}.

Odd minimal cutsets have special properties not shared by the more familiar minimal cutsets (which are not odd) that will be essential to our proofs. Such cutsets arise naturally in our context as follows.
\begin{proposition}\label{level_set_OMCut_prop}
Let $x\in V[G]$, $(B,\mu)$ legal boundary conditions with
non-positive $\mu$ and $f\in\Hom(G,B,\mu)$. If $\LS(f,x,B)\neq
\emptyset$ then $\LS(f,x,B)\in\OMCut(x,B)$.
\end{proposition}
\begin{proof}
By its definition, if $\LS(f,x,B)\neq\emptyset$ then it consists of all edges between a set $\CC\subseteq V[G]$ and its complement where $x\in \CC$ and $B\cap \CC=\emptyset$ (since $\mu$ is non-positive). Hence $\LS(f,x,B)\in\MCut(x,B)$. In addition, by its definition, $f(v)=1$ for all points $v\in \Ei(\LS(f,x,B),x)$. Since our boundary conditions are legal, $\LS(f,x,B)\in\OMCut(x,B)$.
\end{proof}

For non-empty $X,B\subseteq V[G]$ and $\Gamma\in\OMCut(X,B)$, we denote $\bm{E_1(\Gamma)}:=\cup_{x\in X} \Ei(\Gamma,x)$ and $\bm{E_0(\Gamma)}:=\cup_{b\in B} \Ei(\Gamma,b)$.
By definition, $E_1(\Gamma)\subseteq V^\odd$ and $E_0(\Gamma)\subseteq V^\even$. We shall repeatedly use that for $1\le i\le \Delta(G)$,
\begin{equation}\label{parity_neighborhood_contained}
\text{if $v\in \Ei(\Gamma,v)$ and $\{v,v+f_i\}\notin \Gamma$ then $S(v+f_i)\subseteq \comp(\Gamma,v)$}.
\end{equation}
We also define $\bm{E_{1,1}(\Gamma)}:=\{v\in E_1(\Gamma)\ |\
\{v,v+e_1\}\in\Gamma\}$ and $\bm{\Eonee(\Gamma)}:=\{v\in
E_1(\Gamma)\ |\ P_\Gamma(v)\ge \Delta(G)-\sqrt{d}\}$. The letter
``e'' stands for ``exposed'' as vertices in $\Eonee(\Gamma)$ are
exposed to $\Gamma$ from many directions. $E_{1,1}$ and $\Eonee$
will play an important role in the definition of the transformation
$T$ in Section~\ref{T_def_section}. Finally note, following
Proposition~\ref{subcut_identical_prop}, that
$\subcut(\Gamma,x)\in\OMCut(x,B)$ and
$\subcut(\Gamma,b)\in\OMCut(X,b)$ for $x\in X$ and $b\in B$.

For the following propositions, fix non-empty $X,B\subseteq V[G]$ and $\Gamma\in\OMCut(X,B)$. These propositions are generally false for $\MCut$ cutsets. Our first proposition establishes the somewhat surprising property that surrounding every vertex, $\Gamma$ has the same number of edges in every direction.
\begin{proposition}\label{equal_edgenum_in_all_dir}
Setting $E_{v,j}:=\{w\in \Ei(\Gamma,v)\ |\ \{w,w+f_j\}\in\Gamma\}$ for $v\in V[G]$ and $1\le j\le \Delta(G)$, we have $|E_{v,j}|=|E_{v,k}|$ for all $1\le j,k\le \Delta(G)$.
\end{proposition}
\begin{proof}
Set $E_v:=\Ei(\Gamma,v)$ and $C_v:=\comp(\Gamma,v)$. By definition of $\OMCut$, $E_v\subseteq V^\odd$ or $E_v\subseteq V^\even$. Assume WLOG that $E_v\subseteq V^\odd$, then $|\{w\in E_v\ |\ \{w,w+f_j\}\in\Gamma\}|=|C_v\cap V^\odd|-|C_v\cap V^\even|$ since the mapping $w\mapsto w+f_j$ maps points of $(C_v\cap V^\odd)\setminus \{w\in E_v\ |\ \{w,w+f_j\}\in \Gamma\}$ bijectively to $C_v\cap V^\even$. Hence $|E_{v,j}|=|E_{v,k}|$ as required.
\end{proof}
The next proposition shows a connection between the number of $\Gamma$-edges incident to adjacent vertices.
\begin{proposition}\label{neighbor_plaquette_prop}
If $v,w\in V[G]$, $v\adj{G} w$ and $\{v,w\}\in\Gamma$ then
\begin{equation*}
P_\Gamma(v)+P_\Gamma(w)\ge \Delta(G).
\end{equation*}
\end{proposition}
\begin{proof}
If $P_\Gamma(v)=\Delta(G)$ or $P_\Gamma(v)=\Delta(G)-1$, the statement is trivial. Otherwise write $w=v+f_j$ and let $f_{i_1},\ldots, f_{i_{\Delta(G)-P_\Gamma(v)}}$ be such that $\{v,v+f_{i_k}\}\notin\Gamma$ for all $k$. By \eqref{parity_neighborhood_contained}, $v+f_{i_k}+f_j\in \comp(\Gamma,v)$. Since $w$ is adjacent to $(v+f_{i_k}+f_j)_{k=1}^{\Delta(G)-P_\Gamma(v)}$, it follows that $P_\Gamma(w)\ge\Delta(G)-P_\Gamma(v)$.
\end{proof}
A similar property holds for interior vertices of the components $\comp(\Gamma,v)$, as follows.
\begin{proposition}\label{interior_P_gamma_estimate}
For $u,v\in V[G]$, $v\adj{G} u$ and $\{v,u\}\notin\Gamma$ we have $|\{v'\in \Ei(\Gamma,u)\ |\ v'\adj{G} u\}|\ge P_\Gamma(v)$.
\end{proposition}
\begin{proof}
If $P_\Gamma(v)=0$, the claim is trivial. Otherwise, note that $v\in \Ei(\Gamma,u)$ and hence by \eqref{parity_neighborhood_contained}, $u+f_i\in\comp(\Gamma,u)$ for all $i$. Let $f_{i_1}, \ldots, f_{i_{P_\Gamma(v)}}$ be such that $\{v,v+f_{i_k}\}\in\Gamma$. We deduce that for all $k$, $u+f_{i_k}\in \Ei(\Gamma,u)$ since it is adjacent to $v+f_{i_k}$.
\end{proof}

Based on $\Gamma$, we define another graph structure on $V[G]$ which
is a subgraph of $\GC$. We say that $v,v'\in V[G]$ are
$\Gamma$-adjacent,denoted $\bm{v\adj{\Gamma}v'}$, if $v'=v+f_i+f_j$
for some $1\le i,j\le\Delta(G)$ such that $i\neq j$, $f_i\neq-f_j$,
$\{v,v+f_i\}\in\Gamma$ and $\{v,v+f_j\}\notin\Gamma$. Note that if
$v\adj{\Gamma} v'$ then necessarily $v,v'\in \Ei(\Gamma,v)$ ($v\in
\Ei(\Gamma,v)$ since $\{v,v+f_i\}\in\Gamma$ and $v'\in
\Ei(\Gamma,v)$ by \eqref{parity_neighborhood_contained} and since
$v+f_i\notin\comp(\Gamma,v)$). We have
\begin{proposition}\label{Gamma_1_degree_prop}
Each $v\in V[G]$ is $\Gamma$-adjacent to at least
\begin{equation}\label{min_Gamma_neighbors}
P_\Gamma(v)(\Delta(G)-P_\Gamma(v))-\min(P_\Gamma(v),\Delta(G)-P_\Gamma(v))
\end{equation}
$v'\in V[G]$. In particular, if $P_\Gamma(v)\notin\{0,\Delta(G)\}$ then $v$ has at least $d-2$ $\Gamma$-neighbors.
\end{proposition}
\begin{proof}
If $P_\Gamma(v)\in\{0,\Delta(G)\}$ the claim is trivial. Otherwise, let $f_{i_1}, \ldots, f_{i_{P_\Gamma(v)}}$ be the directions such that $\{v,v+f_{i_k}\}\in\Gamma$ and let $f_{j_1},\ldots, f_{j_{\Delta(G)-P(v)}}$ be the other directions. Then every $v'$ of the form $v'=v+f_{i_k}+f_{j_m}$ where $f_{i_k}\neq -f_{j_m}$ is a $\Gamma$-neighbor of $v$ and there are at least $P_\Gamma(v)(\Delta(G)-P_\Gamma(v))-\min(P_\Gamma(v),\Delta(G)-P_\Gamma(v))$ such choices.
The second part of the proposition follows by noting that \eqref{min_Gamma_neighbors} is minimized at $P_\Gamma(v)=1$ over $P_\Gamma(v)\in[1,\Delta(G)-1]$, giving $\Delta(G)-2\ge d-2$.
\end{proof}

Next, fix $x,b\in V[G]$. We say that $\Gamma\in\OMCut(x,b)$ is \textbf{trivial} if $\Gamma$ consists only of the edges incident to $x$ or only of the edges incident to $b$. If $\Gamma$ is trivial then $|\Gamma|=\Delta(G)$. The next proposition gives some properties of non-trivial $\Gamma$ and shows in particular that they must have many more edges than trivial ones.
\begin{proposition}\label{trivial_gamma_prop}
For $\Gamma\in\OMCut(x,b)$ and dimension $d>2$, the following are equivalent:
\begin{enumerate}
\item $\Gamma$ is non-trivial.
\item For all $v\in V[G]$, $P_\Gamma(v)\le\Delta(G)-1$.
\item $|\Gamma|\ge \frac{\Delta(G)^2}{2}$.
\end{enumerate}
\end{proposition}
Note that the third item does not necessarily hold for $\Gamma\in\MCut(x,b)$ since we may have that $\Gamma$ is all edges surrounding $x$ and one of its neighbors.
\begin{proof}
For a trivial $\Gamma$ it is clear that none of the properties hold (since $\frac{\Delta(G)^2}{2}>\Delta(G)$ when $d>2$). Suppose now that $\Gamma$ is non-trivial. If there exists $v\in V[G]$ with $P_\Gamma(v)=\Delta(G)$ then we would have to have $v\in\{x,b\}$ by minimality of $\Gamma$ and then $\Gamma$ would be trivial, again by minimality.

Next, we claim that there exists $v\in V[G]$ with $\frac{\Delta(G)}{2}\le P_\Gamma(v)\le \Delta(G)-1$. Indeed, there exists $w\in V[G]$ with $1\le P_\Gamma(w)\le \Delta(G)-1$. If $P_\Gamma(w)<\frac{\Delta(G)}{2}$ then by Proposition~\ref{neighbor_plaquette_prop} and the previous characterization of non-trivial $\Gamma$, any neighbor $v\adj{G} w$ with $\{v,w\}\in\Gamma$ satisfies $\frac{\Delta(G)}{2}\le P_\Gamma(v)\le \Delta(G)-1$. Fix such a $v$, let $1\le i\le \Delta(G)$ be such that $\{v,v-f_i\}\notin \Gamma$ and let $j_1,\ldots, j_{\left\lceil\frac{\Delta(G)}{2}\right\rceil}$ be such that $\{v,v+f_{j_k}\}\in\Gamma$ for all $k$ (here, we allow $f_{j_k}=f_i$ for some $k$). We have $v+f_{j_k}\notin\comp(\Gamma,v)$ and, by \eqref{parity_neighborhood_contained}, $v-f_i+f_{j_k}\in\comp(\Gamma,v)$ for all $k$. Finally, recalling the definition of $E_{v,i}$ from Proposition~\ref{equal_edgenum_in_all_dir}, it follows that $v-f_i+f_{j_k}\in E_{v,i}$ for all $k$ and hence $|E_{v,i}|\ge \frac{\Delta(G)}{2}$ so that by Proposition~\ref{equal_edgenum_in_all_dir}, $|\Gamma|\ge\Delta(G)|E_{v,i}|\ge \frac{\Delta(G)^2}{2}$.
\end{proof}
\begin{remark}\label{trivial_gamma_prop_remark}
The proof in fact shows that in all dimensions we have that a $\Gamma\in\OMCut(x,b)$ is either trivial or has $|\Gamma|\ge\frac{\Delta(G)^2}{2}$. The assumption $d>2$ is only needed so that these two properties cannot coexist.
\end{remark}

{\bf Combinatorics:} We shall need the following basic counting result.
\begin{proposition}\label{L_decomp_prop}
Given integers $s_1, s_2, L>0$ with $s_2>s_1$, the number of solutions in integers $k$ and $(x_m)_{m=1}^k$ to
\begin{equation}\label{L_decomp}
\sum_{m=1}^k x_m=L
\end{equation}
with each $x_m$ satisfying either $x_m=s_1$ or $x_m\ge s_2$ is at most
\begin{equation*}
\exp\left(\frac{6L\log{s_2}}{s_2}\right).
\end{equation*}
\end{proposition}
\begin{proof}
Suppose that in the sum in \eqref{L_decomp} there are exactly $k_2$ factors of size at least $s_2$ and denote them, in order of appearance in the sum, by $(y_m)_{m=1}^{k_2}$. As
\begin{equation}\label{L_decomp2}
\sum_{m=1}^{k_2} (y_m-s_2)\le L - k_2s_2,
\end{equation}
it follows from standard combinatorial enumeration that the number of possibilities for $(y_m)_{m=1}^{k_2}$, given $k_2$, is at most $\binom{L-k_2(s_2-1)}{k_2}$. In addition, suppose that in \eqref{L_decomp} there are exactly $k_1$ factors $x_m$ of size equal to $s_1$ and note that $k_1$ can be determined from $k_2$ and $(y_m)_{m=1}^{k_2}$. Thus, given $k_2$ and $(y_m)_{m=1}^{k_2}$, the solution $(x_m)_{m=1}^k$ to \eqref{L_decomp} is determined by the choice of which of the $k_1+k_2$ factors are the $k_2$ factors corresponding to the $(y_m)$. As $k_1+k_2\le L$ we see that we have at most
\begin{equation*}
\binom{L-k_2(s_2-1)}{k_2}\binom{L}{k_2}\le \binom{L}{k_2}^2
\end{equation*}
solutions to \eqref{L_decomp} with a given $k_2$. Since $k_2\le \frac{L}{s_2}$ we see that \eqref{L_decomp} has at most
\begin{equation*}
\sum_{i=0}^{\lfloor L/s_2\rfloor} \binom{L}{i}^2\le e^{\frac{6L\log{s_2}}{s_2}}
\end{equation*}
solutions where we used that $\sum_{i=0}^n \binom{L}{i}\le r^{-n}(1+r)^L\le e^{rL-n\log r}$ for $r\le 1$ and then substituted $n=\left\lfloor\frac{L}{s_2}\right\rfloor$, $r=\frac{1}{s_2}$ and squared.
\end{proof}

\section{Proof of Level Set Theorem}\label{level_set_thm_sec}
In this section we prove theorem~\ref{main_thm}.
\subsection{Reduction to an Expanding Transformation}\label{expanding_transformation_sec}

Our probabilistic estimates are all based on the idea of an expanding transformation (as explained in the proof sketch). For an $\Omega\subseteq\Hom(G,B,\mu)$ (for some legal boundary condition $(B,\mu)$), we shall find a transformation $T:\Omega\to\PP(\Hom(G,B,\mu))$, i.e., a transformation taking $f\in\Omega$ to a subset of $\Hom(G,B,\mu)$.
With a slight abuse of notation we denote $T(\Omega):=\cup_{f\in\Omega} T(f)$. We have the following simple
\begin{lemma}\label{transformation_lemma}
Let $(B,\mu)$ be a legal boundary condition, $\Omega\subseteq\Hom(G,B,\mu)$ and $T:\Omega\to\PP(\Hom(G,B,\mu))$. If $f\unifin\Hom(G,B,\mu)$ then
\begin{equation*}
\P(f\in\Omega)=\frac{|\Omega|}{|T(\Omega)|}\P(f\in T(\Omega)).
\end{equation*}
In particular, $\P(f\in\Omega)\le \frac{|\Omega|}{|T(\Omega)|}$.
\end{lemma}
\begin{proof}
By definition,
\begin{equation*}
\P(f\in\Omega) = \frac{|\Omega|}{|\Hom(G,B,\mu)|} = \frac{|\Omega|}{|T(\Omega)|}\cdot\frac{|T(\Omega)|}{|\Hom(G,B,\mu)|}= \frac{|\Omega|}{|T(\Omega)|}\P(f\in T(\Omega)).\qedhere
\end{equation*}
\end{proof}
The previous lemma is of course true also when the set $T(\Omega)$ is replaced by an arbitrary $\Omega'\subseteq\Hom(G,B,\mu)$, however, we wish to emphasize the role of the transformation $T$ since our main use of the lemma will be through it.
\begin{theorem}\label{transformation_thm}
There exist $d_0\in\N$, $c>0$ such that for all $d\ge d_0$, non-linear tori $G$, legal boundary conditions $(B,\mu)$ with non-positive $\mu$, $x\in V[G]$ and integer $L\ge 1$, there exists $T:\Omega_{x,L}\to\PP(\Hom(G,B,\mu))$ satisfying
\begin{enumerate}
\item For all $\emptyset\neq \Omega\subseteq\Omega_{x,L}$ we have
\begin{equation*}
\frac{|\Omega|}{|T(\Omega)|}\le d^3\exp\left(-\frac{cL}{d\log^2 d}\right).
\end{equation*}
\item For all $k\ge 2$, $x_1,\ldots, x_{k-1}\in V[G]$ and integers $L_1,\ldots, L_{k-1}\ge 1$ we have $T(\Omega_{(x_1,\ldots, x_{k-1},x),(L_1,\ldots, L_{k-1},L)})\subseteq \Omega_{(x_1,\ldots, x_{k-1}),(L_1,\ldots, L_{k-1})}$.
\end{enumerate}
\end{theorem}
Note that by definition $\Omega_{(x_1,\ldots, x_{k-1},x),(L_1,\ldots, L_{k-1},L)}\subseteq\Omega_{x,L}$ so that the second part of the theorem makes sense. Theorem~\ref{main_thm} follows immediately from this theorem and the previous lemma, as follows.
\begin{proof}[Proof of Theorem~\ref{main_thm}]
Let $d_0$ and $c>0$ be the numbers from Theorem~\ref{transformation_thm} and fix $d\ge d_0$, non-linear tori $G$, legal boundary conditions $(B,\mu)$ with non-positive $\mu$. Let $k\ge 1$, $x_1,\ldots, x_k\in V[G]$ and integers $L_1,\ldots, L_k\ge 1$. Taking the transformation $T:\Omega_{x_k,L_k}\to\PP(\Hom(G,B,\mu))$ given by Theorem~\ref{transformation_thm}, we obtain using Lemma~\ref{transformation_lemma} and both parts of Theorem~\ref{transformation_thm} that
\begin{align*}
\P(f\in\Omega_{(x_1,\ldots, x_k),(L_1,\ldots, L_k)}&) = \frac{|\Omega_{(x_1,\ldots, x_k),(L_1,\ldots, L_k)}|}{|T(\Omega_{(x_1,\ldots, x_k),(L_1,\ldots, L_k)})|}\P(f\in T(\Omega_{(x_1,\ldots, x_k),(L_1,\ldots, L_k)}))\le\\
&\le d^3\exp\left(-\frac{cL_k}{d\log^2 d}\right)\P(f\in\Omega_{(x_1,\ldots, x_{k-1}),(L_1,\ldots, L_{k-1})}),
\end{align*}
where we interpret $\frac{0}{0}$ as $0$ and, for $k=1$, define $\Omega_{\emptyset,\emptyset}:=\Hom(G,B,\mu)$. By induction on $k$ we have
\begin{equation*}
\begin{split}
\P(f\in\Omega_{(x_1,\ldots, x_k),(L_1,\ldots, L_k)})&\le \min\left(d^{3k}\exp\left(-\frac{c\sum_{i=1}^k L_i}{d\log^2 d}\right),1\right)\le\\
&\le d^k\exp\left(-\frac{c'\sum_{i=1}^k L_i}{d\log^2 d}\right)
\end{split}
\end{equation*}
for some $c'>0$, as required.\qedhere
\end{proof}
Hence all our efforts will be concentrated towards proving Theorem~\ref{transformation_thm}. In the next section we define the transformation $T$ and show why it satisfies the second part of Theorem~\ref{transformation_thm}. Section~\ref{odd_cutsets_structure_sec} develops the structural results on odd cutsets we shall need for the proof of the first part of the theorem, which is subsequently proved in Section~\ref{proof_of_transform_theorem_sec}.

\subsection{Definition of the Transformation}\label{T_def_section}
In this section we define the transformation $T$ of Theorem~\ref{transformation_thm}, establish some of its basic properties and prove that it satisfies the second property in Theorem~\ref{transformation_thm}. Fix a torus $G$ (for some dimension $d$ and any even side lengths $n_i$ satisfying \eqref{torus_side_lengths}), legal boundary conditions $(B,\mu)$ with non-positive $\mu$, $x\in V[G]$ and integer $L\ge 1$.

Throughout the section we denote, for $f\in\Omega_{x,L}$, $\Gamma:=\LS(f,x,B)$ (note that $\Gamma\in\OMCut(x,B)$ by Proposition~\ref{level_set_OMCut_prop}), $\CC_1:=\comp(\Gamma,x)$, $E_1:=E_1(\Gamma)$, $E_0:=E_0(\Gamma)$, $E_{1,1}:=E_{1,1}(\Gamma)$ and $\Eonee:=\Eonee(\Gamma)$. We note especially that
\begin{equation}\label{Gamma_boundary_values}
f(v)=j\text{ for $j\in\{0,1\}$ and $v\in E_j$}.
\end{equation}
The transformation $T$ will take one of two possible forms, which we now describe.

\subsubsection{The Shift Transformation}\label{shift_transformation_sec}
We define the ``shift transformation'' $\shift:\Omega_{x,L}\to\Hom(G,B,\mu)$ by
\begin{equation*}
\shift(f)(v) = \begin{cases}f(v+e_1)-1 &\text{for }v\in \CC_1\\f(v) &\text{otherwise}\end{cases}.
\end{equation*}
\begin{lemma}
We indeed have $\shift(f)\in\Hom(G,B,\mu)$.
\end{lemma}
\begin{proof}
Since $\Gamma\in\OMCut(x,B)$ we have $B\cap \CC_1=\emptyset$. It follows that $\shift(f)(b)=\mu(b)$ for all $b\in B$. Now fix $v\in G$ and $1\le i\le \Delta(G)$. It remains to check that $|\shift(f)(v)-\shift(f)(v+f_i)|=1$. If $v,v+f_i\in \CC_1$ or $v,v+f_i\notin \CC_1$, this follows from the corresponding property of $f$ (using that $(v+e_1)+f_i=(v+f_i)+e_1$ in $G$). Otherwise, assume WLOG that $v\in \CC_1$ and $v+f_i\notin \CC_1$. It follows from \eqref{Gamma_boundary_values} that $f(v)=1$ and $f(v+f_i)=0$. Hence $f(v+e_1)\in\{0,2\}$ and we have $|\shift(f)(v)-\shift(f)(v+f_i)|=|f(v+e_1)-1|=1$.
\end{proof}
The following lemma is key to our definitions.
\begin{lemma}\label{flip_possible_lemma}
For all $v\in E_{1,1}$ and $1\le i\le \Delta(G)$ we have $\shift(f)(v+f_i)=0$.
\end{lemma}
\begin{proof}
Let $v\in E_{1,1}$ and $1\le i\le \Delta(G)$. By definition, $v+e_1\in E_0$. If $v+f_i\in E_0$ then $\shift(f)(v+f_i)=f(v+f_i)=0$ by \eqref{Gamma_boundary_values}. If $v+f_i\notin E_0$ then by \eqref{parity_neighborhood_contained}, $v+f_i+e_1\in E_1$ (since it is adjacent to $v+e_1$) implying that $\shift(f)(v+f_i)=f(v+f_i+e_1)-1=0$ by \eqref{Gamma_boundary_values}.
\end{proof}
We continue to define the transformation $T_1:\Omega_{x,L}\to\PP(\Hom(G,B,\mu))$. $T_1(f)$ is the set of all functions $g$ of the form
\begin{equation*}
g(v)=\begin{cases}\shift(f)(v)&v\notin E_{1,1}\\\eps_v&\text{otherwise}\end{cases}
\end{equation*}
where $\{\eps_v\}_{v\in E_{1,1}}$ is a sequence of $\pm 1$. The
previous lemma shows that these $2^{|E_{1,1}|}$ functions are indeed
a subset of $\Hom(G,B,\mu)$. Since we wish to define a
transformation $T$ with $|T(f)|$ large, one may wonder if $|T_1(f)|$
can be increased by shifting in a direction other than $e_1$ in the
definition of $\shift$. However, by
Proposition~\ref{equal_edgenum_in_all_dir} we have
\begin{equation*}
|E_{1,1}|=|\{v\in E_1\ |\ \{v,v+f_i\}\in \Gamma\}|
\end{equation*}
for all $1\le i\le \Delta(G)$. It follows that $|E_{1,1}|=\frac{L}{\Delta(G)}$ and consequently $|T_1(f)|=2^{\frac{L}{\Delta(G)}}$.

\subsubsection{The Shift+Flip Transformation}\label{shift_flip_transformation_sec}
We now define the transformation $T_2:\Omega_{x,L}\to\PP(\Hom(G,B,\mu))$ as follows. Let $g\in T_1(f)$. By definition of $T_1$, we know that $g(v)\in\{-1,1\}$ for all $v\in E_1$ (since $g(v)=0$ for all $v\in E_0$). For $v\in \Eonee$, we let $R_v$ be the connected component of $v$ in $V[G]\setminus \{w\in V[G]\ |\ g(w)=0\}$. We note that it may happen that $R_v=R_w$ for $v\neq w$, but then we must have $g(v)=g(w)$ since otherwise any path between them will cross a zero of $g$. We also note that $R_v\subseteq \CC_1$ for all $v\in \Eonee$ since $g(w)=0$ for all $w\in E_0$. Finally, we define $T_2(f)$ to be all functions $\tilde{g}$ formed by taking a $g\in T_1(f)$ and defining
\begin{equation}\label{flip_T_def}
\tilde{g}(w):=\begin{cases}-g(w)&\text{if $w\in R_v$ for some $v\in \Eonee$ with $g(v)=-1$}\\g(w)&\text{otherwise}\end{cases}.
\end{equation}
Less formally, $\tilde{g}$ is formed from $g$ by flipping some values to ensure that $\tilde{g}(v)=1$ for all $v\in \Eonee$. By our definition of $R_v$ and since $R_v\subseteq \CC_1$, it follows that $\tilde{g}\in \Hom(G,B,\mu)$ in a straightforward manner. Comparing the definitions of $T_1$ and $T_2$, we see that $|T_2(f)|=2^{\frac{L}{\Delta(G)}-|E_{1,1}\cap \Eonee|}=2^{|E_{1,1}\setminus \Eonee|}$ since by Lemma~\ref{flip_possible_lemma}, $R_v=\{v\}$ for $v\in E_{1,1}\cap \Eonee$.

\subsubsection{The Transformation $T$}
We are now ready to define the transformation $T$.
\begin{equation}\label{T_def}
T(f):=\begin{cases}T_1(f)&\text{if }|\Eonee|\ge (1-\frac{\lambda}{\log^2 d})\frac{L}{\Delta(G)}\\T_2(f)&\text{otherwise}\end{cases}.
\end{equation}
for some small enough constant $\lambda$ (independent of $d$) to be determined later (in Section~\ref{proof_of_transform_theorem_sec}). From our previous discussion, we have
\begin{equation}\label{T_image_size_eq}
|T(f)|=\begin{cases}2^{\frac{L}{\Delta(G)}}&\text{if }|\Eonee|\ge (1-\frac{\lambda}{\log^2 d})\frac{L}{\Delta(G)}\\2^{\frac{L}{\Delta(G)}-|E_{1,1}\cap \Eonee|}&\text{otherwise}\end{cases},
\end{equation}
and also that
\begin{equation}\label{only_Cx_changes}
g(v)=f(v)\text{ for all $g\in T(f)$ and $v\notin \CC_1$}.
\end{equation}

As promised, we now show that the second property of
Theorem~\ref{transformation_thm} holds for this transformation.
\begin{proof}[Proof of second property in Theorem~\ref{transformation_thm}.]
Fix $k\ge 2$, $x_1,\ldots, x_{k-1}\in V[G]$ and integers $L_1,\ldots, L_{k-1}\ge 1$ and assume that $f\in\Omega_{(x_1,\ldots, x_{k-1},x),(L_1,\ldots, L_{k-1},L)}$. We need to show that $T(f)\subseteq\Omega_{(x_1,\ldots, x_{k-1}),(L_1,\ldots, L_{k-1})}$. Fix $g\in T(f)$ and $1\le i\le k-1$. It is sufficient to show that $\LS(f,x_i,B)=\LS(g,x_i,B)$. Let $C_x:=\comp(\LS(f,x,B),x)$. We shall need only that by \eqref{only_Cx_changes}, $g(v)=f(v)$ for all $v\notin C_x$. Let $A$ and $A'$ be the union of those connected components which contain points of $B$ in $\{v\in V[G]\ |\ f(v)\le 0\}$ and in $\{v\in V[G]\ |\ g(v)\le 0\}$ respectively. Let $C_{x_i}$ and $C_{x_i}'$ be the connected components of $x_i$ in $V[G]\setminus A$ and $V[G]\setminus A'$ respectively. The claim will follow once we show that $C_{x_i}=C_{x_i}'$. Note first that since $C_x$ is the connected component of $x$ in $V[G]\setminus A$, $C_x$ and $C_{x_i}$ must be identical or disjoint. But by definition of $\Omega_{(x_1,\ldots, x_{k-1},x),(L_1,\ldots, L_{k-1},L)}$, it follows that $C_x\cap C_{x_i}=\emptyset$. Next, let $E_1^{x_i}:=\{v\in C_{x_i}\ |\ \exists w\notin C_{x_i}, w\adj{G} v\}$. By our definitions, $f(v)=1$ for all $v\in E_1^{x_i}$ and hence also $g(v)=1$ by \eqref{only_Cx_changes}. It follows that $A'\cap C_{x_i}=\emptyset$ and hence $C_{x_i}\subseteq C_{x_i}'$. To see the opposite inequality, note that a point in $A$ is characterized by having a path connecting it to some $b\in B$ which avoids $C_x$ and $\{v\in V[G]\ |\ f(v)\ge 1\}$. This same path shows that point is also in $A'$ and hence $A\subseteq A'$ so that $C_{x_i}\supseteq C_{x_i}'$.\qedhere

\end{proof}

\subsection{Structure Theorems for Odd Cutsets}\label{odd_cutsets_structure_sec}

In this section we shall prove several theorems estimating the number of odd minimal cutsets in various settings. In Section~\ref{rough_bdry_counting_sec} we estimate the number of such cutsets in terms of their boundary roughness. In Section~\ref{partial_edge_cutsets_section} we show that if one is content with finding only an approximation to the cutset,
identifying clearly only vertices whose $P_\Gamma(v)$ is less than $\Delta(G)-\sqrt{d}$, then one can find a relatively small set of approximations, containing such an approximation to every cutset.
This is used in that section to bound the number of ``possible level sets'' for a function $f$ given a function $g\in T_2(f)$ ($T_2$ is defined in Section~\ref{shift_flip_transformation_sec}).

\subsubsection{Counting Cutsets With Rough Boundary}\label{rough_bdry_counting_sec}
To state the main theorem of this section, fix $B\subseteq V[G]$, $x\in V[G]\setminus B$ and for a cutset $\Gamma\in\OMCut(x,B)$, $v\in V[G]$ and subset $E\subseteq V[G]$ define
\begin{equation}\label{R_Gamma_def}
\begin{split}
R_\Gamma(v)&:=\min(P_\Gamma(v), \Delta(G)-P_\Gamma(v)),\\
R_\Gamma(E)&:=\sum_{v\in E} R_\Gamma(v).
\end{split}
\end{equation}
A value of $\frac{R_\Gamma(E_1(\Gamma))}{|E_1(\Gamma)|}$ significantly smaller than $d$ indicates some roughness of $E_1(\Gamma)$. Our theorem will allow us to estimate the number of cutsets having such roughness. For integers $M,R\ge0$, let
\begin{equation*}
 \OMCut(x,B,M,R):=\{\Gamma\in\OMCut(x,B)\ |\ |E_1(\Gamma)|=M, R_\Gamma(E_1(\Gamma))=R\}.
\end{equation*}
Recalling from \eqref{alpha_def} that $\alpha=\prod_{i=1}^{d-1} n_i$, we will prove
\begin{theorem}\label{count_cutsets_thm}
There exist $C,d_0>0$ such that for all $d\ge d_0$ and integers $M,R\ge 0$,
\begin{equation*}
 |\OMCut(x,B,M,R)|\le n_d^{\left\lfloor\frac{M}{\alpha}\right\rfloor} \exp\left(\frac{C\log^2 d}{d}R\right).
\end{equation*}
\end{theorem}
For $\Gamma\in\OMCut(x,B)$ and a $\GC$-connected component $E$ of $E_1(\Gamma)$, we say that $E$ is associated with $b\in B$ if $E\cap E_1(\subcut(\Gamma,b))\neq \emptyset$, that is, if the part of $\Gamma$ which separates $b$ and $x$ has an edge incident to $E$. Note that $E$ may be associated to several $b\in B$. The following proposition is the main step in proving the above theorem.
\begin{proposition}\label{count_components_prop}
There exist $C,d_0>0$ such that for $d\ge d_0$, integers $M,R\ge 0$ and $b\in B$, the number of possibilities for a $\GC$-connected component $E$, associated with $b$ and having $|E|=M$ and $R_\Gamma(E)=R$, of $E_1(\Gamma)$ for some $\Gamma\in\OMCut(x,B)$ is at most
\begin{equation*}
n_d^{\left\lfloor\frac{M}{\alpha}\right\rfloor} \exp\left(\frac{C\log^2 d}{d}R\right).
\end{equation*}
\end{proposition}
We emphasize that in the above proposition, $\Gamma$ is not given. We are estimating the number of possibilities for $E$ from all possible $\Gamma$'s.

We note the following simple lemma for later reference.
\begin{lemma}\label{component_triviality_lemma}
 For $\Gamma\in\OMCut(x,B)$, either $E_1(\Gamma)=\{x\}$, in which case $|E_1(\Gamma)|=1$ and $R_\Gamma(E_1(\Gamma))=0$, or all $\GC$-connected components $E$ of $E_1(\Gamma)$ have $R_\Gamma(E)\ge |E|\ge d-1$.
\end{lemma}
\begin{proof}
 Let $\Gamma\in\OMCut(x,B)$ and $E$ a $\GC$-connected component of $E_1(\Gamma)$. First, if $E=\{x\}$ then $E_1(\Gamma)=\{x\}$ by Propositions~\ref{one_point_cutset_prop} and \ref{Gamma_1_degree_prop}. Second, the same propositions imply that if $E\neq\{x\}$ then $R_\Gamma(E)\ge |E|\ge d-1$.
\end{proof}
\paragraph{Proof of Proposition~\ref{count_components_prop}}

Let $\Gamma\in\OMCut(x,B)$ and $E$ a $\GC$-connected component of $E_1(\Gamma)$. Assume
\begin{equation*}
E\neq\{x\}.
\end{equation*}
The next proposition shows that $E$ is ``dominated'' by a small subset of it.
\begin{proposition}\label{dom_set_prop}
There exists $d_0>0$, independent of $\Gamma$ and $E$, such that for all $d\ge d_0$ there exists $E^{\operatorname{t}}\subseteq E$ with the properties:
\begin{enumerate}
\item $|E^{\operatorname{t}}|\le 10\frac{\log d}{d}|E|$ and $R_\Gamma(E^{\operatorname{t}})\le 10\frac{\log d}{d}R_\Gamma(E)$.
\item For every $v\in E$, either $v\in E^{\operatorname{t}}$ or there exists $v'\in E^{\operatorname{t}}$ such that $v'\adj{\Gamma}v$ (in other words, $E^{\operatorname{t}}$ is a $\Gamma$-dominating set for $E$).
\end{enumerate}
\end{proposition}
\begin{proof}
Choose a subset $E^{\operatorname{s}}$ randomly by adding each $v\in E$ to it independently with probability $3\frac{\log d}{d}$ (assuming $d$ is large enough so that this probability is at most 1). We have $\E|E^{\operatorname{s}}| = 3\frac{\log d}{d}|E|$ and $\E R_\Gamma(E^{\operatorname{s}}) = 3\frac{\log d}{d} R_\Gamma(E)$ so that by Markov's inequality
\begin{align}
&\P\left(|E^{\operatorname{s}}|\ge 9\frac{\log d}{d}|E|\right)\le \frac{1}{3}\qquad \text{and}\label{Markov1}\\
&\P\left(R_\Gamma(E^{\operatorname{s}})\ge 9\frac{\log d}{d}R_\Gamma(E)\right)\le \frac{1}{3}\label{Markov2}.
\end{align}
Let $E^{\text{nd}}\subseteq E$ be those vertices which are not $\Gamma$-dominated by $E^{\operatorname{s}}$. That is, vertices in $E$ such that they and their $\Gamma$-neighbors are not in $E^{\operatorname{s}}$. Using the assumption $E\neq\{x\}$, by Propositions~\ref{one_point_cutset_prop} and \ref{Gamma_1_degree_prop}, the minimal $\Gamma$-degree of vertices in $E$ is at least $d-2$. Hence the probability that some vertex is in $E^{\text{nd}}$ is at most $\left(1-3\frac{\log d}{d}\right)^{d-1}$ implying
\begin{equation*}
\E |E^{\text{nd}}| \le \left(1-3\frac{\log d}{d}\right)^{d-1}|E|\le e^{-\frac{3(d-1)\log d}{d}}|E| < \frac{|E|}{d^{5/2}}
\end{equation*}
for large enough $d$. By Markov's inequality,
\begin{equation}\label{Markov3}
\P\left(|E^{\text{nd}}|\ge \frac{|E|}{d^2}\right)<\frac{1}{3}
\end{equation}
for large enough $d$. Finally, we take $E^{\operatorname{t}}:=E^{\operatorname{s}}\cup E^{\text{nd}}$. Noting that $R_\Gamma(E^{\text{nd}})\le d|E^{\text{nd}}|$ and putting together \eqref{Markov1},\eqref{Markov2} and \eqref{Markov3}, we see that $E^{\operatorname{t}}$ satisfies the requirements of the proposition with positive probability (and in particular, such a set exists).\qedhere
\end{proof}

\begin{lemma}\label{Et_connected_graph_lem}
If $E^{\operatorname{t}}\subseteq E$ is as in Proposition~\ref{dom_set_prop}, then $E^{\operatorname{t}}$ is connected in $G^{\otimes 6}$.
\end{lemma}
\begin{proof}
Fix $v_s,v_t\in E^{\operatorname{t}}$ and let $v_s=v_1,v_2,\ldots, v_m=v_t$ be a $\GC$-path between them of vertices of $E$. By Proposition~\ref{dom_set_prop}, we may find for each $2\le i\le m-1$, a vertex $v_i'\in E^{\operatorname{t}}$ such that $d_G(v_i',v_i)\le 2$. We also take $v_1'=v_s$ and $v_m'=v_t$. It follows that for each $1\le i\le m-1$
\begin{equation*}
d_G(v_i', v_{i+1}')\le d_G(v_i', v_i) + d_G(v_i, v_{i+1}) + d_G(v_{i+1},v_{i+1}')\le 6.
\end{equation*}
Hence $v_s=v_1',v_2',\ldots,v_m'=v_t$ is a $G^{\otimes 6}$-walk, proving the lemma.
\end{proof}
We continue by defining for each $v\in E$, a vector $N_\Gamma(v)\in\{0,1\}^{\Delta(G)}$ by $N_\Gamma(v)_i=1_{(v+f_i\in\comp(\Gamma,x))}$ and letting $N_\Gamma(E^{\operatorname{t}}):=(N_\Gamma(v))_{v\in E^{\operatorname{t}}}$. We have
\begin{lemma}\label{Et_and_Nv_suffice}
The set $E^{\operatorname{t}}$ and the vector $N_\Gamma(E^{\operatorname{t}})$ uniquely determine $E$ among all $\GC$-connected components of $E_1(\Gamma)$ for all $\Gamma\in\OMCut(x,B)$.
\end{lemma}
We emphasize that what we mean in the lemma is that if we are not given $\Gamma$ or $E$, but instead are only given $E^{\operatorname{t}}$ and $N_\Gamma(E^{\operatorname{t}})$ corresponding to some $\Gamma$ and $E$ ($E^{\operatorname{t}}\subseteq E$ is as in Proposition~\ref{dom_set_prop}), then we may reconstruct $E$. In other words, there is a function satisfying that for every $\Gamma\in\OMCut(x,B)$, $E$ a $\GC$-connected component of $E_1(\Gamma)$ and $E^{\operatorname{t}}\subseteq E$ a subset as in Proposition~\ref{dom_set_prop}, the function takes $E^{\operatorname{t}}$ and $N_\Gamma(E^{\operatorname{t}})$ and returns $E$.
\begin{proof}
By property (2) of Proposition~\ref{dom_set_prop} and since $E$ is $\GC$-connected, $E$ equals the set of $v\in V[G]$ satisfying that either $v\in E^{\operatorname{t}}$ or there exists a $v'\in E^{\operatorname{t}}$ such that $v\adj{\Gamma} v'$ (such $v$ are necessarily in $E$ as noted before Proposition~\ref{Gamma_1_degree_prop}). It remains to note that given $v'\in E^{\operatorname{t}}$, we can identify which $v$ satisfy $v\adj{\Gamma} v'$ using only $N_\Gamma(v')$. Indeed, these are exactly those $v$ such that for some $i\neq j$, $f_i\neq-f_j$, we have $v=v'+f_i+f_j$, $v'+f_i\in\comp(\Gamma,x)$ and $v'+f_j\notin\comp(\Gamma,x)$.
\end{proof}
We are finally ready for
\begin{proof}[Proof of Proposition~\ref{count_components_prop}]
We first consider the case $R=0$. By Lemma~\ref{component_triviality_lemma}, the only $\Gamma\in\OMCut(x,B)$ having a $\GC$-connected component $E$ with $R_\Gamma(E)=0$ is the trivial $\Gamma$ having $E_1(\Gamma)=\{x\}$. Hence the proposition is straightforward in this case. For the rest of the proof we assume that $R>0$ in which case we may also assume that
\begin{equation}\label{R_M_ineq}
R\ge M\ge d-1,
\end{equation}
since Lemma~\ref{component_triviality_lemma} shows this is necessary for there to exist $\Gamma\in\OMCut(x,B)$ with $\GC$-connected components $E$ having $|E|=M$ and $R_\Gamma(E)=R$.

Let $\A'$ denote the set of all $(\Gamma,E)$ where $\Gamma\in\OMCut(x,B)$ and $E$ is a $\GC$-connected component of $E_1(\Gamma)$ associated with $b$ and satisfying $E\neq\{x\}$, $|E|=M$ and $R_\Gamma(E)=R$. Let $\A:=\{E\ |\ \exists \Gamma\text{ s.t. }(\Gamma,E)\in\A'\}$. Our goal is to upper bound $|\A|$. For $(\Gamma,E)\in\A'$, we define
\begin{equation*}
 S(\Gamma,E) = (E^{\operatorname{t}},N_\Gamma(E^{\operatorname{t}}))
\end{equation*}
where $E^{\operatorname{t}}\subseteq E$ is as in Proposition~\ref{dom_set_prop} (taking $d$ sufficiently large). By Lemma~\ref{Et_and_Nv_suffice}, $E$ is uniquely determined by $E^{\operatorname{t}}$ and $N_\Gamma(E^{\operatorname{t}})$ and hence
\begin{equation}\label{A_bound_by_A'}
 |\A|\le |S(\A')|.
\end{equation}
We estimate $|S(\A')|$ by showing how to describe succinctly a $(E^{\operatorname{t}}, N)\in S(\A')$. Fix some $(\Gamma,E)\in\A'$ such that $S(\Gamma,E)=(E^{\operatorname{t}},N)$. We describe $E^{\operatorname{t}}$ by prescribing a point $v\in E$, the size $|E^{\operatorname{t}}|$ and the location of the vertices of $E^{\operatorname{t}}$, given $v$ and $|E^{\operatorname{t}}|$. To estimate the number of possibilities for such a description, let $A=A(x,b,M)$ be the set from Proposition~\ref{point_in_cutset_comp_prop}. By that proposition and the fact that $E$ is associated with $b$ we have $|A|\le 40Mn_d^{\lfloor\frac{M}{\alpha}\rfloor}$ and $E\cap A\neq \emptyset$. Hence $v\in E$ may be prescribed as one of $40Mn_d^{\lfloor\frac{M}{\alpha}\rfloor}$ possibilities. We continue by noting that $|E^{\operatorname{t}}|\le |E|=M$, hence the size $|E^{\operatorname{t}}|$ may be prescribed as one of $M$ possibilities. Lastly, note that $E^{\operatorname{t}}\cup\{v\}$ is connected in $G^{\otimes 6}$ by Lemma~\ref{Et_connected_graph_lem} and Proposition~\ref{dom_set_prop}. Thus, Lemma~\ref{span_tree_prop} implies that given $v\in E$ and $|E^{\operatorname{t}}|$, the number of possibilities for $E^{\operatorname{t}}$ is at most $(2d)^{14|E^{\operatorname{t}}|}$ which by Proposition~\ref{dom_set_prop} is at most $(2d)^{\frac{140\log d}{d}M}$. Summing up, the number of possibilities for $E^{\operatorname{t}}$ is at most
\begin{equation}\label{Et_possibilities}
 40M^2n_d^{\lfloor\frac{M}{\alpha}\rfloor}(2d)^{\frac{140\log d}{d}M}.
\end{equation}
We continue by describing $N$. To do so, we prescribe $(P_\Gamma(v))_{v\in E^{\operatorname{t}}}$ and then $N=N_\Gamma(E^{\operatorname{t}})$ given $(P_\Gamma(v))_{v\in E^{\operatorname{t}}}$. The number of possibilities for $(P_\Gamma(v))_{v\in E^{\operatorname{t}}}$ is at most $(2d)^{|E^{\operatorname{t}}|}$ which by Proposition~\ref{dom_set_prop} is at most $(2d)^{\frac{10\log d}{d}M}$. For each $v\in E^{\operatorname{t}}$, given $P_\Gamma(v)$ we may describe $N_\Gamma(v)$ using at most $\binom{\Delta(G)}{P_\Gamma(v)}\le (2d)^{R_\Gamma(v)}$ possibilities. Hence given $(P_\Gamma(v))_{v\in E^{\operatorname{t}}}$, the number of possibilities for $N_\Gamma(E^{\operatorname{t}})$ is at most $(2d)^{R_\Gamma(E^{\operatorname{t}})}$ which by Proposition~\ref{dom_set_prop} is at most $(2d)^{\frac{10\log d}{d}R}$. In conclusion, the number of possibilities for $N$ given $E^{\operatorname{t}}$ is at most
\begin{equation}\label{N_possibilities}
 (2d)^{\frac{10\log d}{d}(M+R)}.
\end{equation}
Putting together \eqref{A_bound_by_A'}, \eqref{Et_possibilities} and \eqref{N_possibilities} we obtain
\begin{equation*}
|\A|\le |S(\A')|\le 40M^2n_d^{\lfloor\frac{M}{\alpha}\rfloor}(2d)^{\frac{150\log d}{d}(M+R)}\le CM^2n_d^{\left\lfloor\frac{M}{\alpha}\right\rfloor} e^{\frac{C\log^2 d}{d}(M+R)}
\end{equation*}
for some $C>0$. Using that $R\ge M\ge d-1$ by \eqref{R_M_ineq}, this implies
\begin{equation*}
 |\A|\le n_d^{\left\lfloor\frac{M}{\alpha}\right\rfloor} e^{\frac{C'\log^2 d}{d}R}
\end{equation*}
for some $C'>0$, as required.\qedhere
\end{proof}

\paragraph{Proof of Theorem~\ref{count_cutsets_thm}}\label{count_cutsets_thm_proof}
In the proof, we shall always assume that $d\ge d_0$ for some large
constant $d_0$. Fix a total order $\prec$ on $V[G]$. For
$\Gamma\in\OMCut(x,B)$, we say that a $\GC$-connected component $E$
of $E_1(\Gamma)$ is \emph{min-associated} to $b\in B$ if $E$ is
associated to $b$ and $E$ is not associated to any $b'\in B$ with
$b'\prec b$. Let $c(\Gamma)$ be the number of $\GC$-connected
components of $E_1(\Gamma)$ and $(E^i)_{i=1}^{c(\Gamma)}$ be these
components. We order the $(E^i)$ in such a way that if $i<j$, $E^i$
is min-associated to $b$ and $E^j$ is min-associated to $b'$ then
either $b\prec b'$, or both $b=b'$ and the $\prec$-least element of $E^i$
is smaller than the $\prec$-least element of $E^j$.
We now inductively define $m(\Gamma)$ and a vector
$(b^j)_{j=1}^{m(\Gamma)}\subseteq B$ as follows: $b^1$ is the
$\prec$-smallest element of $B$. Assuming that $(b^1,\ldots,
b^\ell)$ have already been defined, we set $m(\Gamma):=\ell$ if
$B\subseteq\cup_{j=1}^\ell \comp(\Gamma,b^j)$ or otherwise set
$b^{\ell+1}$ to be the $\prec$-smallest element of $B\setminus
\cup_{j=1}^\ell\comp(\Gamma,b^j)$. We finally let $c^j$ for $1\le
j\le m(\Gamma)$ be the number of $E^i$ which are min-associated to
$b^j$ (note that $c^j$ may be 0). We will write $E^i(\Gamma)$,
$b^j(\Gamma)$ and $c^j(\Gamma)$ for $E^i$, $b^j$ and $c^j$ when we
want to emphasize their dependence on $\Gamma$. Note that our
definitions imply that each $E^i$ is min-associated to one of the
$b^j$ and hence
\begin{equation}\label{m_c_relation}
\sum_{j=1}^{m(\Gamma)} c^j(\Gamma) = c(\Gamma).
\end{equation}

In this section, we say that the \emph{type} of $\Gamma$ is the vector
\begin{equation*}
\left(c(\Gamma),(|E^i|,R_\Gamma(E^i))_{i=1}^{c(\Gamma)}, m(\Gamma), (c^j(\Gamma))_{j=1}^{m(\Gamma)}\right).
\end{equation*}
Recalling the definition of $\OMCut(x,B,M,R)$ from the beginning of Section~\ref{rough_bdry_counting_sec}, we define $\TT(M,R)$, for $M,R\ge 0$, to be the set of all types of $\Gamma\in\OMCut(x,B,M,R)$. We shall need the following
\begin{lemma}\label{type_estimate_lemma}
There exists $C>0$ such that for all $M,R\ge 0$ we have
\begin{equation*}
|\TT(M,R)|\le \exp\left(\frac{C\log d}{d}R\right).
\end{equation*}
\end{lemma}
\begin{proof}
By Lemma~\ref{component_triviality_lemma}, we have that for $R=0$ the set $\OMCut(x,B,M,R)$ contains at most one $\Gamma$, the one with $E_1(\Gamma)=\{x\}$. Hence the current lemma follows trivially in this case. For the rest of the proof we assume that $R>0$ in which case Lemma~\ref{component_triviality_lemma} implies that every $\GC$-connected component $E$ of $E_1(\Gamma)$ for every $\Gamma\in\OMCut(x,B,M,R)$ satisfies
\begin{equation}\label{large_comp_ineq}
 R_\Gamma(E)\ge |E|\ge d-1.
\end{equation}

We continue by noting that if $\Gamma\in\OMCut(x,B,M,R)$ then
\begin{equation}\label{c_E_R_relation}
 \sum_{i=1}^{c(\Gamma)} |E^i| = M\qquad\text{and}\qquad\sum_{i=1}^{c(\Gamma)} R_\Gamma(E^i)=R.
\end{equation}
By Proposition~\ref{L_decomp_prop}, given an integer $L>0$, the number of solutions in integers $k$ and $(x_m)_{m=1}^k$ to
\begin{equation*}
\sum_{m=1}^k x_m=L
\end{equation*}
with each $x_m\ge d-1$ is at most $\exp\left(\frac{6\log d}{d}L\right)$. Hence by \eqref{large_comp_ineq} and \eqref{c_E_R_relation}, the number of possibilities for $\left(c(\Gamma),(|E^i|,R_\Gamma(E^i))_{i=1}^{c(\Gamma)}\right)$ over all $\Gamma\in\OMCut(x,B,M,R)$ is at most $\exp\left(\frac{6\log d}{d}(M+R)\right)$ which by \eqref{large_comp_ineq} is at most $\exp\left(\frac{12\log d}{d}R\right)$.

Next, we note that for $\Gamma\in\OMCut(x,B,M,R)$ we have
\begin{equation}\label{c_and_m_ineq}
c(\Gamma)\le \frac{R}{d-1}\qquad\text{and}\qquad m(\Gamma)\le R.
\end{equation}
The first assertion follows simply from $\eqref{large_comp_ineq}$ and $\eqref{c_E_R_relation}$. To see the second assertion, first note that for any $1\le j\le m(\Gamma)$, $\subcut(\Gamma,b^j)\in\OMCut(x,b^j)$ by the remark after the definition of $\OMCut$. Then, by Proposition~\ref{equal_edgenum_in_all_dir} we have that $|\subcut(\Gamma,b^j)|\ge\Delta(G)$ for all $j$. Since by Proposition~\ref{subcut_identical_prop} (and since identicality of the subcuts occurs only when their interior components are also equal, e.g., since the cutsets are odd), $\subcut(\Gamma,b^{j_1})\cap\subcut(\Gamma,b^{j_2})=\emptyset$ for distinct $1\le j_1,j_2\le m(\Gamma)$, it follows that $m(\Gamma)\le \frac{|\Gamma|}{\Delta(G)}$. The second assertion now follows by noting that $|\Gamma|\le \Delta(G) M\le \Delta(G) R$ by \eqref{large_comp_ineq}.

Using the relations \eqref{m_c_relation} and \eqref{c_and_m_ineq}, it follows that the number of possibilities for $\left(m(\Gamma), (c^j(\Gamma))_{j=1}^{m(\Gamma)}\right)$ over all $\Gamma\in\OMCut(x,B,M,R)$ is bounded by the number of solutions in integers $c, m$ and $(c^j)_{j=1}^m$ to
\begin{equation}\label{m_c_eq}
\sum_{j=1}^m c^j = c
\end{equation}
for $1\le m\le R$, $c^j\ge 0$ for all $j$ and $c\le \frac{R}{d-1}$. By standard combinatorial enumeration, the number of solutions to \eqref{m_c_eq} for fixed $m$ and $c$ is $\binom{m+c-1}{c}$. Thus standard estimates show that \eqref{m_c_eq} has at most $R^2\exp\left(\frac{C\log d}{d}R\right)$ solutions, for some $C>0$. Combining this estimate with the estimate for the number of possibilities for $\left(c(\Gamma),(|E^i|,R_\Gamma(E^i))_{i=1}^{c(\Gamma)}\right)$ obtained previously, we see that
\begin{equation*}
|\TT(M,R)|\le R^2\exp\left(\frac{C'\log d}{d} R\right)
\end{equation*}
for some $C'>0$. Since $R\ge d-1$ by $\eqref{large_comp_ineq}$, the lemma follows.\qedhere
\end{proof}

For $M,R\ge 0$ and $\gamma\in\TT(M,R)$, let
\begin{equation*}
\OMCut(x,B,M,R,\gamma) := \{\Gamma\in\OMCut(x,B,M,R)\ |\ \Gamma\text{ has type }\gamma\}.
\end{equation*}
Fix $\gamma\in\TT(M,R)$ and $1\le k<m(\Gamma)$ where here, $m(\Gamma)$ is the third element of $\gamma$. By our definitions, $b^{k+1}(\Gamma)$ is well defined for $\Gamma\in\OMCut(x,B,M,R,\gamma)$. The next lemma notes that $b^{k+1}(\Gamma)$ is determined also from partial information about $\Gamma$.
\begin{lemma}\label{b_determined_lem}
The point $b^{k+1}(\Gamma)$ is determined as a function only of $(b^1(\Gamma),\ldots, b^k(\Gamma))$ and the set of all $E^i(\Gamma)$ which are associated to some $b^j$ for $j\le k$.
\end{lemma}
\begin{proof}
Knowing $(b^j(\Gamma))_{j=1}^k$ and the given $E^i(\Gamma)$ determines $\comp(\Gamma,b^j(\Gamma))$ for all $j\le k$. By our definitions, $b^{k+1}(\Gamma)$ is the $\prec$-smallest point of $B$ which is not in $\cup_{j=1}^k\comp(\Gamma,b^j)$.
\end{proof}

We finally reach
\begin{proof}[Proof of Theorem~\ref{count_cutsets_thm}] As in the proof of Proposition~\ref{count_components_prop}, we count by showing that a $\Gamma\in\OMCut(x,B,M,R)$ may be described succinctly. We describe a $\Gamma\in\OMCut(x,B,M,R)$ by
\begin{enumerate}
 \item The type $\left(c(\Gamma),(|E^i|,R_\Gamma(E^i))_{i=1}^{c(\Gamma)}, m(\Gamma), (c^j(\Gamma))_{j=1}^{m(\Gamma)}\right)$ of $\Gamma$.
 \item For each $1\le j\le m(\Gamma)$, in this order:
\begin{enumerate}
\item For each of the $c^j(\Gamma)$ of the $E^i$ which are min-associated to $b^j$, in the order they appear in $(E^i)_{i=1}^{c(\Gamma)}$:
\begin{enumerate}
\item A description of $E^i$.
\end{enumerate}
\end{enumerate}
\end{enumerate}
We emphasize that in step 2(a) above, if $c^j(\Gamma)=0$, we do not describe anything and go on to the next $j$.

We first need to show that $\Gamma$ can indeed be recovered from the above description. Then we will estimate the number of possibilities for this description in order to obtain a bound for $|\OMCut(x,B,M,R)|$. To see that $\Gamma$ can be recovered, note that the above description gives all the $\GC$-connected components $(E^i)_{i=1}^{c(\Gamma)}$ of $E_1(\Gamma)$ (since each component is min-associated to some $b\in B$). These, in turn, suffice to recover $\comp(\Gamma,b)$ for all $b\in B$ from which we get that $\Gamma$ is all edges between $\cup_{b\in B}\comp(\Gamma,b)$ and its complement.

We next estimate the number of possibilities for the above description. We start with a definition. For $b\in B$ and $M',R'\ge 0$, define $\A(b,M',R')$ to be the set of all $\GC$-connected components $E$, associated to $b$ and having $|E|=M'$ and $R_{\Gamma'}(E)=R'$, of $E_1(\Gamma')$ for some $\Gamma'\in\OMCut(x,B)$ (which is not fixed in advance). In Proposition~\ref{count_components_prop} we showed that
\begin{equation*}
|\A(b,M',R')|\le n_d^{\left\lfloor\frac{M'}{\alpha}\right\rfloor} \exp\left(\frac{C\log^2 d}{d}R'\right)
\end{equation*}
for some $C>0$.

Fix $\gamma\in\TT(M,R)$ and let us estimate the number of possibilities for the above description for $\Gamma\in\OMCut(x,B,M,R,\gamma)$. Part 1 has just one option since the type of $\Gamma$ is fixed. Hence we need only estimate how many possibilities there are for $E^i$ each time we reach part 2(a)i above, given the partial information about $\Gamma$ described up to that point.

We claim that whenever we reach part 2(a)i above for a particular $b^j$ and $E^i$, we have already described the point $b^j$ itself, $|E^i|$ and $R_\Gamma(E^i)$. To see this, note that by our definitions, the $E^i$ which are min-associated to $b^j$ are exactly those for which $i\in\{i_0+1,i_0+2,\ldots, i_0+c^j\}$ where $i_0=\sum_{k=1}^{j-1} c^k$ and hence $|E^i|$ and $R_\Gamma(E^i)$ are known from $\gamma$. We use induction to show that $b^j$ has also been described. For $j=1$ this follows since $b^1$ is the $\prec$-smallest point in $B$. Assuming the claim is true for all $1\le k<j$, the claim for $j$ follows from Lemma~\ref{b_determined_lem} since when we reach part 2(a)i for that $j$, we have already described $(b^k)_{k=1}^{j-1}$ and all the $E^i$ which are associated to some $b^k$ for $k<j$. We see that we may describe $E^i$ as an element of $\A(b^j,|E^i|,R_\Gamma(E^i))$ and hence have at most
\begin{equation*}
n_d^{\bigl\lfloor\frac{|E^i|}{\alpha}\bigr\rfloor} \exp\left(\frac{C\log^2 d}{d}R_\Gamma(E^i)\right)
\end{equation*}
possibilities for its description. In conclusion, we see that the number of possibilities for the above description for $\Gamma\in\OMCut(x,B,M,R,\gamma)$ is at most
\begin{equation*}
\prod_{i=1}^{c(\Gamma)} n_d^{\bigl\lfloor\frac{|E^i|}{\alpha}\bigr\rfloor} \exp\left(\frac{C\log^2 d}{d}R_\Gamma(E^i)\right)\le n_d^{\left\lfloor\frac{M}{\alpha}\right\rfloor} \exp\left(\frac{C\log^2 d}{d}R\right)
\end{equation*}
which is independent of $\gamma$. Hence, the number of possibilities for the above description for $\Gamma\in\OMCut(x,B,M,R)$ is at most
\begin{equation*}
|\TT(M,R)|n_d^{\left\lfloor\frac{M}{\alpha}\right\rfloor} \exp\left(\frac{C\log^2 d}{d}R\right)
\end{equation*}
which by Lemma~\ref{type_estimate_lemma} is at most
\begin{equation*}
n_d^{\left\lfloor\frac{M}{\alpha}\right\rfloor} \exp\left(\frac{C'\log^2 d}{d}R\right)
\end{equation*}
for some $C'>0$. Since $\Gamma$ may be recovered from the above description, this is also a bound for $|\OMCut(x,B,M,R)|$, proving the theorem.\qedhere
\end{proof}

\subsubsection{Counting Interior Approximations To Cutsets}\label{partial_edge_cutsets_section}
We start with a definition. For $x, b\in V[G]$ and $\Gamma\in\OMCut(x,b)$, recalling the definition of $\Eonee(\Gamma)$ from Section~\ref{preliminaries_sec}, we say that $E\subseteq V[G]$ is an \emph{interior approximation} to $\Gamma$ if
\begin{equation*}
E_1(\Gamma)\setminus \Eonee(\Gamma) \subseteq E\subseteq \comp(\Gamma,x).
\end{equation*}
The following is the main theorem of this section (recall from \eqref{alpha_def} that $\alpha=\prod_{i=1}^{d-1} n_i$).
\begin{theorem}\label{interior_approximation_theorem}
There exist $d_0,C>0$ such that for all $d\ge d_0$, $L\in\N$ and $x,b\in V[G]$, there exists a family $\mathcal{E}$ of subsets of $V[G]$ satisfying
\begin{equation*}
|\mathcal{E}|\le 2n_d^{\left\lfloor\frac{L}{\alpha}\right\rfloor}\exp\left(\frac{C\log^2 d}{d^{3/2}}L\right)
\end{equation*}
and such that for every $\Gamma\in\OMCut(x,b)$ with $|\Gamma|=L$ there is an $E\in\mathcal{E}$ which is an interior approximation to $\Gamma$.
\end{theorem}
Aiming towards an application of this theorem, we make the following definitions. For $x\in V[G]$, legal boundary conditions $(B,\mu)$ with non-positive $\mu$ and $f\in \Hom(G,B,\mu)$,
denoting $\Gamma:=\LS(f,x,B)$ and assuming $\Gamma\neq\emptyset$, we say that a function $g\in\Hom(G,B,\mu)$ is a \emph{$(x,B)$-interior modification} of $f$ if $f(v)=g(v)$ for all $v\notin\comp(\Gamma,x)$ and $g(v)=1$ for all $v\in \Eonee(\Gamma)$. Recalling the transformation $T_2$ of Section~\ref{shift_flip_transformation_sec}, we note that any $g\in T_2(f)$ is a $(x,B)$-interior modification of $f$.
In addition, for $x$ and $(B,\mu)$ as above, $L\in\N$ and $g\in\Hom(G,B,\mu)$, we define
\begin{equation*}
\begin{split}
\PLS(g,& x, B, L)=\\
&=\{\LS(f,x,B)\ |\ f\in \Omega_{x,L}\text{, $g$ is a $(x,B)$-interior modification of $f$}\},
\end{split}
\end{equation*}
the ``possible level sets for $f$ given $g$''. Note that any $\Gamma\in\PLS(g,x,B,L)$ satisfies $\Gamma\in\OMCut(x,B)$ and $|\Gamma|=L$. We will use Theorem~\ref{interior_approximation_theorem} to prove
\begin{theorem}\label{count_level_sets_thm}
There exist $d_0,C>0$ such that for all $d\ge d_0$, $L\in\N$, $x\in V[G]$, legal boundary conditions $(B,\mu)$ with non-positive $\mu$ and $g\in\Hom(G,B,\mu)$, we have
\begin{equation*}
|\PLS(g,x,B,L)|\le 2n_d^{\left\lfloor\frac{L}{\alpha}\right\rfloor}\exp\left(\frac{C\log^2 d}{d^{3/2}}L\right).
\end{equation*}
\end{theorem}
\paragraph{Proof of Theorem~\ref{interior_approximation_theorem}}

Throughout the proof, we fix $x,b\in V[G]$ and shall always assume that $d\ge d_0$ for some large constant $d_0$. Also, for $\Gamma\in\OMCut(x,b)$ we adapt the notation $E_1:=E_1(\Gamma), \CC_1:=\comp(\Gamma,x), E_0:=E_0(\Gamma)$ and $\CC_0:=\comp(\Gamma,b)$ where the dependence on $\Gamma$ is implicit and the choice of $\Gamma$ will be understood from the context. Note that $\CC_0=V[G]\setminus \CC_1$ by minimality of $\Gamma$. We will also write, for $j\in\{0,1\}$ and a condition $c(\cdot)$,
\begin{equation*}
E_{j,c(\cdot)} := \{v\in E_j\ |\ c(P_\Gamma(v))\text{ holds}\}.
\end{equation*}
For example, $E_{1,\sqrt{d}<\cdot<\Delta(G)-\sqrt{d}}=\{v\in E_1\ | \ \sqrt{d}<P_\Gamma(v)<\Delta(G)-\sqrt{d}\}$ and $E_{1,\cdot\ge\Delta(G)-\sqrt{d}}=\Eonee(\Gamma)$. Finally, for $j\in\{0,1\}$ and $v\in E_j$, we let
\begin{equation*}
\begin{split}
A_1(v)&:=\{v'\in E_j\ \big|\ \exists u\in\CC_j\text{ such that }v\adj{G} u\text{, }u\adj{G} v'\},\\
A_2(v)&:=\{u\in S(v)\cap \CC_j\ \big|\ |S(u)\cap E_j|<\sqrt{d}\},\\
A_3(v)&:=S(A_2(v))\cap E_j.
\end{split}
\end{equation*}

We remind that a $\Gamma\in\OMCut(x,b)$ is called trivial if it consists only of the edges incident to $x$ or only of the edges incident to $b$ (see Proposition~\ref{trivial_gamma_prop}), we remind the definition of $R_\Gamma$ from \eqref{R_Gamma_def} and we start our proof with the following ``dominating set'' proposition.
\begin{proposition}\label{dom_set_prop2}
There exists $C>0$ such that for all non-trivial $\Gamma\in\OMCut(x,b)$, there exist $E_{0}^{\operatorname{t}}\subseteq E_0$ and $E_{1}^{\operatorname{t}}\subseteq E_1$ satisfying for both $j\in\{0,1\}$:
\begin{enumerate}
\item[(a)] $R_\Gamma(E_{j}^{\operatorname{t}})\le \frac{C\log d}{d^{3/2}}|\Gamma|$.
\item[(b)] If $v\in E_j$ and $|A_1(v)|\ge \frac{1}{5}d^{3/2}$ then $A_1(v)\cap E_{j}^{\operatorname{t}}\neq\emptyset$.
\item[(c)] If $v\in E_{j,\cdot\ge\Delta(G)/2}$ then $|S(v)\cap E_{1-j}\cap S(E_{j}^{\operatorname{t}})|\ge\sqrt{d}$.
\item[(d)] If $v\in E_{j,\cdot\le\sqrt{d}}$ and $|A_2(v)|\ge \frac{\Delta(G)}{2}$ then $A_3(v)\cap S(E_{1-j}^{\operatorname{t}})\neq\emptyset$.
\end{enumerate}
\end{proposition}
\begin{proof}
Fix a non-trivial $\Gamma\in\OMCut(x,b)$. Note that the non-triviality and Proposition~\ref{trivial_gamma_prop} imply
\begin{equation}\label{gamma_non-triviality_conseq}
\text{$P_\Gamma(v)\le\Delta(G)-1$ for all $v\in V[G]$}.
\end{equation}
For $j\in\{0,1\}$, we choose $E_{j}^{\operatorname{s}}\subseteq E_j$ randomly by adding each $v\in E_j$ to $E_{j}^{\operatorname{s}}$ independently with probability $\frac{30\log d}{(\Delta(G)-P_\Gamma(v))\sqrt{d}}$. These probabilities are indeed at most 1 for sufficiently large $d$ by \eqref{gamma_non-triviality_conseq}.

Fix $j\in\{0,1\}$. Using that $\sum_{k=1}^{\Delta(G)}k|E_{j,\cdot = k}|=|\Gamma|$, since the subsets of $\Gamma$ incident to distinct vertices in $E_j$ are disjoint, we have
\begin{equation*}
\begin{split}
\E R_\Gamma(E_{j}^{\operatorname{s}})&=\frac{30\log d}{\sqrt{d}}\sum_{v\in E_j} \frac{\min(P_\Gamma(v),\Delta(G)-P_\Gamma(v))}{\Delta(G)-P_\Gamma(v)}=\\
&=\frac{30\log d}{\sqrt{d}}\left(|E_{j,\cdot\ge\Delta(G)/2}|+\sum_{k=1}^{\lceil\Delta(G)/2\rceil-1}\frac{k|E_{j,\cdot=k}|}{\Delta(G)-k}\right)\le\\
&\le \frac{30\log d}{\sqrt{d}}\left(\frac{2|\Gamma|}{\Delta(G)}+\frac{2|\Gamma|}{\Delta(G)}\right)\le \frac{120\log d|\Gamma|}{d^{3/2}}.
\end{split}
\end{equation*}
Markov's inequality now implies that
\begin{equation}\label{R_Gamma_prob_estimate}
\P\left(R_\Gamma(E_{j}^{\operatorname{s}})\ge\frac{360\log d|\Gamma|}{d^{3/2}}\right)\le \frac{1}{3}.
\end{equation}

Let $v_1\in E_j$ be such that $|A_1(v_1)|\ge\frac{1}{5}d^{3/2}$. We have
\begin{equation}\label{v1_prob_estimate}
\P(E_{j}^{\operatorname{s}}\cap A_1(v_1)=\emptyset)\le \left(1-\frac{30\log d}{\Delta(G)\sqrt{d}}\right)^{\frac{1}{5}d^{3/2}}\le \exp(-3\log d)=\frac{1}{d^3}.
\end{equation}

Let $v_2\in E_{j,\cdot\ge\Delta(G)/2}$. With part (c) of the proposition in mind, we would like to estimate $\P\big(|S(v_2)\cap E_{1-j}\cap S(E_{j}^{\operatorname{s}})|<\sqrt{d}\big)$. We first let $B(v_2):=S(v_2)\cap E_{1-j,\cdot\ge 2}$ and note that
\begin{equation}\label{Bv_size}
|B(v_2)|\ge \frac{\Delta(G)}{2}-1.
\end{equation}
To see this, note that by \eqref{gamma_non-triviality_conseq}, there exists $1\le i\le \Delta(G)$ such that $v_2+f_i\in\CC_j$. Hence $v_2+f_i+f_k\in\CC_j$ for all $k$ by \eqref{parity_neighborhood_contained}. Thus, each $1\le i'\le \Delta(G)$ for which $v_2+f_{i'}\notin\CC_j$ and $f_{i'}\neq -f_i$ satisfies $v_2+f_{i'}\in B(v_2)$ since $v_2+f_{i'}$ is adjacent to both $v_2$ and $v_2+f_i+f_{i'}$.

Next, for each $w\in B(v_2)$, let $E(w):=(S(w)\cap E_j)\setminus\{v_2\}$ and define a random set $E(w)^{\operatorname{s}}$ by taking each $v'\in E(w)$ into $E(w)^{\operatorname{s}}$ with probability $\frac{15\log d}{(\Delta(G)-P_\Gamma(v'))\sqrt{d}}$ independently for each such $v'$ and $w$. We note that by Proposition~\ref{2nd_neighborhood_prop}, each $v'$ is contained in at most $2$ of the $E(w)$'s and hence
\begin{equation}\label{Ews_domination}
\bigcup_{w\in B(v_2)} E(w)^{\operatorname{s}}\text{ is stochastically dominated by }E_{j}^{\operatorname{s}}.
\end{equation}
Noting that for $w\in B(v_2)$, $P_\Gamma(w)\ge 2$ by definition of
$B(v_2)$ and $\Delta(G)-P_\Gamma(v')\le P_\Gamma(w)$ for all $v'\in
E(w)$ by Proposition~\ref{neighbor_plaquette_prop}, we obtain for
sufficiently large $d$,
\begin{equation*}
\P\big(S(w)\cap E(w)^{\operatorname{s}}=\emptyset\big)\le \left(1-\frac{15\log d}{P_\Gamma(w)\sqrt{d}}\right)^{P_\Gamma(w)-1}\le 1-\frac{15}{\sqrt{d}}.
\end{equation*}

Finally, letting $N:=|\{w\in B(v_2)\ |\ S(w)\cap E(w)^{\operatorname{s}}\neq\emptyset\}|$, it follows that $N$ stochastically dominates a $\Bin(|B(v_2)|,\frac{15}{\sqrt{d}})$ random variable. Using \eqref{Bv_size}, \eqref{Ews_domination} and standard properties of binomial RV's, we deduce that for large enough $d$,
\begin{equation}\label{v2_prob_estimate}
\P\big(|S(v_2)\cap E_{1-j}\cap S(E_{j}^{\operatorname{s}})|<\sqrt{d}\big)\le \P(N<\sqrt{d})\le \frac{1}{d^3}.
\end{equation}

Having now part (d) of the Proposition in mind, we let $v_3\in E_{j,\cdot\le d}$ satisfy $|A_2(v_3)|\ge\frac{\Delta(G)}{2}$. Let $1\le i\le \Delta(G)$ be such that $v_3+f_i\in E_{1-j}$. Let $1\le i'\le \Delta(G)$ be such that $v_3+f_i+f_{i'}\in\CC_{1-j}$, such $i'$ exists by \eqref{gamma_non-triviality_conseq}. It follows from \eqref{parity_neighborhood_contained} that $S(v_3+f_i+f_{i'})\subseteq\CC_{1-j}$.
Let $i_1,\ldots, i_{\lceil\Delta(G)/2\rceil}$ be such that $v_3+f_{i_k}\in A_2(v_3)$ for all $k$. Again, \eqref{parity_neighborhood_contained} implies that $S(v_3+f_{i_k})\subseteq \CC_j$ for all $k$.
We deduce that for all $k$, $v_3+f_i+f_{i_k}\in A_3(v_3)$ and $v_3+f_i+f_{i_k}+f_{i'}\in E_{1-j}$. Furthermore, by Proposition~\ref{interior_P_gamma_estimate}
and the definition of $A_2(v_3)$ (with $v_3+f_{i_k}$ as $u$ and $v_3+f_i+f_{i_k}$ as $v$), $P_\Gamma(v_3+f_i+f_{i_k})<\sqrt{d}$. Hence, by Proposition~\ref{neighbor_plaquette_prop}, $P_\Gamma(v_3+f_i+f_{i_k}+f_{i'})\ge \Delta(G)-\sqrt{d}$. We deduce that
\begin{align}
\nonumber\P\big(A_3(v_3)\cap S(E_{1-j}^{\operatorname{s}})=\emptyset\big)&\le \P\big((v_3+f_i+f_{i_k}+f_{i'})_{k=1}^{\lceil\Delta(G)/2\rceil}\cap E_{1-j}^{\operatorname{s}}=\emptyset\big)\le\\
&\le \left(1-\frac{30\log d}{d}\right)^{\Delta(G)/2}\le \frac{1}{d^3}.\label{v3_prob_estimate}
\end{align}
We now aim to ``correct'' the sets $E_j^{\operatorname{s}}$ by enlarging them slightly to create new sets $E_j^{\operatorname{t}}$ which will satisfy the requirements of the proposition. Defining
\begin{equation*}
\begin{split}
B_{j,1}&:=\left\{v\in E_j\ \big|\ |A_1(v)|\ge \frac{1}{5}d^{3/2},\ E_{j}^{\operatorname{s}}\cap A_1(v)=\emptyset\right\},\\
B_{j,2}&:=\left\{v\in E_{j,\cdot\ge\Delta(G)/2}\ \big|\ |S(v)\cap E_{1-j}\cap S(E_{j}^{\operatorname{s}})|<\sqrt{d}\right\},\\
B_{j,3}&:=\left\{v\in E_{j,\cdot\le\sqrt{d}}\ \big|\ |A_2(v)|\ge\frac{\Delta(G)}{2},\ A_3(v)\cap S(E_{1-j}^{\operatorname{s}})=\emptyset\right\},
\end{split}
\end{equation*}
and using the three probabilistic estimates \eqref{v1_prob_estimate}, \eqref{v2_prob_estimate} and \eqref{v3_prob_estimate}, we see that
\begin{equation}\label{B_exp_size}
\max\big(\E|B_{j,1}|, \E|B_{j,2}|, \E|B_{j,3}|\big)\le \frac{|E_j|}{d^3}.
\end{equation}
Let $M:=\max_{\substack{j\in\{0,1\}\\k\in\{1,2,3\}}}|B_{j,k}|$. For $j\in\{0,1\}$, we let $E_j^{\operatorname{t}}:=E_j^{\operatorname{s}}\cup D_j$ where the $D_j$ satisfy $D_j\subseteq E_j$ and $|D_j|\le 3M$ and are chosen in such a way that parts (b), (c) and (d) of the proposition hold. The exact choice of $D_j$ does not matter and for sufficiently large $d$, one may take, for example, $D_j$ to be $B_{j,1}\cup B_{j,2}$ union with a set containing a neighbor in $E_j$ for each $v\in B_{1-j,3}$.
It follows directly that for each $j\in \{0,1\}$,
\begin{equation*}
R_\Gamma(E_{j}^{\operatorname{t}})\le R_\Gamma(E_{j}^{\operatorname{s}}) + R_\Gamma(D_j)\le R_\Gamma(E_{j}^{\operatorname{s}}) + 3M\frac{\Delta(G)}{2}\le R_\Gamma(E_{j}^{\operatorname{s}}) + 3dM.
\end{equation*}
Hence it is sufficient to show that with positive probability $\max_{j\in\{0,1\}} R_\Gamma(E_{j}^{\operatorname{s}})\le \frac{C|\Gamma|\log d}{d^{3/2}}$ and $M\le \frac{C|\Gamma|}{d^{3}}$ for some $C>0$. Using \eqref{B_exp_size}, Markov's inequality and the fact that $|E_j|\le |\Gamma|$ we have $\P\big(M\ge \frac{20|\Gamma|}{d^3}\big)\le \P\big(\sum_{\substack{j\in\{0,1\}\\k\in\{1,2,3\}}} |B_{j,k}|\ge \frac{20|E_j|}{d^3}\big)<\frac{1}{3}$.
Combined with \eqref{R_Gamma_prob_estimate} and a union bound, this proves the proposition.\qedhere
\end{proof}

For $\Gamma\in\OMCut(x,b)$, $v\in V[G]$ and $E\subseteq V[G]$, define $N_\Gamma(v)\in\{0,1\}^{\Delta(G)}$ by
\begin{equation*}
N_\Gamma(v)_i:=1_{(v+f_i\in\CC_1)}\text{ and }N_\Gamma(E):=(N_\Gamma(v))_{v\in E}.
\end{equation*}
The next proposition formalizes the fact that for a non-trivial $\Gamma\in\OMCut(x,b)$ knowing only the $(E_{j}^{\operatorname{t}})_{j\in\{0,1\}}$ of Proposition~\ref{dom_set_prop2} and $(N_\Gamma(E_{j}^{\operatorname{t}}))_{j\in\{0,1\}}$, we can determine a set $E$ satisfying $E_{1,\cdot<\Delta(G)-\sqrt{d}}\subseteq E\subseteq\CC_1$. $E$ is determined by the following algorithm:
\begin{enumerate}
\item For $j\in\{0,1\}$, let
\begin{enumerate}
\item $R_j^{\operatorname{a}}$ be all $v\in V[G]$ satisfying that there exist $v'\in E_{1-j}^{\operatorname{t}}$ and $1\le i\le\Delta(G)$ such that $N_\Gamma(v')_i=j$ and $v=v'+f_i$.
\item $R_j^{\operatorname{b}}$ be all $v\in V[G]$ satisfying that there exist $v'\in E_{j}^{\operatorname{t}}$ and $1\le i\le\Delta(G)$ such that $N_\Gamma(v')_i=j$ and $v\adj{G}v'+f_i$.
\end{enumerate}
\item For $j\in \{0,1\}$, let $V_j:=\{v\in V[G]\ \big|\ |S(v)\cap R_{1-j}^{\operatorname{a}}|<\sqrt{d}\}$ and define
\begin{equation*}
U:=\left\{u\in V_0\setminus R_0^{\operatorname{b}}\ \big|\ S(u)\cap V_1\cap R_1^{\operatorname{a}}\neq\emptyset\right\}.
\end{equation*}
Set $E:=R_1^{\operatorname{b}}\cup S(U)$.
\end{enumerate}
\begin{proposition}\label{E_alg_prop}
For any non-trivial $\Gamma\in\OMCut(x,b)$, the set $E$ obtained from the previous algorithm, taking as input the sets $(E_{j}^{\operatorname{t}})_{j\in\{0,1\}}$ of Proposition~\ref{dom_set_prop2} and $(N_\Gamma(E_{j}^{\operatorname{t}}))_{j\in\{0,1\}}$, satisfies
\begin{equation*}
E_{1,\cdot<\Delta(G)-\sqrt{d}}\subseteq E\subseteq\CC_1.
\end{equation*}
In other words, $E$ is an interior approximation to $\Gamma$.
\end{proposition}
To gain some intuition for the above algorithm, one should have in mind the following claims which are used in the proof of the proposition. $R_j^{\operatorname{a}}$ and $R_j^{\operatorname{b}}$ consist of vertices we know are in $E_j$ and $\CC_j$, respectively, directly from the definitions of $(E_{j}^{\operatorname{t}})_{j\in\{0,1\}}$ and $(N_\Gamma(E_{j}^{\operatorname{t}}))_{j\in\{0,1\}}$. $E_{1,\sqrt{d}<\cdot<\Delta(G)-\sqrt{d}}$ is seen to be a subset of $R_1^{\operatorname{b}}$ in a relatively straightforward manner and our main difficulty lies in showing that vertices of $E_{1,\cdot\le\sqrt{d}}$ can also be recovered from the given input. To this end, we define $V_j$ which is shown to be disjoint from $E_{j,\cdot\ge\Delta(G)/2}$. We deduce that $U$ consists only of vertices in $\CC_1\cap V^\even$. It follows from the definition of $\OMCut$ that $S(U)\subseteq \CC_1$. Finally, we are able to show that if $v\in E_{1,\cdot\le\sqrt{d}}\setminus R_1^{\operatorname{b}}$ then $v\in S(U)$.
\begin{proof}[Proof of Proposition~\ref{E_alg_prop}]
The proof is via several claims.

{\bf Claim 1:} $R_j^{\operatorname{a}}\subseteq E_j$ and $R_j^{\operatorname{b}}\subseteq \CC_j$ for $j\in\{0,1\}$.

We prove the claim for $R_0^{\operatorname{a}}$ and $R_0^{\operatorname{b}}$. The proofs for $R_1^{\operatorname{a}}$ and $R_1^{\operatorname{b}}$ are similar. Let $v\in R_0^{\operatorname{a}}$ and $v'\in E_1^{\operatorname{t}}$ be such that $v=v'+f_i$ and $N_\Gamma(v')_i=0$. Then $v\in E_0$ by definition of $N_\Gamma(v')$ and $E_0$. Let $v\in R_0^{\operatorname{b}}$ and $v'\in E_0^{\operatorname{t}}$ be such that $v\adj{G}v'+f_i$ and $N_\Gamma(v')_i=0$.
Then $v\in\CC_0$ by definition of $N_\Gamma(v')$ and \eqref{parity_neighborhood_contained}.

{\bf Claim 2:} For $j\in\{0,1\}$, $E_{j,\sqrt{d}<\cdot<\Delta(G)-\sqrt{d}}\subseteq R_j^{\operatorname{b}}$.

Fix $j\in\{0,1\}$ and $v\in E_{j,\sqrt{d}<\cdot<\Delta(G)-\sqrt{d}}$. By Proposition~\ref{Gamma_1_degree_prop} we know that $v$ has at least $\sqrt{d}(\Delta(G)-\sqrt{d})-\sqrt{d}\ge \frac{1}{2}d^{3/2}$ $\Gamma$-neighbors. Since all these neighbors are in $A_1(v)$, part (b) of Proposition~\ref{dom_set_prop2} implies that there exists $v'\in E_{j}^{\operatorname{t}}\cap A_1(v)$.
Hence $v\in R_j^{\operatorname{b}}$.

{\bf Claim 3:} For $j\in\{0,1\}$, $E_{j,\cdot\ge \Delta(G)/2}\cap V_j=\emptyset$.

Fix $j\in\{0,1\}$ and $v\in E_{j,\cdot\ge \Delta(G)/2}$. Any vertex in $S(v)\cap E_{1-j}\cap S(E_{j}^{\operatorname{t}})$ is in $R_{1-j}^{\operatorname{a}}$. Thus the claim follows from part (c) of Proposition~\ref{dom_set_prop2}.

{\bf Claim 4:} $U\subseteq\CC_1\cap V^\even$.

Let $u\in U$. $u\in V^\even$ since $S(u)\cap R_1^{\operatorname{a}}\neq\emptyset$ and
$R_1^{\operatorname{a}}\subseteq E_1\subseteq V^\odd$ by Claim 1 and the definition
of $\OMCut$. Assume, in order to get a contradiction, that
$u\notin\CC_1$. Since $S(u)\cap R_1^{\operatorname{a}}\neq\emptyset$ and
$R_1^{\operatorname{a}}\subseteq E_1$ by Claim~1, it follows that $u\in E_0$. If
$\sqrt{d}<P_\Gamma(u)<\Delta(G)-\sqrt{d}$ then $u\in R_0^{\operatorname{b}}$ by Claim
2, contradicting the definition of $U$. If
$P_\Gamma(u)\ge\Delta(G)/2$ we have $u\notin V_0$ by Claim 3,
contradicting again the definition of $U$. Finally, if
$P_\Gamma(u)\le \sqrt{d}$, let $v\in S(u)\cap V_1\cap R_1^{\operatorname{a}}$ (which
exists by the definition of $U$) and note that by Claim 1 and
Proposition~\ref{neighbor_plaquette_prop}, $P_\Gamma(v)\ge
\Delta(G)-\sqrt{d}\ge \frac{\Delta(G)}{2}$. It follows from Claim 3
that $v\notin V_1$, a contradiction. The contradiction proves the
claim.

{\bf Claim 5:} $S(U)\subseteq\CC_1$.

This follows immediately from Claim 4 since $E_1$, the boundary of $\CC_1$, is a subset of $V^\odd$.

{\bf Claim 6:} $E_{1,\cdot\le\sqrt{d}}\subseteq E$.

Let $v\in E_{1,\cdot\le\sqrt{d}}$. We distinguish two cases:
\begin{enumerate}
\item $|A_2(v)| < \frac{\Delta(G)}{2}$. We note that by definition of $A_2(v)$, for any $1\le i\le \Delta(G)$ such that $v+f_i\in\CC_1\setminus A_2(v)$, we have at least $\sqrt{d}$ vertices $v'\in E_1$ of the form $v'=v+f_i+f_k$ for some $k$ ($v$ being one of these vertices). Since $|S(v)\cap(\CC_1\setminus A_2(v))|\ge \frac{\Delta(G)}{2}-\sqrt{d}$ by our assumption, we see using Proposition~\ref{2nd_neighborhood_prop} that $|A_1(v)|\ge \frac{1}{2}(\frac{\Delta(G)}{2}-\sqrt{d})(\sqrt{d}-1)\ge \frac{d^{3/2}}{5}$ for large enough $d$. Hence, by part (b) of Proposition~\ref{dom_set_prop2}, $E_{1}^{\operatorname{t}}\cap A_1(v)\neq \emptyset$ implying that $v\in R_1^{\operatorname{b}}$.

\item $|A_2(v)|\ge \frac{\Delta(G)}{2}$. In this case, by part (d) of Proposition~\ref{dom_set_prop2} there exists $v'\in A_3(v)\cap S(E_{0}^{\operatorname{t}})$ implying that $v'\in R_1^{\operatorname{a}}$. By definition of $A_3(v)$, we may write $v'=v+f_i+f_k$ for some $1\le i,k\le \Delta(G)$ where $u:=v+f_i\in A_2(v)$. Using that $R_1^{\operatorname{a}}\subseteq E_1$ and $R_0^{\operatorname{b}}\subseteq \CC_0$ by Claim~1 and using the definition of $A_2(v)$ we deduce $u\in V_0\setminus R_0^{\operatorname{b}}$. Proposition~\ref{interior_P_gamma_estimate} implies that $P_\Gamma(v')<\sqrt{d}$ by definition of $A_2(v)$. Hence, since $R_0^{\operatorname{a}}\subseteq E_0$ by Claim 1, we have $v'\in V_1$. It follows that $u\in U$ and hence $v\in S(U)$.
\end{enumerate}
Claims 1,2,5 and 6 prove the proposition.
\end{proof}

\begin{lemma}\label{8_connect_lemma}
For all non-trivial $\Gamma\in\OMCut(x,b)$, denoting $F:=E_0\cup E_1$ and $F^{\operatorname{t}}:=E_{0}^{\operatorname{t}}\cup E_{1}^{\operatorname{t}}$ for the $(E_{j}^{\operatorname{t}})_{j\in\{0,1\}}$ of Proposition~\ref{dom_set_prop2}, if $F^{\operatorname{c}}$ is a $\GC$-connected component of $F$ then
\begin{enumerate}
\item[(a)] $F^{\operatorname{c}}\cap F^{\operatorname{t}}\neq\emptyset$.
\item[(b)] For every $v\in F^{\operatorname{c}}$, $(F^{\operatorname{t}}\cap F^{\operatorname{c}})\cup \{v\}$ is connected in $G^{\otimes 8}$.
\end{enumerate}
\end{lemma}
\begin{proof}
Fix a non-trivial $\Gamma\in\OMCut(x,b)$ and a $\GC$-connected component $F^{\operatorname{c}}$ of $F$. By part (c) of Proposition~\ref{dom_set_prop2}, for any $v\in E_{j,\cdot\ge\Delta(G)/2}\cap F^{\operatorname{c}}$ for some $j\in \{0,1\}$ we have $d_G(v,F^{\operatorname{t}}\cap F^{\operatorname{c}})\le 2$. For any $v\in E_{j,\cdot<\Delta(G)/2}\cap F^{\operatorname{c}}$  for some $j\in\{0,1\}$ we have by Proposition~\ref{neighbor_plaquette_prop} that $S(v)\cap E_{1-j,\cdot\ge\Delta(G)/2}\neq\emptyset$. Thus $d_G(v,F^{\operatorname{t}}\cap F^{\operatorname{c}})\le 3$ for all $v\in F^{\operatorname{c}}$. In particular, $F^{\operatorname{t}}\cap F^{\operatorname{c}}\neq\emptyset$ since $F^{\operatorname{c}}$ is non-empty, proving part (a) of the lemma.

Fix $v\in F^{\operatorname{c}}$, $v_s,v_t\in (F^{\operatorname{t}}\cap F^{\operatorname{c}})\cup\{v\}$ and let $v_s=v_1,v_2,\ldots, v_m=v_t$ be a $\GC$-path of vertices of $F^{\operatorname{c}}$. For each $2\le i\le m-1$, let $v_i'\in F^{\operatorname{t}}\cap F^{\operatorname{c}}$ be such that $d_G(v_i',v_i)\le 3$. We also take $v_1'=v_s$ and $v_m'=v_t$. It follows that for each $1\le i\le m-1$,
\begin{equation*}
d_G(v_i', v_{i+1}')\le d_G(v_i', v_i) + d_G(v_i, v_{i+1}) + d_G(v_{i+1},v_{i+1}')\le 3+2+3=8.
\end{equation*}
Hence $v_s=v_1',v_2',\ldots,v_m'=v_t$ is a $G^{\otimes 8}$-walk, proving part (b) of the lemma.
\end{proof}

\begin{lemma}\label{n_gamma_count_lemma}
Given $M,R\in \N$ and $E\subseteq V[G]$ with $|E|=M$, we have
\begin{equation*}
\big|\big\{N_\Gamma(E)\ \big|\ \Gamma\in\OMCut(x,b)\text{ satisfies }R_\Gamma(E)=R\big\}\big|\le (3d)^{M+R}.
\end{equation*}
\end{lemma}
\begin{proof}
We use the fact that for $v\in E$, given $P_\Gamma(v)$, the number of possibilities for $N_\Gamma(v)$ is at most $\binom{\Delta(G)}{P_\Gamma(v)}$ since we need only choose the directions $1\le i\le \Delta(G)$ for which $v+f_i\in\comp(\Gamma,x)$ and in the case that $v\in V^\odd$, these are the directions for which $\{v,v+f_i\}\notin \Gamma$ and in the case that $v\in V^\even$, these are the directions for which $\{v,v+f_i\}\in\Gamma$. Let
\begin{equation*}
\Omega:=\left\{X\in \{0,\ldots, \Delta(G)\}^E\ \big|\ \sum_{v\in E} \min(X_v,\Delta(G)-X_v)=R\right\}.
\end{equation*}
Then if $\Gamma\in \OMCut(x,b)$ satisfies $R_\Gamma(E)=R$ then $P_\Gamma(E)\in\Omega$. Hence
\begin{equation*}
\begin{split}
\big|\big\{N_\Gamma(E)&\ \big|\ \Gamma\in\OMCut(x,B)\text{ satisfies }R_\Gamma(E)=R\big\}\big|\le\\
&\le \sum_{X\in\Omega} \prod_{v\in E} \binom{\Delta(G)}{X_v}\le \sum_{X\in\Omega} \prod_{v\in E} (2d)^{\min(X_v,\Delta(G)-X_v)}=\\
&= (2d)^R|\Omega|\le (2d)^R(2d+1)^M\le (3d)^{M+R}.\qedhere
\end{split}
\end{equation*}
\end{proof}

We are finally ready for
\begin{proof}[Proof of Theorem~\ref{interior_approximation_theorem}]
Fix $L\in\N$ and define
\begin{equation*}
\OMCut(x,b,L):=\big\{\Gamma\in\OMCut(x,b)\ \big|\ |\Gamma|=L\big\}.
\end{equation*}
By Proposition~\ref{trivial_gamma_prop}, if $\OMCut(x,b,L)\neq\emptyset$ we must either have $L=\Delta(G)$ in which case $|\OMCut(x,b,L)|=2$ or $L\ge\frac{\Delta(G)^2}{2}$. The theorem follows simply when $L=\Delta(G)$ by taking $\mathcal{E}:=\{E_1(\Gamma)\ |\ \Gamma\in\OMCut(x,b,\Delta(G))\}$. Thus we assume henceforth that $L\ge\frac{\Delta(G)^2}{2}$. We note that then $\OMCut(x,b,L)$ consists only of non-trivial cutsets.

Define a function $S$ on $\OMCut(x,b,L)$ by
\begin{equation*}
S(\Gamma) := (E_0^{\operatorname{t}}, E_1^{\operatorname{t}}, N_\Gamma(E_0^{\operatorname{t}}), N_\Gamma(E_1^{\operatorname{t}}))
\end{equation*}
where $E_0^{\operatorname{t}}, E_1^{\operatorname{t}}$ are some sets satisfying the requirements of Proposition~\ref{dom_set_prop2} (arbitrarily chosen from the possible sets) and $N_\Gamma(E_0^{\operatorname{t}}), N_\Gamma(E_1^{\operatorname{t}})$ are defined after Proposition~\ref{dom_set_prop2}.
We shall use the notation $E_{0,\Gamma}^{\operatorname{t}}, E_{1,\Gamma}^{\operatorname{t}}, N_\Gamma(E_{0,\Gamma}^{\operatorname{t}}), N_\Gamma(E_{1,\Gamma}^{\operatorname{t}})$ for the components of $S(\Gamma)$. We define $\mathcal{E}$ to be the family of sets $E$ obtained by running the algorithm appearing before Proposition~\ref{E_alg_prop} on each vector in $S(\OMCut(x,b,L))$. Proposition~\ref{E_alg_prop} ensures that the $\mathcal{E}$ thus defined satisfies the requirements of the theorem. Since $|\mathcal{E}|\le |S(\OMCut(x,b,L))|$, the rest of the proof is devoted to bounding $|S(\OMCut(x,b,L))|$.

We start by partitioning $\OMCut(x,b,L)$ into \emph{types}. We say that $\Gamma\in\OMCut(x,b,L)$ has type $\gamma$, where $\gamma:=(k,(M_i)_{i=1}^k,(M_i^{\operatorname{t}})_{i=1}^k, (R_i^{\operatorname{t}})_{i=1}^k)$ for integers $k,(M_i)_{i=1}^k,(M_i^{\operatorname{t}})_{i=1}^k, (R_i^{\operatorname{t}})_{i=1}^k$, if $E_0(\Gamma)\cup E_1(\Gamma)$ has exactly $k$ $\GC$-connected components $F_1,\ldots, F_k$ (ordered in some predetermined manner) and for each $1\le i\le k$ we have $|F_i|=M_i$, $|F_i\cap (E_{0,\Gamma}^{\operatorname{t}}\cup E_{1,\Gamma}^{\operatorname{t}})|=M_i^{\operatorname{t}}$ and $R_\Gamma(F_i\cap (E_{0,\Gamma}^{\operatorname{t}}\cup E_{1,\Gamma}^{\operatorname{t}}))=R_i^{\operatorname{t}}$. Let $\TT$ be the set of possible types for $\Gamma\in\OMCut(x,b,L)$ and for $\gamma\in\TT$, denote $\OMCut(x,b,L,\gamma):=\{\Gamma\in\OMCut(x,b,L)\ |\ \text{$\Gamma$ has type $\gamma$}\}$. The following sequence of claims proves the theorem (it follows directly from claim 5).

{\bf Claim 1:} For any $\gamma=(k,(M_i)_{i=1}^k,(M_i^{\operatorname{t}})_{i=1}^k, (R_i^{\operatorname{t}})_{i=1}^k)\in\TT$ and any $1\le i\le k$ we have $R_i^{\operatorname{t}}\ge M_i^{\operatorname{t}}\ge 1$.

Fix $\gamma=(k,(M_i)_{i=1}^k,(M_i^{\operatorname{t}})_{i=1}^k, (R_i^{\operatorname{t}})_{i=1}^k)\in\TT$, $\Gamma\in\OMCut(x,b,L,\gamma)$ and $1\le i\le k$. Part (a) of Lemma~\ref{8_connect_lemma} implies that $M_i^{\operatorname{t}}\ge 1$. Since $E_{j,\Gamma}^{\operatorname{t}}\subseteq E_j(\Gamma)$ for $j\in\{0,1\}$ by Proposition~\ref{dom_set_prop2}, we have $P_\Gamma(v)\ge 1$ for all $v\in E_{0,\Gamma}^{\operatorname{t}}\cup E_{1,\Gamma}^{\operatorname{t}}$. Hence, part 2 of Proposition~\ref{trivial_gamma_prop} implies that $R_i^{\operatorname{t}}\ge M_i^{\operatorname{t}}$.

{\bf Claim 2:} $|\TT|\le 2L^3\exp\left(\frac{C\log d}{d^{3/2}}L\right)$ for some $C>0$.

For every $\gamma=(k,(M_i)_{i=1}^k,(M_i^{\operatorname{t}})_{i=1}^k, (R_i^{\operatorname{t}})_{i=1}^k)\in\TT$ and $\Gamma\in\OMCut(x,b,L,\gamma)$ we obtain using Claim~1,
\begin{align}
\sum_{i=1}^k M_i&\le L,\label{M_i_inequality}\\
\sum_{i=1}^k M_i^{\operatorname{t}}&\le \sum_{i=1}^k R_i^{\operatorname{t}}= R_\Gamma(E_{0,\Gamma}^{\operatorname{t}}\cup E_{1,\Gamma}^{\operatorname{t}})\le \frac{C'\log d}{d^{3/2}}L,\label{M_i_t_inequality}\\
\sum_{i=1}^k R_i^{\operatorname{t}}&=R_\Gamma(E_{0,\Gamma}^{\operatorname{t}}\cup E_{1,\Gamma}^{\operatorname{t}})\le \frac{C'\log d}{d^{3/2}}L\label{R_i_t_inequality}
\end{align}
for some $C'>0$, where we used part (a) of Proposition~\ref{dom_set_prop2} to bound $R_\Gamma(E_{0,\Gamma}^{\operatorname{t}}\cup E_{1,\Gamma}^{\operatorname{t}})$. These inequalities imply that the number of $\gamma\in\TT$ having $k=1$ is at most $L^3$ (for $d$ sufficiently large). Next, we note that if $\gamma$ has $k\ge2$ then $M_i\ge\alpha$ for all $1\le i\le k$ by Proposition~\ref{full_proj_prop}. We also note that for all $\gamma$ and $1\le i\le k$, $R_i^{\operatorname{t}}\ge M_i^{\operatorname{t}}\ge 1$ by Claim~1. Hence, applying Proposition~\ref{L_decomp_prop} with $s_1=\alpha, s_2=\alpha+1$ to \eqref{M_i_inequality} and applying it again with $s_1=1, s_2=2$ to \eqref{M_i_t_inequality} and \eqref{R_i_t_inequality}, we see that the number of $\gamma\in\TT$ having $k\ge 2$ is at most $L^3\exp\left(\frac{6L\log(\alpha+1)}{\alpha+1}+\frac{C''L\log d}{d^{3/2}}\right)\le L^3\exp\left(\frac{2C''L\log d}{d^{3/2}}\right)$ for some $C''>0$ and $d$ sufficiently large. Together with the bound on the number of $\gamma\in\TT$ having $k=1$, this proves the claim.

{\bf Claim 3:} For every $M>0$, there exists $A\subseteq V[G]$ with $|A|\le 40Mn_d^{1(M\ge\alpha)}$ such that for every $\Gamma\in\OMCut(x,b,L)$ and every $\GC$-connected component $F^{\operatorname{c}}$ of $E_0(\Gamma)\cup E_1(\Gamma)$ with $|F^{\operatorname{c}}|\le M$, we have $F^{\operatorname{c}}\cap A\neq\emptyset$.

The claim follows directly from Proposition~\ref{point_in_cutset_comp_prop} by noting that each such $F^{\operatorname{c}}$ contains a $\GC$-connected component of $E_1(\Gamma)$.

{\bf Claim 4:} There exists $C>0$ such that for each $\gamma\in \TT$,
\begin{equation*}
|S(\OMCut(x,b,L,\gamma))|\le L n_d^{\left\lfloor\frac{L}{\alpha}\right\rfloor}\exp\left(\frac{C\log^2 d}{d^{3/2}}L\right).
\end{equation*}

Denote $\gamma:=(k, (M_i)_{i=1}^k, (M_i^{\operatorname{t}})_{i=1}^k, (R_i^{\operatorname{t}})_{i=1}^k)$ and for $1\le i\le k$, let $A_i$ be the set of Claim~3 corresponding to $M=M_i$. For $p:=(E_0^{\operatorname{t}},E_1^{\operatorname{t}},N_0,N_1)\in S(\OMCut(x,b,L,\gamma))$ we pick an arbitrary $\Gamma(p)\in \OMCut(x,b,L,\gamma)$ such that $S(\Gamma(p))=p$. Let $F_1(p),\ldots, F_k(p)$ be the $\GC$-connected components (ordered in the same predetermined manner as before) of $E_0(\Gamma(p))\cup E_1(\Gamma(p))$. The vector $p$ is uniquely described by specifying the following for each $1\le i\le k$:
\begin{enumerate}
\item A point $v_i\in A_i\cap F_i(p)$.
\item The set $F_i^{\operatorname{t}}:=F_i(p)\cap (E_0^{\operatorname{t}}\cup E_1^{\operatorname{t}})$ (which has $|F_i^{\operatorname{t}}|=M_i^{\operatorname{t}}$ and $R_\Gamma(F_i^{\operatorname{t}})=R_i^{\operatorname{t}}$).
\item For each $v\in F_i^{\operatorname{t}}$, whether it is in $E_0^{\operatorname{t}}$ or in $E_1^{\operatorname{t}}$.
\item The set $N_{\Gamma(p)}(F_i^{\operatorname{t}})$.
\end{enumerate}
Hence we may bound $|S(\OMCut(x,b,L,\gamma))|$ by bounding the number of possibilities for each item of the above list, given its predecessors. For fixed $1\le i\le k$, we have at most $|A_i|\le 40M_i n_d^{1(M_i\ge\alpha)}$ possibilities for the first item. By part (b) of Lemma~\ref{8_connect_lemma} and Proposition~\ref{span_tree_prop}, we have at most $(2d)^{18M_i^{\operatorname{t}}}$ possibilities for the second item (given the point $v_i$). We have at most $2^{M_i^{\operatorname{t}}}$ possibilities for the third item. By Lemma~\ref{n_gamma_count_lemma}, we have at most $(3d)^{M_i^{\operatorname{t}}+R_i^{\operatorname{t}}}$ possibilities for the fourth item. Thus, for a given $1\le i\le k$, we have at most
\begin{equation*}
40M_i n_d^{1(M_i\ge\alpha)}(2d)^{18M_i^{\operatorname{t}}}2^{M_i^{\operatorname{t}}}(3d)^{M_i^{\operatorname{t}}+R_i^{\operatorname{t}}}\le M_i n_d^{\left\lfloor\frac{M_i}{\alpha}\right\rfloor}\exp(CR_i^{\operatorname{t}}\log d)
\end{equation*}
possibilities for the above list for some $C>0$, where we used that $R_i^{\operatorname{t}}\ge M_i^{\operatorname{t}}\ge 1$ by Claim~1. Hence, multiplying over all $i$, denoting $R^{\operatorname{t}}:=\sum_{i=1}^k R_i^{\operatorname{t}}=R_{\Gamma(p)}(E_0^{\operatorname{t}}\cup E_1^{\operatorname{t}})$ and noting that $R^{\operatorname{t}}\le \frac{C'\log d}{d^{3/2}}L$ for some $C'>0$ by Proposition~\ref{dom_set_prop2}, we find
\begin{equation*}
\begin{split}
|S(&\OMCut(x,b,L,\gamma))|\le \prod_{i=1}^k M_i n_d^{{\left\lfloor\frac{M_i}{\alpha}\right\rfloor}}\exp(CR_i^{\operatorname{t}}\log d)\le\\
&\le \left(\prod_{i=1}^k M_i\right) n_d^{\left\lfloor\frac{M}{\alpha}\right\rfloor}\exp(CR^{\operatorname{t}}\log d)\le \left(\prod_{i=1}^k M_i\right) n_d^{\left\lfloor\frac{L}{\alpha}\right\rfloor}\exp\left(\frac{C''\log^2 d}{d^{3/2}}L\right)
\end{split}
\end{equation*}
for some $C''>0$. Finally, noting that $\sum_{i=1}^k M_i\le L$ and that if $k\ge 2$ then by Proposition~\ref{full_proj_prop}, $M_i\ge\alpha\ge 2^{d-1}$ for all $1\le i\le k$, we deduce that
\begin{equation*}
\left(\prod_{i=1}^k M_i\right)\le L\left(\prod_{i=2}^k M_i\right)\le L\exp\left(\frac{\tilde{C}\log^2 d}{d^{3/2}}\sum_{i=2}^k M_i\right)\le L\exp\left(\frac{\tilde{C}\log^2 d}{d^{3/2}}L\right)
\end{equation*}
for some $\tilde{C}>0$ and sufficiently large $d$, from which the claim follows.

{\bf Claim 5:} There exists $C>0$ such that
\begin{equation*}
|S(\OMCut(x,b,L))|\le
n_d^{\left\lfloor\frac{L}{\alpha}\right\rfloor}\exp\left(\frac{C\log^2
d}{d^{3/2}}L\right).
\end{equation*}

By Claims 2 and 4 we have
\begin{equation*}
\begin{split}
|S(\OMCut&(x,b,L))|= \sum_{\gamma\in\TT} |S(\OMCut(x,b,L,\gamma))|\le\\
&\le L n_d^{\left\lfloor\frac{L}{\alpha}\right\rfloor}\exp\left(\frac{C'\log^2 d}{d^{3/2}}L\right)|\TT|\le 2L^4 n_d^{\left\lfloor\frac{L}{\alpha}\right\rfloor}\exp\left(\frac{C''\log^2 d}{d^{3/2}}L\right)
\end{split}
\end{equation*}
for some $C', C''>0$. The claim follows since $L\ge\frac{\Delta(G)^2}{2}\ge\frac{d^2}{2}$.\qedhere
\end{proof}

\paragraph{Proof of Theorem~\ref{count_level_sets_thm}}
Throughout the proof, we fix $L\in\N$, $x\in V[G]$, legal boundary conditions $(B,\mu)$ with non-positive $\mu$ and $g\in\Hom(G,B,\mu)$ and we shall always assume that $d\ge d_0$ for some large constant $d_0$.
An important step in proving our theorem is to prove a slightly
stronger version of it for the case $B=\{b\}$ for some $b\in V[G]$.
We start with two definitions. As in the previous section, we set
$\OMCut(x,B,L):=\big\{\Gamma\in\OMCut(x,B)\ \big|\
|\Gamma|=L\big\}$. We also set, for $v\in V[G]$, $\triv_v$ to be the
set of edges incident to $v$ (so that $|\triv_v|=\Delta(G)$). We
then have
\begin{proposition}\label{count_level_sets_prop}
There exists $C>0$ such that if $B=\{b\}$ for some $b\in V[G]$ then
\begin{equation*}
|\PLS(g,x,B,L)\setminus \{\triv_x\}|\le
n_d^{\left\lfloor\frac{L}{\alpha}\right\rfloor}\exp\left(\frac{C\log^2
d}{d^{3/2}}L\right).
\end{equation*}
\end{proposition}
\begin{proof}
Fix $b\in V[G]$ and assume $B=\{b\}$. If $L=\Delta(G)$ then by
Proposition~\ref{trivial_gamma_prop}, $\OMCut(x,b,L)$ contains at
most two elements: $\triv_x$ and $\triv_b$. Since
$\PLS(g,x,B,L)\subseteq\OMCut(x,b,L)$, the proposition follows.

Assume now that $L\neq \Delta(G)$. Using Proposition~\ref{trivial_gamma_prop} again, we see that we may assume that $L\ge\frac{\Delta(G)^2}{2}$ since otherwise $\OMCut(x,b,L)=\emptyset$. Assume this and let $f\in\Omega_{x,L}$ be such that $g$ is a $(x,B)$-interior modification of $f$. Denote $\Gamma:=\LS(f,x,B)$. We claim that given any set $E\subseteq V[G]$ which is an interior approximation to $\Gamma$, we may recover $\Gamma$ as a function only of $g$ and $E$. Letting $\mathcal{E}$ be the family of Theorem~\ref{interior_approximation_theorem}, this implies that for some $C,C'>0$,
\begin{equation*}
|\PLS(g,x,B,L)|\le |\mathcal{E}|\le 2n_d^{\left\lfloor\frac{L}{\alpha}\right\rfloor}\exp\left(\frac{C\log^2 d}{d^{3/2}}L\right)\le n_d^{\left\lfloor\frac{L}{\alpha}\right\rfloor}\exp\left(\frac{C'\log^2 d}{d^{3/2}}L\right),
\end{equation*}
since $L\ge\frac{\Delta(G)^2}{2}$, proving the proposition. To see this claim, fix an interior approximation $E$ to $\Gamma$. Let $A$ be the connected component of $b$ in $\{v\in V[G]\ |\ f(v)\le 0\}$ and $A'$ be the connected component of $b$ in $\{v\in V[G]\setminus E\ |\ g(v)\le 0\}$. Since $\Gamma$ is, by definition, all edges between $A$ and the connected component of $x$ in $V[G]\setminus A$ and since $A'$ is determined solely from $g$ and $E$, it is sufficient to show that $A=A'$.
To see this, recall that $g(v)=f(v)$ for every $v\notin\comp(\Gamma,x)$ and $g(v)=f(v)=1$ for $v\in \Eonee(\Gamma)$. This implies $A'\supseteq A$ since $A\cap\comp(\Gamma,x)=\emptyset$ and $E\subseteq\comp(\Gamma,x)$. Next, note that by $\Gamma$'s definition, every $w\in V[G]\setminus A$ such that $w\adj{G} v$ for some $v\in A$ satisfies $f(w)=1$ and either $w\notin\comp(\Gamma,x)$, $w\in \Eonee(\Gamma)$ or $w\in E_1(\Gamma)\setminus \Eonee(\Gamma)$. In the first two cases, $g(w)=1$ implying $w\notin A'$ and in the third case, $w\notin A'$ since $E_1(\Gamma)\setminus \Eonee(\Gamma)\subseteq E$. Thus $A'\subseteq A$, as required.
\end{proof}

We proceed to prove the theorem for the case of general $B$. As in Section~\ref{count_cutsets_thm_proof}, we fix a total order $\prec$ on $V[G]$ and for $\Gamma\in\OMCut(x,B)$, define inductively $m(\Gamma)$ and a vector $(b^i)_{i=1}^{m(\Gamma)}\subseteq B$ as follows: $b^1$ is the $\prec$-smallest element of $B$. Assuming that $(b^1,\ldots, b^\ell)$ have already been defined, we set $m(\Gamma):=\ell$ if $B\subseteq\cup_{i=1}^\ell \comp(\Gamma,b^i)$ or otherwise set $b^{\ell+1}$ to be the $\prec$-smallest element of $B\setminus \cup_{i=1}^\ell\comp(\Gamma,b^i)$. For $1\le i\le m(\Gamma)$, we also set $\Gamma^i:=\subcut(\Gamma,b^i)$ and $L^i:=|\Gamma^i|$. Note that $\Gamma^i\in\OMCut(x,b^i)$ by the remark after the definition of $\OMCut$ in Section~\ref{preliminaries_sec}. In this section, we define the type of $\Gamma$ to be the vector $\gamma:=(m(\Gamma),(L^i)_{i=1}^{m(\Gamma)})$. As in the previous section, we set
\begin{equation*}
\begin{split}
&\TT:=\text{set of possible types for $\Gamma\in\OMCut(x,B,L)$},\\
&\OMCut(x,B,L,\gamma):=\{\Gamma\in\OMCut(x,B,L)\ \big|\ \text{$\Gamma$ has type $\gamma$}\} \quad \text{for $\gamma\in\TT$.}
\end{split}
\end{equation*}
We note that since for any $\Gamma\in\OMCut(x,B,L)$ with type $(m(\Gamma),(L^i)_{i=1}^{m(\Gamma)})$ we have $\sum_{i=1}^{m(\Gamma)} L^i = L$ and for each $i$, $L^i=\Delta(G)$ or $L^i\ge\frac{\Delta(G)^2}{2}$ by Proposition~\ref{trivial_gamma_prop}, it follows from Proposition~\ref{L_decomp_prop} with $s_1=\Delta(G)$ and $s_2=\left\lceil\frac{\Delta(G)^2}{2}\right\rceil$ that
\begin{equation}\label{type_bound}
|\TT|\le \exp\left(\frac{30\log d}{d^2}L\right)
\end{equation}
for sufficiently large $d$. We proceed to state and prove several lemmas from which the theorem will follow.
\begin{lemma}\label{local_transformation_hereditary_lem}
If $g$ is a $(x,B)$-interior modification of some $f\in\Omega_{x,L}$, then $g$ is also a $(x,b)$-interior modification of $f$ for every $b\in B$.
\end{lemma}
\begin{proof}
Let $f\in\Omega_{x,L}$ be such that $g$ is a $(x,B)$-interior modification of $f$ and let $b\in B$. Denoting $\Gamma:=\LS(f,x,B)$ and $\Gamma_b:=\LS(f,x,b)$ we have $\Gamma_b\subseteq\Gamma$ and hence $\comp(\Gamma_b,x)\supseteq\comp(\Gamma,x)$. Since $g(v)=f(v)$ for every $v\notin\comp(\Gamma,x)$, this holds in particular for every $v\notin\comp(\Gamma_b,x)$. Furthermore, for every $v\in V[G]$ we have $P_{\Gamma_b}(v)\le P_\Gamma(v)$. Hence, if $v\in \Eonee(\Gamma_b)$ then $P_\Gamma(v)\ge\Delta(G)-\sqrt{d}$. Such a $v$ must belong to $\comp(\Gamma,x)$ by Proposition~\ref{subcut_identical_prop} (one can also see this since $\Gamma_b\in\OMCut(x,b)$ and $v\in V^\odd$). Thus, $v\in \Eonee(\Gamma)$ implying $g(v)=f(v)=1$, as required.
\end{proof}

\begin{lemma}\label{b_determined_lem2}
There is a function satisfying that for every $\Gamma\in\OMCut(x,B)$, the function takes as input $1\le j<m(\Gamma)$, $(b^1,\ldots, b^j)$ and $(\subcut(\Gamma,b^i))_{i\le j}$ and returns $b^{j+1}$.
\end{lemma}
\begin{proof}
Knowing $(\subcut(\Gamma,b^i))_{i\le j}$ determines $\comp(\Gamma,b^i)$ for all $1\le i\le j$. By our definitions, $b^{j+1}$ is the $\prec$-smallest element of $B\setminus \cup_{i=1}^j\comp(\Gamma,b^i)$.
\end{proof}

\begin{lemma}\label{PLS_gamma_estimate_lem}
There exists $C>0$ such that for all $\gamma\in\TT$ we have
\begin{equation*}
|\PLS(g,x,B,L)\cap \OMCut(x,B,L,\gamma)|\le 2n_d^{\left\lfloor\frac{L}{\alpha}\right\rfloor}\exp\left(\frac{C\log^2 d}{d^{3/2}}L\right).
\end{equation*}
\end{lemma}
\begin{proof}
Fix $\gamma:=(k,(L^i)_{i=1}^k)\in\TT$. We start by assuming that
some $\Gamma\in\OMCut(x,B,L,\gamma)$ has $\triv_x\subseteq \Gamma$.
It follows that $\Gamma=\triv_x$ by minimality of $\Gamma$, $k=1$
and $L=\Delta(G)$. Let also $b\in B$. Since for any
$\Gamma\in\OMCut(x,B,L,\gamma)$ we have $\subcut(\Gamma,b)\subseteq
\Gamma$ and $\subcut(\Gamma,b)\in\OMCut(x,b)$ (see
Proposition~\ref{subcut_identical_prop} and remark after the
definition of $\OMCut$ in Section~\ref{preliminaries_sec}), it
follows from Proposition~\ref{trivial_gamma_prop} that
$\OMCut(x,B,L,\gamma)$ can contain at most two elements: $\triv_x$
and $\triv_b$. Thus the lemma follows in this case.

We assume henceforth that no $\Gamma\in\OMCut(x,B,L,\gamma)$ has
$\triv_x\subseteq \Gamma$. We note the following facts: A $\Gamma\in
\PLS(g,x,B,L)\cap \OMCut(x,B,L,\gamma)$ is uniquely described by
specifying $(\subcut(\Gamma,b^i))_{i\le k}$. By definition, $b^1$ is
the $\prec$-smallest element of $B$ and by
Lemma~\ref{b_determined_lem2}, for each $1\le j< k$, $b^{j+1}$ is
determined as a function of $(b^i)_{i\le j}$ and
$(\subcut(\Gamma,b^i))_{i\le j}$. By
Lemma~\ref{local_transformation_hereditary_lem} and
Proposition~\ref{count_level_sets_prop}, for each $1\le i\le k$, the
number of possibilities for $\subcut(\Gamma,b^i)$, other than
$\triv_x$, given $g$, $b^i$ and $L^i$ is at most
\begin{equation*}
n_d^{\left\lfloor\frac{L^i}{\alpha}\right\rfloor}\exp\left(\frac{C\log^2 d}{d^{3/2}}L^i\right)
\end{equation*}
for some $C>0$. Putting these facts together, we see that
\begin{equation*}
\begin{split}
|\PLS(g,x,B,L)\cap \OMCut(x,B,L,\gamma)|&\le \prod_{i=1}^k n_d^{\left\lfloor\frac{L^i}{\alpha}\right\rfloor}\exp\left(\frac{C\log^2 d}{d^{3/2}}L^i\right)\le\\
&\le n_d^{\left\lfloor\frac{L}{\alpha}\right\rfloor}\exp\left(\frac{C\log^2 d}{d^{3/2}}L\right)
\end{split}
\end{equation*}
as required.
\end{proof}

\begin{proof}[Proof of Theorem~\ref{count_level_sets_thm}]
By Lemma~\ref{PLS_gamma_estimate_lem} and \eqref{type_bound}, we have
\begin{equation*}
\begin{split}
|\PLS(g,x,B,L)|&=\sum_{\gamma\in\TT} |\PLS(g,x,B,L)\cap\OMCut(x,B,L,\gamma)|\le\\
&\le 2n_d^{\left\lfloor\frac{L}{\alpha}\right\rfloor}\exp\left(\frac{C\log^2 d}{d^{3/2}}L\right)|\TT|\le 2n_d^{\left\lfloor\frac{L}{\alpha}\right\rfloor}\exp\left(\frac{C'\log^2 d}{d^{3/2}}L\right)
\end{split}
\end{equation*}
for some $C'>0$, as required.
\end{proof}

\subsection{Proof of Theorem~\ref{transformation_thm}}\label{proof_of_transform_theorem_sec}
In this section we prove the first part of Theorem~\ref{transformation_thm} for the transformation $T$ of Section~\ref{T_def_section}. The second part was proved in Section~\ref{T_def_section}.

Fix $d$ large enough for the arguments of the section, a non-linear torus $G$, legal boundary conditions $(B,\mu)$ with non-positive $\mu$, $x\in V[G], L\in\N$ and $\emptyset\neq \Omega\subseteq\Omega_{x,L}$. For $f\in \Omega$, introduce the notation $E_{1,1}(f)$ and $\Eonee(f)$ for $E_{1,1}(\LS(f,x,B))$ and $\Eonee(\LS(f,x,B))$ respectively.
Recall the role of $\lambda$ from $\eqref{T_def}$. Denote $M_\lambda:=\left\lceil\left(1-\frac{\lambda}{\log^2 d}\right)\frac{L}{\Delta(G)}\right\rceil$ and for $0\le k\le m <M_\lambda$, let
\begin{equation*}
\begin{split}
\Omega_{x,L,1}&:=\left\{f\in\Omega_{x,L}\ |\ |\Eonee(f)|\ge M_\lambda\right\},\\
\Omega_{x,L,2,m,k}&:=\{f\in \Omega_{x,L}\ |\ |\Eonee(f)|=m,\ |E_{1,1}(f)\cap \Eonee(f)|=k\}.
\end{split}
\end{equation*}
Note that $\Omega_{x,L}=\Omega_{x,L,1}\cup\left(\cup_{0\le k\le m<M_\lambda} \Omega_{x,L,2,m,k}\right)$. From \eqref{T_def} and \eqref{T_image_size_eq}, we have
\begin{align}
T(f) &= T_1(f)\text{ for $f\in\Omega_{x,L,1}$},\\
T(f) &= T_2(f)\text{ for $f\in\Omega_{x,L,2,m,k}$},\\
|T(f)| &= \begin{cases}2^{\frac{L}{\Delta(G)}}&\text{if }f\in\Omega_{x,L,1}\\2^{\frac{L}{\Delta(G)}-k}&\text{if }f\in\Omega_{x,L,2,m,k}\label{T_size_on_Omega}\end{cases}.
\end{align}
We note that
\begin{equation}\label{Omega_decomp}
\begin{split}
\frac{|\Omega|}{|T(\Omega)|}&\le \frac{|\Omega\cap \Omega_{x,L,1}|}{|T_1(\Omega\cap\Omega_{x,L,1})|} + \sum_{0\le k\le m<M_\lambda}\frac{|\Omega\cap \Omega_{x,L,2,m,k}|}{|T_2(\Omega\cap\Omega_{x,L,2,m,k})|}
\end{split}
\end{equation}
where, as before, we interpret $T(\Omega):=\cup_{f\in\Omega} T(f)$ and $\frac{0}{0}=0$. We also have the simple
\begin{lemma}\label{transformation_inverse_lemma}
Let $N,M>0$, $X,Y$ be finite sets and $R:X\to\PP(Y)$ a function satisfying for each $f\in X$ and $g\in Y$,
\begin{equation*}
|R(f)|\ge N\text{ and }|\{h\in X\ |\ g\in R(h)\}|\le M.
\end{equation*}
Then for each $\emptyset\neq X'\subseteq X$ we have $\frac{|X'|}{|\cup_{f\in X'}R(f)|}\le \frac{M}{N}$.
\end{lemma}
\begin{proof}
It is straightforward that $|\cup_{f\in X'}R(f)|\ge \frac{N}{M}|X'|$, implying the lemma.
\end{proof}
Aiming to use this lemma to estimate the RHS of \eqref{Omega_decomp}, we will show
\begin{proposition}\label{Omega_x_L_1_bound_prop}
For $\lambda\ge \frac{\sqrt d \log^2 d}{\Delta(G)}$ and $g\in T_1(\Omega_{x,L,1})$, we have $|\{f\in \Omega_{x,L,1}\ |\ g\in T_1(f)\}|\le \left(1+\lambda L\right)^2 \exp\left(\frac{C\lambda L}{d}\right)$ for some $C>0$.
\end{proposition}
\begin{proposition}\label{Omega_x_L_2_bound_prop}
For $0\le k\le m<M_\lambda$ and $g\in T_2(\Omega_{x,L,2,m,k})$, we have $|\{f\in \Omega_{x,L,2,m,k}\ |\ g\in T_2(f)\}|\le 2^{m-k+1}\exp\left(\frac{CL}{d\log^3 d}\right)$ for some $C>0$.
\end{proposition}
We remark that these propositions are the only place in our proof
where the non-linearity of $G$ is used. Let us first show how these
propositions can be used to prove the (first part of the) theorem
and then proceed to prove them. The propositions along with
\eqref{T_size_on_Omega}, \eqref{Omega_decomp} and
Lemma~\ref{transformation_inverse_lemma} imply that for $\frac{\sqrt
d \log^2 d}{\Delta(G)}\le \lambda\le 1$ we have
\begin{equation*}
\begin{split}
\frac{|\Omega|}{|T(\Omega)|}&\le \frac{\left(1+\lambda L\right)^2 \exp\left(\frac{C\lambda L}{d} \right)}{2^{\frac{L}{\Delta(G)}}} + \sum_{0\le k\le m<M_\lambda} \frac{2^{m-k+1}\exp\left(\frac{CL}{d\log^3 d}\right)}{2^{\frac{L}{\Delta(G)}-k}}\le\\
&\le 4L^2\left(\exp\left(\frac{C\lambda L}{d}\right) + \exp\left(\frac{CL}{d\log^3 d}\right)2^{M_\lambda-1}\right)2^{-\frac{L}{\Delta(G)}}\le\\
&\le 4L^2\left(\exp\left(\frac{C\lambda L}{d}\right)2^{-\frac{L}{\Delta(G)}} + \exp\left(\frac{CL}{d\log^3 d}\right)2^{-\frac{\lambda L}{\Delta(G)\log^2 d}}\right).
\end{split}
\end{equation*}
Hence if $\lambda$ is a small enough constant (independent of $d$) and $d$ is sufficiently large, we have
\begin{equation*}
\frac{|\Omega|}{|T(\Omega)|}\le 8L^2 \exp\left(-\frac{cL}{d\log^2 d}\right)\le d^3\exp\left(-\frac{c'L}{d\log^2 d}\right)
\end{equation*}
for some $c,c'>0$, proving the theorem.
\begin{proof}[Proof of Proposition~\ref{Omega_x_L_1_bound_prop}]
Fix $g\in T_1(\Omega_{x,L,1})$. We note that for any $\Gamma\in\OMCut(x,B)$ with $|\Gamma|=L$ there is at most one $f\in \Omega_{x,L,1}$ such that $\LS(f,x,B)=\Gamma$ and $g\in T_1(f)$. This follows from the fact that if $f$ satisfies these two properties then we may recover it from $g$ by performing the inverse of the shift transformation $\shift$, i.e.,
\begin{equation}\label{inverse_shift_eq}
f(v)=\begin{cases}g(v-e_1)+1 &\text{for }v\in \CC_1\\g(v) &\text{otherwise}\end{cases}
\end{equation}
where $\CC_1:=\comp(\Gamma,x)$.
Let us verify this claim. First note that $g$ may differ from $\shift(f)$ only on $E_{1,1}(\Gamma)=\{v\in \CC_1\ |\ v+e_1\notin \CC_1\}$ and the values of $g$ on these points are not used in \eqref{inverse_shift_eq}. Next, note that for all $v\notin \CC_1$ we have $g(v)=f(v)$ and for all $v\in \CC_1$ such that $v-e_1\in \CC_1$ we have $g(v-e_1) = f(v)-1$ by definition of $\shift$. Finally note that if $v\in \CC_1$ is such that $v-e_1\notin \CC_1$ then necessarily $f(v)=1$ and $f(v-e_1)=g(v-e_1)=0$ by definition of $\LS(f,x,B)$. These facts prove \eqref{inverse_shift_eq}.

We deduce that
\begin{equation}\label{Omega_x_L_1_bound}
|\{f\in \Omega_{x,L,1}\ |\ g\in T_1(f)\}|\le |\{\LS(f,x,B)\ |\ f\in\Omega_{x,L,1}\}|.
\end{equation}
This is a rough bound since the RHS is independent of $g$, but we will see that it will suffice for this proposition because of the irregularities in $\LS(f,x,B)$ for $f\in\Omega_{x,L,1}$. For $\Gamma\in \OMCut(x,B)$, recalling the definition of $R_\Gamma$ from Section~\ref{rough_bdry_counting_sec}, we denote $M(\Gamma):=|E_1(\Gamma)|$, $R(\Gamma):=R_\Gamma(E_1(\Gamma))=\sum_{v\in E_1(\Gamma)} \min(P_\Gamma(v),\Delta(G)-P_\Gamma(v))$ and for $1\le i\le\Delta(G)$, $a_i(\Gamma):=|\{v\in E_1(\Gamma)\ |\ P_\Gamma(v)=i\}|$. Let $O:=\{\Gamma\in\OMCut(x,B)\ |\ |\Gamma|=L, |\Eonee(\Gamma)|\ge (1-\frac{\lambda}{\log^2 d})\frac{L}{\Delta(G)}\}$. By definition of $\Omega_{x,L,1}$ we have
\begin{equation}\label{O_dominates}
\{\LS(f,x,B)\ |\ f\in\Omega_{x,L,1}\}\subseteq O.
\end{equation}
We continue by estimating $M(\Gamma)+R(\Gamma)$ for $\Gamma\in O$. Note that for $\Gamma\in O$, $\sum_{i=1}^{\Delta(G)} ia_i(\Gamma) = L$ and $\sum_{i=\lceil \Delta(G)-\sqrt{d}\rceil}^{\Delta(G)} a_i(\Gamma)\ge (1-\frac{\lambda}{\log^2 d})\frac{L}{\Delta(G)}$. Hence for $\Gamma\in O$,
\begin{equation*}
\begin{split}
M(\Gamma)&+R(\Gamma)=\sum_{i=1}^{\Delta(G)} (1+\min(i,\Delta(G)-i))a_i(\Gamma)\le\\
&\le 2\sum_{i=1}^{\lceil \Delta(G)-\sqrt{d}\rceil-1} ia_i(\Gamma) + \sum_{i=\lceil \Delta(G)-\sqrt{d}\rceil}^{\Delta(G)} (1+\Delta(G)-i)a_i(\Gamma)\le\\
&\le 2\left(L-\sum_{i=\lceil \Delta(G)-\sqrt{d}\rceil}^{\Delta(G)}ia_i\right)+(1+\sqrt{d})\frac{L}{\Delta(G)-\sqrt{d}}\le\\
&\le 2L\left(1-\left(1-\frac{\lambda}{\log^2 d}\right)\frac{(\Delta(G)-\sqrt{d})}{\Delta(G)}\right)+(1+\sqrt{d})\frac{L}{\Delta(G)-\sqrt{d}}\le\\
&\le \left(\frac{2\lambda }{\log^2 d} + \frac{6\sqrt d}{\Delta(G)}\right)L\le \frac{8\lambda L}{\log^2 d},
\end{split}
\end{equation*}
taking $\lambda\ge \frac{\sqrt d \log^2 d}{\Delta(G)}$ in the last step.

By Theorem~\ref{count_cutsets_thm} and using that $G$ is non-linear, if $M,R\ge 0$ satisfy $M+R\le \frac{8\lambda L}{\log^2 d}$ then
\begin{equation*}
\begin{split}
|\{\Gamma\in O\ |\ M(\Gamma)=M,\ &R(\Gamma)=R\}|\le n_d^{\left\lfloor\frac{M}{\prod_{i=1}^{d-1} n_i}\right\rfloor}\exp\left(\frac{C\log^2 d}{d}R\right)\le\\
&\le \exp\left(\frac{C'\log^2 d}{d}(M+R)\right)\le \exp\left(\frac{C''\lambda}{d} L\right)
\end{split}
\end{equation*}
for some $C',C''>0$. Hence
\begin{equation}\label{O_estimate}
\begin{split}
|O|&\le \left|\left\{M,R\ge 0\ \Big|\ M+R\le \frac{8\lambda L}{\log^2 d}\right\}\right|\exp\left(\frac{C''\lambda}{d} L\right)\le\\
&\le \left(1+\lambda L\right)^2 \exp\left(\frac{C''\lambda}{d} L\right)
\end{split}
\end{equation}
for large enough $d$. The proposition follows from \eqref{Omega_x_L_1_bound}, \eqref{O_dominates} and \eqref{O_estimate}.\qedhere
\end{proof}
\begin{proof}[Proof of Proposition~\ref{Omega_x_L_2_bound_prop}]
Fix $0\le k\le m<M_\lambda$ and $g\in T_2(\Omega_{x,L,2,m,k})$.

We first claim that for any $\Gamma\in\OMCut(x,B)$ with $|\Gamma|=L$ and $s\in\{-1,1\}^{\Eonee(\Gamma)\setminus E_{1,1}(\Gamma)}$ there is at most one $f\in \Omega_{x,L,2,m,k}$ such that $\LS(f,x,B)=\Gamma$, $f(v+e_1)-1=s(v)$ for all $v\in \Eonee(\Gamma)\setminus E_{1,1}(\Gamma)$ and $g\in T_2(f)$. To see this, suppose $f$ is such a function. Define for $h\in\Hom(G,B,\mu)$ and $v\in V[G]$, as in Section~\ref{shift_flip_transformation_sec}, $R_v(h)$ to be the connected component of $v$ in $V[G]\setminus \{w\in V[G]\ |\ h(w)=0\}$. Recall from \eqref{flip_T_def} that if $g\in T_2(f)$ then there is an $h\in T_1(f)$ such that
\begin{equation}\label{g_from_h_def}
g(w):=\begin{cases}-h(w)&\text{if $w\in R_v(h)$ for some $v\in \Eonee(\Gamma)$ with $h(v)=-1$}\\h(w)&\text{otherwise}\end{cases}.
\end{equation}
Fixing this $h$ we note that, as discussed in Sections~\ref{shift_transformation_sec} and \ref{shift_flip_transformation_sec} (see Lemma~\ref{flip_possible_lemma}), for any $v\in E_{1,1}(\Gamma)$ we have $R_v(h)=\{v\}$ and flipping $h(v)$ to $-h(v)$ for such $v$ still results in a function in $T_1(f)$. Hence we may and will assume that $h(v)=1$ for all $v\in E_{1,1}(\Gamma)$ so that the flipping in \eqref{g_from_h_def} takes place only for $v\in \Eonee(\Gamma)\setminus E_{1,1}(\Gamma)$. We note also that for $v\in \Eonee(\Gamma)\setminus E_{1,1}(\Gamma)$ we have $h(v)=s(v)$ by our assumption on $f$ and the definition of $T_1$. Hence, keeping in mind that $R_v(h)=R_v(g)$ for all $v$ since the zero level set is unchanged by flipping, we see that $h$ may be recovered from $g$, given $\Gamma$ and $s$, by
\begin{equation*}
h(w)=\begin{cases}-g(w)&\text{if $w\in R_v(g)$ for some $v\in \Eonee(\Gamma)\setminus E_{1,1}(\Gamma)$}\\&\text{with $g(v)=-s(v)$}\\g(w)&\text{otherwise}\end{cases}.
\end{equation*}
As in the proof of Proposition~\ref{Omega_x_L_1_bound_prop}, we know that $f$ is determined from $h$ given $\Gamma$ (since $h\in T_1(f)$, see \eqref{inverse_shift_eq}) and hence $f$ is uniquely determined from $g$ given $\Gamma$ and $s$, as claimed.

Note that by definition of $\Omega_{x,L,2,m,k}$, if $f\in\Omega_{x,L,2,m,k}$ then $\Gamma=\LS(f,x,B)$ satisfies $|\Eonee(\Gamma)\setminus E_{1,1}(\Gamma)|=|\Eonee(\Gamma)|-|\Eonee(\Gamma)\cap E_{1,1}(\Gamma)|= m-k$. Recalling from Section~\ref{partial_edge_cutsets_section} the notation for $\PLS(g,x,B,L)$ and that $g$ is a $(x,B)$-interior modification of $f$ whenever $g\in T_2(f)$, it follows from Theorem~\ref{count_level_sets_thm} and the fact $G$ is non-linear that
\begin{equation*}
\begin{split}
|\{f\in &\Omega_{x,L,2,m,k}\ |\ g\in T_2(f)\}|\le |\PLS(g,x,B,L)|\cdot|\{-1,1\}^{m-k}|\le\\
&\le 2n_d^{\left\lfloor\frac{L}{\prod_{i=1}^{d-1} n_i}\right\rfloor}\exp\left(\frac{C\log^2 d}{d^{3/2}}L\right)2^{m-k}\le 2^{m-k+1}\exp\left(\frac{C'}{d\log^3 d}L\right),
\end{split}
\end{equation*}
for some $C'>0$. \qedhere
\end{proof}

\section{Isoperimetry, Height, Range and Lipschitz}\label{isoperim_height_range_Lip_sec}
\subsection{Isoperimetry}\label{general_isoperimetry_section}
For integer $r\ge 0$ and $v\in V[G]$, define the sphere and ball of radius $r$ around $v$ by $S_r(v):=\{w\in V[G]\ |\ d_{G}(v,w)=r\}$ (where $d_G$ is the graph distance in $G$) and $B_r(v):=\cup_{i=0}^r S_i(v)$. We also recall that $\vol(r)=|B_r(v)|$ (it is independent of $v$). Let also $E_r(v):=\{w\in B_r(v)\ |\ w+e_d\notin B_r(v)\}$. Finally, let $s_r$ denote the number of edges between $B_r(v)$ and its complement in $G$ ($s_r$ does not depend on $v$). Since we either have $S_r(v)\subseteq V^{\odd}$ or $S_r(v)\subseteq V^\even$, we have by Proposition~\ref{equal_edgenum_in_all_dir} that
\begin{equation}\label{s_r_a_r_relation}
s_r=\Delta(G)|E_r(v)|.
\end{equation}
For an integer $r\ge 0$, we define our isoperimetric functions as
\begin{equation*}
\begin{split}
&I_r(x,y):=\min\left\{|\Gamma|\ \big|\ \Gamma\in\OMCut(B_r(x),B_r(y))\right\}\qquad\quad (x,y\in V[G]),\\
&I_r:=\min\left\{I_r(x,y)\ |\ x,y\in V[G]\right\}\qquad\text{and}\\
&I_r(E):=\min\left\{|\Gamma|\ \big|\ y\in V[G], \Gamma\in\OMCut(E,B_r(y))\cup\OMCut(B_r(y),E)\right\}
\end{split}
\end{equation*}
for $\emptyset\neq E\subseteq V[G]$, where $I_r(x,y), I_r$ and $I_r(E)$ are defined to be infinity if the sets minimized over are empty. Recalling the definition of full projection sets from before Theorem~\ref{height_at_point_thm}, we will prove the following theorems in the next two sections.
\begin{theorem}\label{I_r_bound_thm}
For all integer $r\ge 0$ we have $I_r\ge \frac{s_r}{2\min(4(2r+1),\Delta(G))}$. Moreover, if $s_r\le (d-1)n_d$ then $I_r\ge s_r$.
\end{theorem}
\begin{theorem}\label{I_r_E_bound_thm}
For all integer $r\ge 0$ and full projection sets $E\subseteq V[G]$, we have $I_r(E)\ge s_r$.
\end{theorem}
In addition, we collect in Section~\ref{isoperimetric_relations_sec} several simple relations for $s_r$ and $\vol(r)$.

\subsubsection{Full Projection Isoperimetry}
In this section we prove Theorem~\ref{I_r_E_bound_thm}. Fix a full projection set $\emptyset\neq E\subseteq V[G]$ and let $1\le i_0\le d$ be such that in the coordinate system \eqref{G_coordinates}, every cycle of the form $\{w+k e_{i_0}\ |\ k\in\Z\}$, for $w\in V[G]$, intersects $E$. Fix an integer $r\ge 0$, $y\in V[G]$ and $\Gamma\in\OMCut(E,B_r(y))\cup\OMCut(B_r(y),E)$ (noting that if for all $y\in V[G]$, $\OMCut(E,B_r(y))\cup\OMCut(B_r(y),E)=\emptyset$, then the theorem is trivial). It is sufficient to show that $|\Gamma|\ge s_r$.

Let $E_{r,i_0}:=\{w\in B_r(y)\ |\ w+e_{i_0}\notin B_r(y)\}$. As in \eqref{s_r_a_r_relation}, we then have
\begin{equation}\label{s_r_a_r_i_0_relation}
s_r=\Delta(G)|E_{r,i_0}|.
\end{equation}
For $w\in E_{r,i_0}$, let $P(w):=\{w+k e_{i_0}\ |\ k\in\Z\}$ be the cycle in the $i_0$ direction passing through it. By the definition of $E_{r,i_0}$ and properties of balls in $G$, the cycles $P(w)$ and $P(w')$ do not intersect for distinct $w,w'\in E_{r,i_0}$. Since each such cycle intersects $E$ (by the full projection property), it follows that each such cycle must contain an edge of $\Gamma$. Thus, $\Gamma$ contains at least $|E_{r,i_0}|$ edges of the form $(v,v+e_{i_0})$ for $v\in\comp(\Gamma,y)$. Hence, by Proposition~\ref{equal_edgenum_in_all_dir} and \eqref{s_r_a_r_i_0_relation}, $|\Gamma|\ge \Delta(G)|E_{r,i_0}|=s_r$, as required.

\subsubsection{General Isoperimetry}
In this section we prove Theorem~\ref{I_r_bound_thm}. The moreover part of the theorem follows from the following proposition.
\begin{proposition}
For all integer $r\ge 0$, $I_r\ge\min\left(s_r, (d-1)n_d\right)$.
\end{proposition}
\begin{proof}
Fix an integer $r\ge 0$, $x,y\in V[G]$ and $\Gamma\in\OMCut(B_r(x),B_r(y))$. For each $v\in E_r(x)\cup E_r(y)$, let $P(v):=\{v+k e_d\ |\ k\in\Z\}$ be the cycle in $G$ going in the $e_d$ direction and passing through $v$.
We consider three cases:
\begin{enumerate}
\item For all $v\in E_r(x)$, $E(P(v))\cap\Gamma\neq \emptyset$ (where $E(P(v))$ are the edges of the cycle $P(v)$). In this case, since by definition of $E_r(x)$ we have for all $v\in E_r(x)$ that $P(v)\cap E_r(x)=\{v\}$, we deduce that $\Gamma$ contains at least $|E_r(x)|$ edges of the form $\{u,u+n_d\}$ for some $u\in V^\odd$. Hence, by Proposition~\ref{equal_edgenum_in_all_dir} and \eqref{s_r_a_r_relation}, we obtain $|\Gamma|\ge \Delta(G)|E_r(x)|=s_r$.
\item For all $v\in E_r(y)$, $E(P(v))\cap\Gamma\neq \emptyset$. As in the first case, we deduce $|\Gamma|\ge s_r$.
\item There exist $v\in E_r(x)$ and $w\in E_r(y)$ such that $\Gamma\cap E(P(v))=\Gamma\cap E(P(w))=\emptyset$. For $0\le k\le n_d-1$, let $G_k$ be the sub-torus induced by the vertices of $G$ with $d$'th coordinate equal to $k$. Let $v_k$ and $w_k$ be the intersection of $V[G_k]$ with $P(v)$ and $P(w)$ respectively. By our assumption, for each $0\le k\le n_d-1$, $\Gamma$ must contain some $\Gamma_k\in\OMCut_{G_k}(v_k,w_k)$ where $\OMCut_{G_k}$ is the set of odd minimal cutsets in $G_k$. This follows by noting that otherwise, for some $0\le k\le n_d-1$, there exists a path going from $v$ to $v_k$ along $P(v)$ then inside $G_k$ to $w_k$ and then to $w$ along $P(w)$ without intersecting $\Gamma$ at all. Thus, since $|\Gamma_k|\ge d-1$ for all $k$, by Proposition~\ref{equal_edgenum_in_all_dir}, we deduce $|\Gamma|\ge (d-1)n_d$.
\end{enumerate}
Hence in all cases, $|\Gamma|\ge\min\left(s_r, (d-1)n_d\right)$. Since this is true for any $x,y\in V[G]$ and $\Gamma\in\OMCut(B_r(x),B_r(y))$, the proposition follows.
\end{proof}
The rest of the section is devoted to proving the general case of Theorem~\ref{I_r_bound_thm}, see Corollary~\ref{general_I_r_bound} below. Our main tool for finding lower bounds for $I_r$ is the following
\begin{lemma}\label{MCut_size_lower_bound_lem}
For $X,Y\subseteq V[G]$, if there exist $k$ paths, each connecting a vertex of $X$ to a vertex of $Y$ such that each edge in $G$ is traversed by at most $m$ of these paths, then for every $\Gamma\in\MCut(X,Y)$ we have $|\Gamma|\ge \frac{k}{m}$.
\end{lemma}
The lemma follows directly from the fact that each $\Gamma\in\MCut(X,Y)$ must have an edge in common with every one of the given paths.

We note the following simple geometric lemmas.
\begin{lemma}\label{sphere_triangle_ineq_lem}
For any $x\in V[G]$ and integer $r\ge 0$ we have $d_G(v,w)\le 2r$ for $v,w\in B_r(x)$.
\end{lemma}
The lemma follows directly from the definition of $B_r$ and the triangle inequality.
\begin{lemma}\label{sphere_intersection_line_lem}
For any $v,w\in V[G]$, integer $r\ge 0$ and $1\le i\le d$ we have
\begin{equation*}
|S_r(v)\cap\{w+k e_i\ |\ k\in\Z\}|\le 2.
\end{equation*}
\end{lemma}
The lemma is straightforward from the definition of $S_r(v)$.

For $i\in\Z$ we introduce the notation
\begin{equation*}
\{i\}:=i\bmod d\quad\text{ and }\quad[i]:=((i-1)\bmod d) + 1.
\end{equation*}
That is, $i$ normalized to be in the range $0$ to $d-1$ and in the range $1$ to $d$ respectively. For $v,w\in V[G]$, we note that there is a unique way to write
\begin{equation*}
w = v + \sum_{i=1}^d k_i e_i
\end{equation*}
where $0\le k_i\le n_{i}-1$ for all $i$. We denote $(w-v)_i:=k_i$. For $u\in V[G]$, we let $u+(w-v)$ equal the vertex $u+\sum_{i=1}^d (w-v)_i e_i$. In addition, for a path $P$, we denote by $P+(w-v)$ the path obtained from $P$ by adding $w-v$ to each vertex. We denote by $E(P)$ the set of edges that $P$ traverses.
For any integers $m\ge 1$, $1\le i_1,\ldots, i_m\le d$, $k_1,\ldots, k_m\in\Z$ and $v\in V[G]$, we let $v+P_{i_1}^{k_1}P_{i_2}^{k_2}\cdots P_{i_m}^{k_m}$ be the path which starts from $v$, moves to $v+k_1 e_{i_1}$ by adding $e_{i_1}$ each step, then moves to $v+k_1 e_{i_1} + k_2 e_{i_2}$ by adding $e_{i_2}$ each step and so on until reaching $v+\sum_{j=1}^m k_j e_{i_j}$. We have
\begin{lemma}\label{path_intersection_lem}
For $v_1, v_2\in V[G]$, $1\le i,j\le d$ and $k_1,\ldots, k_d\in\Z$ satisfying $0\le k_i\le n_i-1$, let
\begin{equation*}
\begin{split}
P_1&:= v_1 + P_{i}^{k_{i}}P_{[i+1]}^{k_{[i+1]}}\cdots P_{[i+d-1]}^{k_{[i+d-1]}},\\
P_2&:= v_2 + P_{j}^{k_{j}}P_{[j+1]}^{k_{[j+1]}}\cdots P_{[j+d-1]}^{k_{[j+d-1]}}.
\end{split}
\end{equation*}
Then if $\{u,u+e_m\}\in E(P_1)\cap E(P_2)$ for some $u\in V[G]$ and $1\le m\le d$ with $\{j-i\}\le \{m-i\}$ then for $0\le \ell<\{j-i\}$ we have $(v_2-v_1)_{[i+\ell]}=k_{[i+\ell]}$ and for $\{j-i\}\le \ell<d$, $\ell\neq \{m-i\}$ we have $(v_2-v_1)_{[i+\ell]}=0$.
\end{lemma}
\begin{proof}
Assume that $\{u,u+e_m\}\in E(P_1)\cap E(P_2)$ for some $u\in V[G]$ and $1\le m\le d$ with $\{j-i\}\le \{m-i\}$. Let $x_1:=v_1+\sum_{\ell=0}^{\{m-i\}-1} k_{[i+\ell]} e_{[i+\ell]}$ and $x_2:=v_1+\sum_{\ell=0}^{\{m-i\}} k_{[i+\ell]} e_{[i+\ell]}$. Since the edge $\{u,u+e_m\}$ is in the direction of $e_m$, it must lie in $P_1$ in the segment of the path between $x_1$ and $x_2$. For the same reason, it must lie in $P_2$ in the segment of the path between $y_1:=v_2+\sum_{\ell=0}^{\{m-j\}-1} k_{[j+\ell]} e_{[j+\ell]}$ and $y_2:=v_2+\sum_{\ell=0}^{\{m-j\}} k_{[j+\ell]} e_{[j+\ell]}$. This implies that $u$ differs from each of $x_1$ and $y_1$ only in the $m$'th coordinate, so that $y_1=x_1+ke_m$ for some $k$. Hence $ke_m=y_1-x_1=v_2-v_1-\sum_{\ell=0}^{\{j-i\}-1} k_{[i+\ell]} e_{[i+\ell]}$ (using that $\{j-i\}\le \{m-i\}$) and the lemma follows.
\end{proof}

\begin{proposition}\label{diff_coord_isoperim}
For all integer $r\ge 0$ and $x,y\in V[G]$, if $x$ and $y$ differ at exactly $k$ coordinates then $I_r(x,y)\ge\frac{ks_r}{2\min(2r+1,k)\Delta(G)}$.
\end{proposition}
\begin{proof}
Fix an integer $r\ge 0$ and $x,y\in V[G]$ which differ at exactly $k$ coordinates. Our proof does not depend on the order of the coordinates and we assume that the coordinates that $x$ and $y$ differ at are the first $k$ coordinates. For $1\le i\le d$, we also let $k_i:=(y-x)_i$ (so that $k_i=0$ for $i>k$). Then define, for $1\le i\le k$, paths from $x$ to $y$ by
\begin{equation*}
P_i:=x+P_i^{k_i}P_{[i+1]}^{k_{[i+1]}}\cdots P_{[i+d-1]}^{k_{[i+d-1]}}.
\end{equation*}
We let $\PP$ be all paths of the form $P_i+(v-x)$ for some $1\le i\le k$ and $v\in E_r(x)$. For an edge $\{u,u+e_m\}$ for some $u\in V[G]$ and $1\le m\le d$, let $\PP(u,m)$ be the set of all paths $P\in\PP$ which pass through $\{u,u+e_m\}$. Since every $P\in\PP$ connects $B_r(x)$ and $B_r(y)$, if we show that $|\PP(u,m)|\le 2\min(2r+1,k)$ for all $u$ and $m$ then by Lemma~\ref{MCut_size_lower_bound_lem} and \eqref{s_r_a_r_relation}, $I_r(x,y)\ge\frac{k|E_r(x)|}{2\min(2r+1,k)}=\frac{ks_r}{2\min(2r+1,k)\Delta(G)}$, as required. Fix $u\in V[G]$ and $1\le m\le d$. We note that $\PP(u,m)=\emptyset$ if $m>k$. Assume $P_i+(v-x), P_j+(w-x)\in \PP(u,m)$ for some $1\le i,j\le k$ and $v,w\in E_r(x)$ (in particular, $1\le m\le k$). If we also assume that $\{j-i\}\le \{m-i\}$, it follows from Lemma~\ref{path_intersection_lem} that
\begin{align}
(w-v)_{[i+\ell]}&=k_{[i+\ell]}&(0\le \ell<\{j-i\}),\label{w_v_k_i}\\
(w-v)_{[i+\ell]}&=0&(\{j-i\}\le \ell<d,\ \ell\neq \{m-i\})\label{w_v_0}.
\end{align}
Let $I:=\{1\le i\le k\ |\ \exists v\in E_r(x),\ P_i+(v-x)\in\PP(u,m)\}$. Fix $i$ to be the $i\in I$ for which $\{m-i\}$ is maximal. Fix also $v\in E_r(x)$ satisfying $P_i+(v-x)\in\PP(u,m)$. Note that the extremality of $i$ implies $\{j-i\}\le \{m-i\}$ for all $j\in I$.
It follows from \eqref{w_v_k_i}, \eqref{w_v_0} and
Lemma~\ref{sphere_intersection_line_lem} that for any $j\in I$ there are at most two $w\in E_r(x)$ so that $P_j+(w-x)\in\PP(u,m)$. In addition, since $k_\ell\neq 0$ for $1\le \ell\le k$, it follows from \eqref{w_v_k_i} and Lemma~\ref{sphere_triangle_ineq_lem} that $|I|\le 2r+1$ (since for $w\in E_r(x)$, $(w-v)$ may have at most $2r$ non-zero coordinates by Lemma~\ref{sphere_triangle_ineq_lem}).
Of course, we also have the trivial $|I|\le k$. In conclusion, we see that $|\PP(u,m)|\le 2\min(2r+1,k)$, as required.
\end{proof}
\begin{proposition}\label{equal_coord_isoperim}
For all integer $r\ge 0$ and $x,y\in V[G]$, if $x$ and $y$ differ at $k$ coordinates then $I_r(x,y)\ge\frac{(\Delta(G)-k+1)s_r}{4\min(2r+1,\Delta(G)-k+1)\Delta(G)}$.
\end{proposition}
\begin{proof}
Fix an integer $r\ge 0$ and $x,y\in V[G]$ which differ at $k$ coordinates.
Denote $q:=\Delta(G)-k$. Our proof does not depend on the order of the coordinates and we assume that the equal coordinates of $x$ and $y$ are the first $q$ coordinates. For $1\le i\le d$, we also let $k_i:=(y-x)_i$ (so that $k_i=0$ for $i\le q$). Then define, for $1\le i\le q+1$ paths from $x$ to $y$ by
\begin{equation}\label{equal_coord_path_form}
P_i:=x+P_i^{1}P_{i+1}^{1}\cdots P_{q}^1 P_{q+1}^{k_{q+1}}\cdots P_d^{k_d}P_q^{-1}P_{q-1}^{-1}\cdots P_i^{-1},
\end{equation}
where if $i=q+1$, we start the path with $P_{q+1}^{k_{q+1}}$ and end it with $P_d^{k_d}$. We let $\PP$ be all paths of the form $P_i+(v-x)$ for some $1\le i\le q+1$ and $v\in E_r(x)$. For an edge $\{u,u+e_m\}$ for some $u\in V[G]$ and $1\le m\le d$, let $\PP(u,m)$ be the set of all paths $P\in\PP$ which pass through $\{u,u+e_m\}$. As in the proof of Proposition~\ref{diff_coord_isoperim}, it is sufficient to show that $|\PP(u,m)|\le 4\min(2r+1,q+1)$ for all $u$ and $m$. Fix $u\in V[G]$ and $1\le m\le d$. If $m\le q$, let $\PP^1(u,m)$ (respectively $\PP^2(u,m)$) be those $P\in\PP(u,m)$ which traverse the edge in their $P_m^1$ segment (respectively in their $P_m^{-1}$ segment). If $m>q$, let $\PP^1(u,m)=\PP^2(u,m)=\PP$. Assume $P_i+(v-x), P_j+(w-x)\in \PP^a(u,m)$ for some $1\le i\le j\le q+1$, $v,w\in E_r(x)$ and $a\in\{1,2\}$. We observe that we may not have $m<j$.
Thus, by \eqref{equal_coord_path_form} (similarly to Lemma~\ref{path_intersection_lem}), we must have
\begin{equation}\label{w_v_ell_equation}
(w-v)_\ell=\begin{cases} 1& i\le \ell< j\\ 0&\ell\notin [i,j]\cup\{m\}\end{cases}.
\end{equation}
Hence, by Lemma~\ref{sphere_triangle_ineq_lem} and since $v,w\in E_r(x)$, we must have $j-i\le 2r$. Of course, we must also have $j-i\le q$. Furthermore, we deduce from \eqref{w_v_ell_equation} and Lemma~\ref{sphere_intersection_line_lem} that there are at most two $w'$ such that $P_j+(w'-x)\in\PP^a(u,m)$. We conclude that $|\PP^a(u,m)|\le 2\min(2r+1,q+1)$ for each $a\in\{1,2\}$ and hence $|\PP(u,m)|\le 4\min(2r+1,q+1)$, as required.
\end{proof}
\begin{corollary}\label{general_I_r_bound}
For all $r\ge 0$, $I_r\ge \frac{s_r}{2\min(4(2r+1),\Delta(G))}$.
\end{corollary}
\begin{proof}
Fix an integer $r\ge 0$ and $x,y\in V[G]$. Let $k$ be the number of coordinates at which $x$ and $y$ differ. Proposition~\ref{diff_coord_isoperim} gives
\begin{equation*}
I_r(x,y)\ge \frac{ks_r}{2\min(2r+1,k)\Delta(G)}\ge \frac{s_r}{2\Delta(G)}.
\end{equation*}
Furthermore, Propositions~\ref{diff_coord_isoperim} and \ref{equal_coord_isoperim} give
\begin{equation*}
I_r(x,y)\ge \max\left(\frac{ks_r}{2(2r+1)\Delta(G)}, \frac{(\Delta(G)-k+1)s_r}{4(2r+1)\Delta(G)}\right)\ge \frac{s_r}{8(2r+1)}.
\end{equation*}
Since both the above bounds hold uniformly in $x$ and $y$, the corollary follows.
\end{proof}

\subsubsection{Isoperimetric relations}\label{isoperimetric_relations_sec}
In this section we note several simple relations for $s_r$ and $\vol(r)$.
\begin{proposition}\label{sum_s_r_prop}
For any $t\ge 0$ and torus $G$ with side lengths satisfying \eqref{torus_side_lengths} we have $\frac{\Delta(G)}{2}\vol(t)\le \sum_{r=0}^{t} s_r\le \Delta(G)\vol(t)$.
\end{proposition}
\begin{proof}
Fix $v\in G$. The upper bound follows directly from \eqref{s_r_a_r_relation}. To see the lower bound, note that by \eqref{s_r_a_r_relation}, it is sufficient to show that $\sum_{r=0}^t |E_r(v)|\ge \frac{1}{2}|B_t(v)|$. For $0\le r\le t$, let $E_r'(v):=\{w\in S_r(v)\ |\ w+e_d\in S_{r-1}\}$. Noting that $E_r(v)=\{w\in S_r(v)\ |\ w+e_d\in S_{r+1}\}$ and using the fact that $n_d$ is even, we have $S_r(v)=E_r(v)\cup E_r'(v)$ for all $r$. By our definitions and symmetry, $|E_r'(v)|= |E_{r-1}(v)|$ for all $1\le r\le t$. Thus,
\begin{equation*}
|B_t(v)|=\sum_{r=0}^t |S_r|\le \sum_{r=0}^t |E_r(v)| + |E_r'(v)| = \sum_{r=0}^t |E_r(v)|+\sum_{r=0}^{t-1} |E_r(v)|\le 2\sum_{r=0}^t |E_r(v)|
\end{equation*}
as required.
\end{proof}
\begin{proposition}\label{ball_volume_estimate}
There exists $c>0$ such that for any $d\ge 4$, torus $G$ with side lengths satisfying \eqref{torus_side_lengths} and integer $0\le r\le \diam(G)$ (where $\diam(G)$ is the diameter of $G$), we have $\vol(r)\ge crd^2$.
\end{proposition}
\begin{proof}
Fix $d\ge 4$ and a torus $G$ with side lengths satisfying \eqref{torus_side_lengths}. Fix also $0\le r\le \diam(G)$ and $v\in V[G]$. The claim holds for $0\le r\le 2$ since $\vol(0)=1$, $\vol(1)=\Delta(G)+1\ge d$ and $\vol(2)\ge \binom{d}{2}$. Thus we assume that $r\ge 3$. Let $E$ be the set of all vertices of the form $v+e_i+e_j+e_k+\ell e_d$ for $1\le i<j<k\le d-1$ and $0\le \ell\le\min(r-3, n_d)$. Note that $E\subseteq B_r(v)$. Since $r\le \diam(G)\le dn_d$, we deduce
\begin{equation*}
|E|=\binom{d-1}{3}(\min(r-3,n_d)+1)\ge cd^3(\min(r+1,n_d))\ge crd^2
\end{equation*}
for some $c>0$, as required.
\end{proof}
\begin{proposition}\label{linear_volume_growth}
For any torus $G$ with side lengths satisfying \eqref{torus_side_lengths} and any integer $0\le r\le \frac{n_d-3}{4}$, we have $\vol(2r+1)\ge 2\vol(r)$.
\end{proposition}
\begin{proof}
Fix $v_0\in V[G]$ and let $B^1:=B_r(v_0)$ and $B^2:=B_{2r+1}(v_0)$. $B^1\subseteq B^2$ by definition. Hence it is sufficient to define a one-to-one $T:B^1\to B^2$ satisfying $T(B^1)\cap B^1 = \emptyset$. Let $w\in B^1$ and write $w=v+ke_d$ for some integer $-r\le k\le r$, where $v=v_0+\sum_{i=0}^{d-1} k_i e_i$ for some integers $k_i$. If $k\ge 0$, define $T(w):=v+(k+r+1)e_d$ and if $k<0$, define $T(w):=v+(k-r)e_d$. It is straightforward to check that $T$ has the required properties.
\end{proof}
\begin{proposition}\label{tori_size_prop}
For any $\lambda>0$ there exists $d_0(\lambda)$ such that for all $d\ge d_0(\lambda)$ and tori $G$ with side lengths satisfying \eqref{torus_side_lengths}, if $k:=\min\{m\in\N\ |\ \vol(m)\ge\lambda\log^2 d\log|V[G]|\}$ then:
\begin{enumerate}
\item If $n_d\le d^3$ then $k=2$.
\item If $n_d\ge d^3$ then $s_\ell\le (d-1) n_d$ for all integer $0\le \ell\le k$.
\end{enumerate}
\end{proposition}
\begin{proof}
Fix $\lambda>0$ and let $G$ be a torus with side lengths satisfying \eqref{torus_side_lengths}. Let $k$ be as in the lemma.

For part 1, we note first that $\vol(1)=\Delta(G)+1\le 3d$ and $\vol(2)\ge cd^2$ for some $c>0$ (independent of $d$ and $G$). Second, we note that $|V[G]|\ge 2^d$ and $|V[G]|\le n_d^{d}\le d^{3d}$. Thus, $\vol(1)<\lambda\log^2 d\log|V[G]|$ and $\vol(2)\ge \lambda\log^2 d\log|V[G]|$ if $d_0(\lambda)$ is sufficiently large, as required.

For part 2, fix $0\le \ell\le k$. By Proposition~\ref{sum_s_r_prop}, $s_\ell\le \Delta(G)\vol(k)$. Also by our definitions, $\vol(m)\le (\Delta(G)+1)\vol(m-1)$ for all $m\in\N$. Thus,
\begin{equation*}
s_\ell\le \Delta(G) \vol(k)\le \Delta(G)(\Delta(G)+1) \vol(k-1)< 6d^2\lambda\log^2 d\log|V[G]|.
\end{equation*}
Since we also have $\log|V[G]|\le d\log n_d$, it follows that $s_\ell\le (d-1)n_d$ whenever
\begin{equation*}
\frac{n_d}{\log n_d}\ge \frac{6d^3\lambda\log^2 d}{d-1},
\end{equation*}
which is satisfied if $n_d\ge d^3$ and $d_0(\lambda)$ is sufficiently large.
\end{proof}

\subsection{Height}
In this section we prove Theorem~\ref{height_at_point_thm}, Corollary~\ref{even_lattice_zero_cor} and Theorems~\ref{coloring_intro_thm} and \ref{coloring_thm}.

We start by defining the level set of a function at height $i$. For a torus $G$ (with side lengths satisfying \eqref{torus_side_lengths}), legal boundary conditions $(B,\mu)$, $g\in\Hom(G,B,\mu)$ and $i\in\N$, assuming $\mu(b)\le i-1$ for all $b\in B$, we define
\begin{equation*}
A_i:=\text{union of the connected components of points of $B$ in $G\setminus\{v\in V[G]\ |\ g(v)=i\}$}
\end{equation*}
and $\LS_i(g,x,B)$ to be the empty set if $x\in A_i$ or otherwise be all edges between $A_i$ and the connected component of $x$ in $V[G]\setminus A_i$. In words, $\LS_i(g,x,B)$ is the outermost height $i$ level set of $g$ around $x$ when coming from $B$. Note that if it is not empty then it belongs to $\OMCut(x,B)\cup\OMCut(B,x)$. Note also that $\LS_1(g,x,B)=\LS(g,x,B)$.
As a first step in the proof of Theorem~\ref{height_at_point_thm} we establish the following proposition.
\begin{proposition}\label{height_and_level_set_prop}
There exist $d_0\in\N$, $c>0$ such that for all $d\ge d_0$, non-linear tori $G$, legal boundary conditions $(B,\mu)$ with non-positive $\mu$, $x\in V[G]$ and $t\in\N$, if we let $f\unifin\Hom(G,B,\mu)$ and define, for each integer $1\le i\le t$,
\begin{equation*}
L_{i,t}:=\min\left(|\LS_i(g,x,B)|\ \big|\ g\in\Hom(G,B,\mu),\ g(x)\ge t\right),
\end{equation*}
where $L_{i,t}$ is defined to be infinity if the set minimized over is empty, then
\begin{equation}
\P(f(x)\ge t)\le d^{3t}\exp\left(-\frac{c\sum_{i=1}^t L_{i,t}}{d\log^2 d}\right).
\end{equation}
\end{proposition}

For the proof, we fix a non-linear torus $G$, legal boundary conditions $(B,\mu)$ with non-positive $\mu$ and $x\in V[G]$, and set $f\unifin\Hom(G,B,\mu)$. We will need the following definitions and lemma. Define $B_i:=\Ei(\LS_i(f,x,B),x)$ if $\LS_i(f,x,B)\neq\emptyset$ and otherwise $B_i:=\emptyset$, and $\mu_i:B_i\to\Z$ by $\mu_i(b):=i$ for all $b\in B_i$. For a set $\CC\subseteq V[G]$, we shall write $f|_{\CC}$ for the function $f$ restricted to $\CC$.
\begin{lemma}\label{level_set_cond_lemma}
Conditionally on $\LS_1(f,x,B)$ we have on the event $\LS_1(f,x,B)\neq\emptyset$ that
\begin{equation*}
f|_{\CC}\eqd f'|_{\CC}
\end{equation*}
for $\CC:=\comp(\LS_1(f,x,B),x)$ and $f'\unifin\Hom(G,B_i,\mu_i)$.
\end{lemma}
The lemma is standard and follows from the facts that the event $\LS_i(f,x,B)=\Gamma$, for some $\Gamma\in\MCut(x,B)$, is determined solely by the values of $f$ outside of $\comp(\Gamma,x)$ (since $\mu$ is non-positive), that the constraints on $f$ are of nearest-neighbor type and that the measure on $f$ is uniform. We omit the detailed proof.
\begin{proof}[Proof of Proposition~\ref{height_and_level_set_prop}]
It is sufficient to show that under the assumptions of the proposition, for any integers $(L_i)_{i=1}^t\subseteq \N$, we have
\begin{equation}\label{given_level_set_lengths_inequality}
\P\left(|\LS_i(f,x,B)|=L_i\text{ for all $1\le i\le t$}\right)\le d^t\exp\left(-\frac{c\sum_{i=1}^t L_i}{d\log^2 d}\right)
\end{equation}
for some $c>0$. The proposition follows from this inequality by summing over all $L_i\ge L_{i,t}$ for $1\le i\le t$ (using that if $g\in\Hom(G,B,\mu)$ satisfies $g(x)\ge t$ then necessarily $\LS_i(g,x,B)\neq\emptyset$ for $1\le i\le t$).

We prove \eqref{given_level_set_lengths_inequality} by induction on $t$. For $t=1$, the inequality follows from Theorem~\ref{main_thm} (taking $d$ large enough). Assume \eqref{given_level_set_lengths_inequality} holds for any legal boundary conditions $(B,\mu)$ with non-positive $\mu$, for $t=1$ and for a given $t\ge 1$, and let us prove it for $t+1$. Fix a non-linear torus $G$, legal boundary conditions $(B,\mu)$ with non-positive $\mu$ and integer $(L_i)_{i=1}^{t+1}\subseteq \N$, and let $f\unifin\Hom(G,B,\mu)$. Conditioning on $\LS_1(f,x,B)$ and on the event $\LS_1(f,x,B)\neq\emptyset$, we let $f'\unifin\Hom(G,B_1,\mu_1)$ and note that since, for $i\ge 2$, $\LS_i(f,x,B)$ depends only on $f|_{\comp(\LS_1(f,x,B),x)}$, we have by Lemma~\ref{level_set_cond_lemma} that $\LS_i(f,x,B)\eqd\LS_i(f',x,B_1)$ for all $i\ge 2$. Thus, by the induction hypothesis for $t=1$ we have
\begin{equation}\label{apriori_L_i_bound}
\begin{split}
\P&\left(|\LS_i(f,x,B)|=L_i\text{ for all $1\le i\le t+1$})=\P(|\LS_1(f,x,B)|=L_1\right)\cdot\\
&\cdot\P\left(|\LS_i(f,x,B)|=L_i\text{ for all $2\le i\le t+1$}\ \big|\ |\LS_1(f,x,B)|=L_1\right)\le\\
&\le d\exp\left(-\frac{cL_1}{d\log^2 d}\right)\P\left(|\LS_i(f',x,B_1)|=L_i\text{ for all $2\le i\le t+1$}\ \big|\ |\LS_1(f,x,B)|=L_1\right).
\end{split}
\end{equation}

We now note that if we set $f'':=f'-1$ then $f''\unifin\Hom(G,B_1,\mu-1)$, $\LS_i(f',x,B_1) = \LS_{i-1}(f'',x,B_1)$ for $i\ge 2$ and $(B_1,\mu_1-1)$ are legal boundary conditions (if we switch the roles of $V^{\even}$ and $V^{\odd}$ or alternatively shift $B_1$ by one coordinate on the torus) having $\mu_1-1$ non-positive. Thus, since by our induction hypothesis the bound \eqref{given_level_set_lengths_inequality} holds uniformly in the boundary conditions, we obtain
\begin{equation}\label{inductive_L_i_bound}
\begin{split}
\P&\left(|\LS_i(f',x,B_1)|=L_i\text{ for all $2\le i\le t+1$}\ \big|\ |\LS_1(f,x,B)|=L_1\right)=\\
&=\P\left(|\LS_i(f'',x,B_1)|=L_{i+1}\text{ for all $1\le i\le t$}\ \big|\ |\LS_1(f,x,B)|=L_1\right)\le\\
&\le d^{t}\exp\left(-\frac{c\sum_{i=1}^{t}L_{i+1}}{d\log^2 d}\right).
\end{split}
\end{equation}
Inequality \eqref{given_level_set_lengths_inequality} now follows for $t+1$ by \eqref{apriori_L_i_bound} and \eqref{inductive_L_i_bound}, completing the proof of the induction and the proposition.
\end{proof}

We are now ready to prove the theorem.
\begin{proof}[Proof of Theorem~\ref{height_at_point_thm}]
We assume $d$ is sufficiently large for the following arguments and fix a non-linear torus $G$, legal boundary conditions $(B,\mu)$ with non-positive $\mu$, $x\in V[G]$ and $t\in\N$. Let $f\unifin\Hom(G,B,\mu)$. By Proposition~\ref{height_and_level_set_prop} we have
\begin{equation}\label{height_by_level_sets_estimate}
\P(f(x)\ge t)\le d^{3t}\exp\left(-\frac{c\sum_{i=1}^t L_{i,t}}{d\log^2 d}\right).
\end{equation}

We next aim to estimate $L_{i,t}$ from below. For an integer $r\ge 0$ we define $B_r(B):=\cup_{v\in B} B_r(v)$ and
observe that for all integers $1\le i\le t$ and $g\in\Hom(G,B,\mu)$ with $g(x)\ge t$ we have
\begin{equation}\label{LS_i_OMCut}
\LS_i(g,x,B)\in\OMCut(B_{t-i}(x), B_{i-1}(B))\cup\OMCut(B_{i-1}(B), B_{t-i}(x))
\end{equation}
since $g$ changes by one between adjacent vertices. Thus, recalling the definitions of Section~\ref{general_isoperimetry_section}, we have (since $B\neq \emptyset$)
\begin{align}
&L_{i,t}\ge I_{\min(i-1,t-i)}\label{L_i_t_I_r_relation}\qquad\text{ and}\\
&L_{i,t}\ge I_{t-i}(B)\label{L_i_I_r_E_relation}.
\end{align}
We now proceed to examine several cases separately.
\begin{enumerate}
\item Assume $t\ge 3$. By \eqref{L_i_t_I_r_relation} and Theorem~\ref{I_r_bound_thm} we have $L_{i,t}\ge \frac{c_1 s_{\min(i-1,t-i)}}{\min(t,d)}$ for some $c_1>0$ and all $1\le i\le t$. Setting $r_0:=\lceil t/2\rceil-1$ and plugging the last bound into \eqref{height_by_level_sets_estimate} we obtain
\begin{equation*}
\P(f(x)\ge t)\le d^{3t}\exp\left(-\frac{c_2\sum_{r=0}^{r_0} s_r}{\min(t,d)d\log^2 d}\right)
\end{equation*}
for some $c_2>0$. Now applying Proposition~\ref{sum_s_r_prop} we deduce
\begin{equation*}
\P(f(x)\ge t)\le d^{3t}\exp\left(-\frac{c_3\vol(r_0)}{\min(t,d)\log^2 d}\right)
\end{equation*}
for some $c_3>0$. Finally noting that if $t>\diam(G)$ then $\P(f(x)>t)=0$ since $\mu$ is non-positive, whereas if $t\le \diam(G)$ then Proposition~\ref{ball_volume_estimate} implies $\vol(r_0)\ge c_4 td^2$ for some $c_4>0$ (using that $t\ge 3$ and hence $r_0\ge 1$), from which it follows (checking separately the cases $t\le d$ and $t>d$) that for some $c_5>0$,
\begin{equation*}
\P(f(x)\ge t)\le \exp\left(-\frac{c_5\vol(r_0)}{\min(t,d)\log^2 d}\right).
\end{equation*}
\item Assume $t\ge 3$ and $\vol(\lceil t/2\rceil-1)\le\frac{1}{3}n_d$. Setting $r_0:=\lceil t/2\rceil-1$, we observe that the volume condition and Proposition~\ref{sum_s_r_prop} imply that $s_\ell\le (d-1)n_d$ for all $0\le \ell\le r_0$. Thus, by \eqref{L_i_t_I_r_relation} and Theorem~\ref{I_r_bound_thm} we have $L_{i,t}\ge s_{\min(i-1,t-i)}$ for all $1\le i\le t$. Continuing in the same way as in the first case above, we deduce from this that for some $c>0$,
\begin{equation*}
\P(f(x)\ge t)\le \exp\left(-\frac{c\vol(r_0)}{\log^2 d}\right).
\end{equation*}
\item Assume $t\ge 2$ and $B$ has full projection. By \eqref{L_i_I_r_E_relation} and Theorem~\ref{I_r_E_bound_thm} we have $L_{i,t}\ge s_{t-1}$ for all $1\le i\le t$. Continuing in the same way as in the first case above, we deduce from this that for some $c>0$,
\begin{equation*}
\P(f(x)\ge t)\le \exp\left(-\frac{c\vol(t-1)}{\log^2 d}\right).\qedhere
\end{equation*}
\end{enumerate}

\end{proof}

\begin{proof}[Proof of Corollary~\ref{even_lattice_zero_cor}]
Under the assumptions of Theorem~\ref{height_at_point_thm} we have for any $v\in V^\even$, by the third part of Theorem~\ref{height_at_point_thm}, that $\P(f(v)\ge 2)\le \exp\left(-\frac{c\vol(1)}{\log^2 d}\right)\le\exp\left(-\frac{cd}{\log^2 d}\right)$. Since $\mu$ is zero, we also obtain $\P(f(v)\le -2)\le \exp\left(-\frac{cd}{\log^2 d}\right)$ by symmetry of the distribution of $f(v)$ around $0$. The corollary follows.
\end{proof}

\begin{proof}[Proof of Theorems~\ref{coloring_intro_thm} and \ref{coloring_thm}]
As explained before Theorem~\ref{coloring_thm}, for the zero BC $(B,\mu)$, the set $\Hom(G,B,\mu)$ is in bijection with $\col(G,B,\mu)$ under the map $f\mapsto f\bmod 3$. Thus Theorem~\ref{coloring_thm} is an immediate corollary of Corollary~\ref{even_lattice_zero_cor}. Theorem~\ref{coloring_intro_thm} is the special case of Theorem~\ref{coloring_thm} when $G=\Z_n^d$.
\end{proof}

\subsection{Range}
In this section we prove Theorems~\ref{range_thm} and \ref{finite_range_thm}, and Corollary~\ref{range_sharpness_cor}.
We deduce Theorems~\ref{homogeneous_torus_thm} and \ref{homogeneous_tori_finite_range_thm} for the homomorphism case. We start with a proposition which relates the range to isoperimetric quantities.

\begin{proposition}\label{range_bound_prop}
There exists $c,d_0>0$ such that for all $d\ge d_0$, non-linear tori $G$ and legal boundary conditions $(B,\mu)$ with zero $\mu$, if $f\unifin\Hom(G,B,\mu)$ and $k\in\N$, we have
\begin{equation*}
\P(\range(f)>2k+1)\le 9d^{6k+3}|V[G]|^4 \exp\left(-\frac{c\sum_{i=0}^{k} I_i}{d\log^2 d}\right).
\end{equation*}
\end{proposition}

\begin{proof}
We assume $d$ is sufficiently large for the following arguments and fix a non-linear torus $G$ and legal boundary conditions $(B,\mu)$ with zero $\mu$. Let $f\unifin\Hom(G,B,\mu)$ and $k\in\N$. We denote by $\A$ the set of 4-tuples $(x,y,t,s)$ with $x,y\in V[G]$, $t,s\in\Z$, and $t-s=2k+1$ for which there exists $g\in\Hom(G,B,\mu)$ satisfying $g(x)=t$ and $g(y)=s$. We observe that $|\A|\le (2|V[G]|+1)^2|V[G]|^2\le 9|V[G]|^4$ since $\mu$ is zero and $\diam(G)\le |V[G]|$. Defining the events $\Omega:=\{\range(f)>2k+1\}$ and, for $\gamma=(x,y,t,s)\in\A$, $\Omega_\gamma:=\{f(x)=t\text{ and }f(y)=s\}$, we note that
\begin{equation*}
\Omega\subseteq \cup_{\gamma\in\A} \Omega_\gamma.
\end{equation*}
Hence, by a union bound, it is sufficient to show that for each fixed $\gamma\in\A$ we have
\begin{equation}\label{Omega_gamma_inequality}
\P(\Omega_\gamma)\le d^{6k+3}\exp\left(-\frac{c\sum_{i=0}^{k} I_{i}}{d\log^2 d}\right)
\end{equation}
for some $c>0$.

We proceed to prove \eqref{Omega_gamma_inequality}. Fix $\gamma=(x,y,t,s)\in\A$. We note that since $\mu$ is zero, we have that $(y,x,-s,-t)\in\A$ and $\P(\Omega_\gamma)=\P(\Omega_{(y,x,-s,-t)})$ by symmetry of the model under replacing $f$ by $-f$. Hence we can, and do, assume WLOG that $t\ge k+1$ (using that $t-s=2k+1$). We observe that
\begin{equation}\label{Omega_gamma_first_estimate}
\P(\Omega_\gamma)=\P(f(y)=s)\P(f(x)=t\ |\ f(y)=s)\le \P(f(x)\ge t\ |\ f(y)=s).
\end{equation}
We let $B':=B\cup\{y\}$ and $\mu':B'\to\Z$ be defined by $\mu'(v)=\mu(v)$ for $v\in B$ and $\mu'(y)=s$. We then let $f'\unifin\Hom(G,B',\mu')$ and note that conditioned on $f(y)=s$, $f\eqd f'$. Hence, by \eqref{Omega_gamma_first_estimate}, we have
\begin{equation}\label{Omega_gamma_f'_inequality}
\P(\Omega_\gamma)\le \P(f'(x)\ge t).
\end{equation}
Define $r:=\max(s,0)$. We define $\mu'':B'\to\Z$ by $\mu''(v):=\mu'(v)-r$. We observe that $(B',\mu'')$ is a legal boundary condition with non-positive $\mu''$ (if needs be, we exchange $V^\even$ and $V^\odd$ to ensure this). Furthermore, letting $f''\unifin\Hom(G,B',\mu'')$, we note that $f''\eqd f'-r$. Thus,
\begin{equation}\label{f'_f''_equality}
\P(f'(x)\ge t)=\P(f''(x)\ge t-r).
\end{equation}
Denoting $m:=t-r$, we note that $k+1\le m\le 2k+1$ since $t\ge k+1$, $t-s=2k+1$ and by the definition of $r$. Furthermore, $m-\mu''(y)=t-r-(s-r)=t-s=2k+1$. By Proposition~\ref{height_and_level_set_prop} applied to $f''$ (using that $\mu''$ is non-positive) we have
\begin{equation}\label{f_L_i_bound}
\P(f''(x)\ge m)\le d^{3m}\exp\left(-\frac{c\sum_{i=1}^{m} L_{i,m}}{d\log^2 d}\right)
\end{equation}
for some $c>0$, where for $1\le i\le m$,
\begin{equation*}
L_{i,m}:=\min(|\LS_i(g,x,B')|\ \big|\ g\in\Hom(G,B',\mu''),\ g(x)\ge m),
\end{equation*}
with $L_{i,m}$ defined to be infinity if the set minimized over is empty. Fix $m-k\le i\le m$ and note that $i\ge 1$. Fix $g\in\Hom(G,B',\mu'')$ satisfying $g(x)\ge m$. Since $g$ changes by one between adjacent vertices we deduce that
\begin{equation*}
\LS_i(g,x,B')\in\OMCut(B_{m-i}(x),B_{i-g(y)-1}(y))\cup\OMCut(B_{i-g(y)-1}(y),B_{m-i}(x)).
\end{equation*}
Moreover, since by our assumption $m-i\le k$ and $i-g(y)-1=i-\mu''(y)-1\ge k$, we conclude that
\begin{equation*}
\LS_i(g,x,B')\in\OMCut(B_{m-i}(x),B_{m-i}(y))\cup\OMCut(B_{m-i}(y),B_{m-i}(x)).
\end{equation*}
Thus, by definition, $|\LS_i(g,x,B')|\ge I_{m-i}$. Plugging this into \eqref{f_L_i_bound} and using that $k+1\le m\le 2k+1$ we obtain
\begin{equation*}
\P(f''(x)\ge m)\le d^{3m}\exp\left(-\frac{c\sum_{i=1}^{m} I_{m-i}}{d\log^2 d}\right)\le d^{6k+3}\exp\left(-\frac{c\sum_{i=0}^{k} I_{i}}{d\log^2 d}\right).
\end{equation*}
Substituting this last inequality into \eqref{f'_f''_equality} and \eqref{Omega_gamma_f'_inequality} proves \eqref{Omega_gamma_inequality}, from which the proposition follows.
\end{proof}

\begin{proof}[Proof of Theorem~\ref{range_thm}]
Fix $\lambda>0$ to be chosen below. We assume that $d$ is sufficiently large for the following arguments and, in particular, $d\ge d_0(\lambda)$ for the $d_0(\lambda)$ of Proposition~\ref{tori_size_prop}. Fix a non-linear torus $G$ and legal boundary conditions $(B,\mu)$ with zero $\mu$. Set $k:=\min\{m\in\N\ |\ \vol(m)\ge\lambda\log^2 d\log|V[G]|\}$ and let $f\unifin\Hom(G,B,\mu)$. By Proposition~\ref{range_bound_prop}, we have
\begin{equation}\label{k_range_bound}
\P(\range(f)>2k+1)\le 9d^{6k+3}|V[G]|^4 \exp\left(-\frac{c_0\sum_{i=0}^{k} I_i}{d\log^2 d}\right)
\end{equation}
for some $c_0>0$. Next, we note that by Proposition~\ref{tori_size_prop} (using that $d\ge d_0(\lambda)$), either $k=2$ or $s_{i}\le (d-1)n_d$ for all $0\le i\le k$. In both cases, we obtain by Theorem~\ref{I_r_bound_thm} that $I_{i}\ge c_1 s_{m-i}$ for some $c_1>0$ and all $0\le i\le k$. Plugging this inequality into \eqref{k_range_bound} and using Proposition~\ref{sum_s_r_prop}, we obtain
\begin{equation}\label{range_estimate_with_d}
\P(\range(f)>2k+1)\le 9d^{6k+3}|V[G]|^4 \exp\left(-\frac{c_2\vol(k)}{\log^2 d}\right)
\end{equation}
for some $c_2>0$. Noting now that if $k>\diam(G)$ then $\P(\range(f)>2k+1)=0$ and that if $k\le \diam(G)$, then by Proposition~\ref{ball_volume_estimate} we have $\vol(k)\ge c_3kd^2$ for some $c_3>0$, we deduce (using that $k\ge 1$ by assumption) that
\begin{equation}\label{range_estimate_without_d}
\P(\range(f)>2k+1)\le |V[G]|^4\exp\left(-\frac{c_4\vol(k)}{\log^2 d}\right)
\end{equation}
for some $c_4>0$. Finally, taking $\lambda:=\frac{8}{c_4}$, we obtain by the definition of $k$ that
\begin{equation*}
\P(\range(f)>2k+1)\le \exp\left(-\frac{\frac{1}{2}c_4\vol(k)}{\log^2 d}\right)\le |V[G]|^{-4},
\end{equation*}
as required.
\end{proof}

\begin{proof}[Proof of Corollary~\ref{range_sharpness_cor}]
Assume $d$ is sufficiently large for the following arguments and fix
a non-linear torus $G$ and a one-point BC $(B,\mu)$. Let
$f\unifin\Hom(G,B,\mu)$ and $r:=\min\left\{ m\in\N\ \big|\ \vol(m)\ge
\log|V[G]|\right\}$. Let also $k_1:=\min\left\{ m\in\N\cup\{0\}\
\big|\ \vol(m)\le \frac{1}{2}\log|V[G]|\right\}$ and $k_2:=\min\left\{
m\in\N\ \big|\ \vol(m)\ge \log^3 d\cdot\log|V[G]|\right\}$. By
Theorems~\ref{BYY_thm} and \ref{range_thm}, we have
\begin{equation*}
\P(k_1\le \range(f)\le k_2)\ge 1-\frac{1}{|V[G]|^3}.
\end{equation*}
Thus it remains only to note that since $\vol(dn_d)=|V[G]|\ge \log^3
d\cdot \log|V[G]|$, Proposition~\ref{linear_volume_growth} implies
that $C_d r\ge k_2$ for some $C_d>0$ and either $c_d r\le k_1$ for
some $c_d>0$, or $k_1=0$ and $c_d r\le 1$ for some $c_d>0$. Since
$\range(f)\ge 1$ with probability 1, the corollary follows.
\end{proof}

\begin{proof}[Proof of Theorem~\ref{homogeneous_torus_thm} for the homomorphism case]
The theorem follows by specializing Theorems~\ref{height_at_point_thm}, \ref{range_thm} and Corollary~\ref{range_sharpness_cor} to the case $G=\Z_n^d\times\Z_2^{m}$ (with $m$ possibly equal to $0$ and $d+m$ large enough so that these theorems apply) and observing that for these graphs there exist $C_{d,m},c_{d,m}>0$ such that $|V[G]|=2^{m} n^d$, $\diam(G)\ge \frac{1}{2}n$ and $c_{d,m} s^d\le |\vol(s)|\le C_{d,m} s^d$ for integer $1\le s\le\diam(G)$ (we are also using the fact that under the assumptions of the theorem, $\P(f(x)>\diam(G))=0$ for all $x$ since $f$ changes by one between adjacent sites).
\end{proof}

\begin{proof}[Proof of Theorem~\ref{finite_range_thm}]
Fix an integer $k\ge 2$. We assume $d$ is sufficiently large as a function of $k$ for the following arguments and fix a non-linear torus $G$ and legal boundary conditions $(B,\mu)$ with zero $\mu$. Let $f\unifin\Hom(G,B,\mu)$. By Proposition~\ref{range_bound_prop} we have
\begin{equation}\label{first_range_inequality}
\P(\range(f)>2k+1)\le 9d^{6k+3}|V[G]|^4 \exp\left(-\frac{c\sum_{i=0}^{k} I_i}{d\log^2 d}\right)
\end{equation}
for some $c>0$. By Theorem~\ref{I_r_bound_thm}, we have $I_i\ge c_k s_i$ for some $c_k>0$ and all $0\le i\le k$. Plugging this into \eqref{first_range_inequality} and using Proposition~\ref{sum_s_r_prop}, we obtain
\begin{equation}\label{second_range_inequality}
\P(\range(f)>2k+1)\le 9d^{6k+3}|V[G]|^4 \exp\left(-\frac{c'_k \vol(k)}{\log^2 d}\right)
\end{equation}
for some $c'_k>0$. As in the passage from \eqref{range_estimate_with_d} to \eqref{range_estimate_without_d}, this implies
\begin{equation*}
\P(\range(f)>2k+1)\le |V[G]|^4 \exp\left(-\frac{c''_k \vol(k)}{\log^2 d}\right)
\end{equation*}
for some $c''_k>0$. Now if $d\ge k$, then $\vol(k)\ge \binom{d}{k}\ge c'''_k d^k$ for some $c'''_k>0$, which, when plugged into the previous inequality, gives
\begin{equation*}
\P(\range(f)>2k+1)\le |V[G]|^4 \exp\left(-\frac{\tilde{c}_k d^k}{\log^2 d}\right)
\end{equation*}
for some $\tilde{c}>0$. Thus the result follows from the assumption that $|V[G]|^4 \le \exp\left(\frac{\tilde{c}_k d^k}{2\log^2 d}\right)$.
\end{proof}

\begin{proof}[Proof of Theorem~\ref{homogeneous_tori_finite_range_thm} for the homomorphism case]
The theorem follows by specializing Theorem~\ref{finite_range_thm} to the case $G=\Z_n^d$.
\end{proof}

\subsection{Lipschitz}
In this section we prove our theorems for Lipschitz height functions: Theorems~\ref{Yadin_bijection_thm}, \ref{Lipschitz_height_at_point_thm}, \ref{Lipschitz_range_thm}, \ref{Lipschitz_finite_range_thm}, \ref{Lipschitz_linear_torus_thm}, the Lipschitz case of Theorems~\ref{homogeneous_torus_thm} and \ref{homogeneous_tori_finite_range_thm}, and Corollaries~\ref{Lip_zero_one_cor}, \ref{Lip_one_point_cor}, \ref{two_values_Lip_cor} and \ref{two_sided_Lipschitz_range_cor}.

\begin{proof}[Proof of Theorem~\ref{Yadin_bijection_thm}]
As in the theorem, fix graphs $G$, $G_2$ and boundary conditions $(B,\mu)$, $(B_2,\mu_2)$ and define $T:\Hom(G_2,B_2,\mu_2)\to\Lip(G,B,\mu)$ by
\begin{equation*}
T(f)(v):=\max(f((v,0)), f((v,1)))
\end{equation*}
and also $S:\Lip(G,B,\mu)\to\Hom(G_2,B_2,\mu_2)$ by
\begin{equation}\label{inverse_Yadin_bijection_from_proof}
S(g)((v,i)):=\begin{cases} g(v)& i=g(v)\bmod 2\\ g(v)-1&i\neq g(v)\bmod 2\end{cases}.
\end{equation}
It is straightforward to verify that if $(B,\mu)$ is a Lipschitz legal boundary conditions then for each $g\in\Lip(G,B,\mu)$, $S(g)\in\Hom(G_2,B_2,\mu_2)$ and $T(S(g))=g$. Furthermore, if $(B_2,\mu_2)$ is a homomorphism legal boundary conditions then for each $f\in\Hom(G_2,B_2,\mu_2)$, $T(f)\in\Lip(G,B,\mu)$ and $S(T(f))=f$. Thus, $(B,\mu)$ is a Lipschitz legal boundary condition if and only if $(B_2,\mu_2)$ is a homomorphism legal boundary condition and in this case, $T$ is a bijection and $S=T^{-1}$.

Finally, we observe that if $(B_2,\mu_2)$ is a homomorphism legal boundary condition and $f\in\Hom(G_2,B_2,\mu_2)$ then by definition
\begin{equation}\label{verifying_range_relation}
\begin{split}
&\max\{f((v,i))\ |\ (v,i)\in G_2\} = \max\{T(f)(v)\ |\ v\in G\}\qquad\text{and}\\
&\min\{f((v,i))\ |\ (v,i)\in G_2\} = \min\{T(f)(v)\ |\ v\in G\}-1.
\end{split}
\end{equation}
We conclude that $\range(T(f))=\range(f)-1$, as required.
\end{proof}

\begin{proof}[Proof of Corollary~\ref{Lip_zero_one_cor}]
Let $(B,\Psi)$ be Lipschitz legal BC with zero-one $\Psi$ and set $B_2':=\{(v,0)\ |\ v\in B\}$ and $\mu_2':B_2'\to\Z$ to be identically zero. Let $T$ be the Yadin bijection and $S$ be the transformation defined in \eqref{inverse_Yadin_bijection_from_proof} above. It is straightforward to verify that for any $f\in\Hom(G_2,B_2',\mu_2')$ we have $T(f)\in\Lip(G,B,\Psi)$ and $S(T(f))=f$, and that for any $g\in\Lip(G,B,\Psi)$ we have $S(g)\in\Hom(G_2,B_2',\mu_2')$ and $T(S(g))=g$. Furthermore, as in \eqref{verifying_range_relation}, for any $f\in\Hom(G_2,B_2',\mu_2')$ we have $\range(T(f))=\range(f)-1$. The corollary follows.
\end{proof}

\begin{proof}[Proof of Corollary~\ref{Lip_one_point_cor}]
Fix graphs $G$ and $G_2$ as in the corollary and let $(B,\mu)$ and $(B'_2,\mu'_2)$ be one-point BCs on $G$ and $G_2$ respectively. Assume first that $B'_2=\{(v,i)\}$ for some $i\in\{0,1\}$ and the same $v\in V[G]$ for which $B=\{v\}$. Let $g\unifin\Lip(G,B,\mu)$ and $f\unifin\Hom(G_2,B'_2,\mu'_2)$. Define also $\tilde{B}_2=\{(v,i),(v,1-i)\}$ and $\tilde{\mu}_2:\tilde{B}\to\Z$ by $\tilde{\mu}_2((v,i))=0$ and $\tilde{\mu}_2((v,1-i))=-1$. Let $\tilde{f}\in\Hom(G_2,\tilde{B}_2,\tilde{\mu}_2)$ (noting that $(\tilde{B}_2,\tilde{\mu}_2)$ are legal BC). By Theorem~\ref{Yadin_bijection_thm}, $\range(g)\eqd\range(\tilde{f})-1$. Next, we observe that by symmetry of the distribution of $f$ under negating all values, $f$ may be sampled by sampling $\tilde{f}$ with probability $\frac{1}{2}$ and $-\tilde{f}$ with probability $\frac{1}{2}$. Thus, $\range(f)\eqd\range(\tilde{f})$ which shows that $\range(g)\eqd\range(f)-1$ as required.

Finally, suppose $B'_2=\{(w,i)\}$ for some $i\in\{0,1\}$ and $w\in V[G]$ which is possibly different from $v$. Letting $(B_{v,2},\mu_{v,2})$ be the one-point BC with $B_{v,2}=\{(v,j)\}$ for some $j$ and $h\in\Hom(G_2,B_{v,2},\mu_{v,2})$, the corollary follows by noting that $\range(f)\eqd\range(h)$ since there exists a translation of the torus carrying $(v,j)$ into $(w,i)$.
\end{proof}

We proceed to deduce analogues of the theorems of Section~\ref{hom_main_results_sec} for Lipschitz height functions. We start by making a few observations. Fix a torus $G$ and let $G_2:=G\times \Z_2$. First, note that if $G$ is a non-linear torus, then $G_2$ is also a non-linear torus. Second, note that if $g\in\Lip(G,B,\mu)$ for some Lipschitz legal BC $(B,\mu)$ and if $t\in\N$, $v\in V[G]$ and $S$ is the inverse Yadin bijection defined in \eqref{inverse_Yadin_bijection_from_proof} then
\begin{equation}\label{at_least_t_relation}
g(v)\ge t\text{ if and only if }\max(S^{-1}(g)((v,0)),S^{-1}(g)((v,1)))\ge t.
\end{equation}
Finally, note that for any $r\in\N$,
\begin{equation}\label{volume_relation}
V_G(r) \le V_{G_2}(r)\le 2V_G(r),
\end{equation}
where $V_G(r)$ is the volume of a (graph) ball of radius $r$ in $G$ and $V_{G_2}(r)$ is the same in $G_2$.

\begin{proof}[Proof of Theorem~\ref{Lipschitz_height_at_point_thm}]
The theorem follows from the Yadin bijection Theorem~\ref{Yadin_bijection_thm}, from Theorem~\ref{height_at_point_thm} and observation~\eqref{at_least_t_relation}. For the second part of the theorem we add to these observation~\eqref{volume_relation} (which implies that if $V_G(\lceil t/2\rceil-1)\le \frac{1}{6}n_d$ then $V_{G_2}(\lceil t/2\rceil-1)\le \frac{1}{3}n_d$, where $n_d$ is the largest dimension of both $G$ and $G_2$) and for the last part of the theorem we use that if $B$ has full projection in $G$ then $B_2=\{(v,i)\ |\ v\in B, i\in\{0,1\}\}$ has full projection in $G_2$.
\end{proof}

\begin{proof}[Proof of Corollary~\ref{two_values_Lip_cor}]
Letting $(B_2',\mu_2')$ be the BC corresponding to $(B^\square, \Psi)$ as in Corollary~\ref{Lip_zero_one_cor}, we note that $B_2'$ has full projection in $G_2$ and $\mu_2'$ is zero. Thus, Corollary~\ref{even_lattice_zero_cor} implies that $f\unifin\Hom(G_2,B_2',\mu_2')$ will satisfy
\begin{equation*}
\frac{\E|\{(v,0)\in V_2^\even\ |\ f((v,0))\neq 0\}|}{|V_2^\even|}\le \exp\left(-\frac{c d}{\log^2 d}\right).
\end{equation*}
It remains to notice that $f((v,0))=0$ implies that $\max(f((v,0)),f((v,1)))\in\{0,1\}$ and to apply Corollary~\ref{Lip_zero_one_cor}.
\end{proof}

\begin{proof}[Proof of Theorem~\ref{Lipschitz_range_thm}]
 Let $(B_2',\mu_2')$ be either the BC corresponding to $(B,\Psi)$ by Corollary~\ref{Lip_zero_one_cor}, in the case that $g\unifin\Lip(G,B,\Psi)$, or a one-point BC on $G_2$, in the case that $g\unifin\Lip(G,B,\mu)$ for a one-point BC $(B,\mu)$. Applying Theorem~\ref{range_thm} to our setup, we have that there exists $d_0\in\N$, $C>0$ such that (so long as $d\ge d_0$) if we set
\begin{equation*}
k_2:=\min\left\{ m\in\N\ \big|\ V_{G_2}(m)\ge C\log^2 d\log|V[G_2]|\right\}
\end{equation*}
and let $f\unifin\Hom(G_2,B_2,\mu_2)$, then
\begin{equation*}
\P(\range(f)>2k_2+1)\le \frac{1}{|V[G_2]|^4}.
\end{equation*}
Hence Corollaries~\ref{Lip_zero_one_cor} and \ref{Lip_one_point_cor} imply that (for the $g$ of the theorem)
\begin{equation}\label{Lipschitz_range_with_k_2}
\P(\range(g)>2k_2)\le \frac{1}{|V[G_2]|^4}.
\end{equation}
Letting now
\begin{equation*}
k:=\min\left\{ m\in\N\ \big|\ V_G(m)\ge 2C\log^2 d\log|V[G]|\right\}
\end{equation*}
we observe that $k\ge k_2$ since $V_G(m)\le V_{G_2}(m)$ by \eqref{volume_relation} and $2\log|V[G]|\ge \log|V[G_2]|$ since $|V[G]|=\frac{1}{2}|V[G_2]|$ and $|V[G]|\ge 2^d$. Thus \eqref{Lipschitz_range_with_k_2} implies
\begin{equation*}
\P(\range(g)>2k)\le \frac{1}{|V[G_2]|^4}\le \frac{1}{|V[G]|^4},
\end{equation*}
as required.
\end{proof}

\begin{proof}[Proof of Corollary~\ref{two_sided_Lipschitz_range_cor}]
Let $(B_2,\mu_2)$ be a one-point BC on $G_2$. By Corollary~\ref{range_sharpness_cor} there exists $d_0\in\N$, $C_d,c_d>0$ such that (so long as $d\ge d_0$) if $f\unifin\Hom(G_2,B_2,\mu_2)$ then
\begin{equation*}
\P(c_d r_2 \le \range(f) \le C_d r_2)\ge 1-\frac{1}{|V[G_2]|^3},
\end{equation*}
where $r_2:=\min\left\{ m\in\N\ \big|\ V_{G_2}(m)\ge \log|V[G_2]|\right\}$. By Corollary~\ref{Lip_one_point_cor}, we deduce that
\begin{equation*}
\P(c_d r_2 \le \range(g)+1 \le C_d r_2)\ge 1-\frac{1}{|V[G_2]|^3}.
\end{equation*}
Since $\range(g)\ge 1$ with probability 1 and $|V[G]|=\frac{1}{2}|V[G_2]|$, we obtain
\begin{equation*}
\P(\frac{c_d}{2} r_2 \le \range(g) \le C_d r_2)\ge 1-\frac{1}{|V[G_2]|^3}\ge 1-\frac{1}{|V[G]|^3}
\end{equation*}
Hence, defining $r:=\min\left\{ m\in\N\ \big|\ V_{G}(m)\ge \log|V[G]|\right\}$, the corollary will follow if we show that $c_d'\le \frac{r_2}{r}\le C_d'$ for some $C_d',c_d'>0$. This, in turn, follows from \eqref{volume_relation} and Proposition~\ref{linear_volume_growth} (as in the proof of Corollary~\ref{range_sharpness_cor}).
\end{proof}

\begin{proof}[Proof of Theorem~\ref{Lipschitz_finite_range_thm}]
 The theorem follows directly from Theorem~\ref{finite_range_thm} using Corollaries~\ref{Lip_zero_one_cor} and \ref{Lip_one_point_cor}.
\end{proof}

\begin{proof}[Proof of Theorem~\ref{Lipschitz_linear_torus_thm}]
 Noting that if $G$ is $\lambda$-linear with $\lambda<\frac{1}{4\log 2}$ then $G_2=G\times \Z_2$ is $\lambda_2$-linear with $\lambda_2<\frac{1}{2\log 2}$, the theorem follows directly from Theorem~\ref{linear_torus_thm} using Corollary~\ref{Lip_one_point_cor} (with a possibly smaller $\alpha$ than in Theorem~\ref{linear_torus_thm}).
\end{proof}

\begin{proof}[Proof of Theorems~\ref{homogeneous_torus_thm} and \ref{homogeneous_tori_finite_range_thm} for Lipschitz case]
 Theorem~\ref{homogeneous_torus_thm} for the Lipschitz case follows by specializing Theorem~\ref{Lipschitz_height_at_point_thm} and Corollary~\ref{two_sided_Lipschitz_range_cor} to the case $G=\Z_n^d\times\Z_2^m$, in a similar way as it was done when proving the theorem for the homomorphism case. Theorem~\ref{homogeneous_tori_finite_range_thm} for the Lipschitz case follows by specializing Theorem~\ref{Lipschitz_linear_torus_thm} to the case $G=\Z_n^d$.
\end{proof}

\begin{section}{Linear Tori}\label{linear_tori_sec}
In this section we prove Theorem~\ref{linear_torus_thm}.

The idea of the proof is to reduce the problem to a problem on a one-dimensional torus and use the known fact that a random walk bridge has large fluctuations. We first introduce the definitions and lemmas we use and show how they suffice to prove the theorem. Then we give the proof of these lemmas.

Given $0<\lambda<\frac{1}{2\log 2}$, we fix parameters $\beta,\gamma>0$ to some arbitrary values satisfying
\begin{align}
&\gamma>9\beta,\label{gamma_beta_cond}\\
&\beta+\gamma+\lambda\log 2<1/2.\label{beta_gamma_lambda_cond}
\end{align}

We fix also a $\lambda$-linear torus $G$ and set
\begin{equation*}
n:=n_d\quad\text{ and }\quad m:=\prod_{i=1}^{d-1} n_i
\end{equation*}
so that
\begin{equation}\label{m_n_lambda_relation}
m\le\lambda\log n
\end{equation}
by definition of $\lambda$-linear torus. We let $G^-$ be the $(d-1)$-dimensional torus with dimensions $n_1,\ldots, n_{d-1}$ and fix a distinguished vertex of $G^-$ denoted by $\vec{0}$.
We fix a coordinate system on $V[G]$ such that
\begin{equation*}
V[G]=\{(x,y)\ |\ 0\le x\le n-1,\ y\in V[G^-]\}
\end{equation*}
and two vertices $(x_1,y_1), (x_2,y_2)$ are adjacent if $|x_1-x_2|\in\{1,n-1\}$ and $y_1=y_2$ or $x_1=x_2$ and $y_1$ is adjacent to $y_2$ in $G^-$. WLOG, we assume the coordinate system is chosen so that the boundary conditions are $B=\{(0,\vec{0})\}$ and $\mu((0,\vec{0}))=0$. Correspondingly, the bi-partition classes of $G$, $V^{\even}$ and $V^{\odd}$, are chosen so that $(0,\vec{0})\in V^{\even}$.

For $\eta>0$ and even $t$, let
\begin{equation*}
\begin{split}
\Omega_{\operatorname{low},\eta} &:= \{f\in\Hom(G,B,\mu)\ |\ \range(f)\le \eta n^\beta\},\\
\Omega_t &:= \{f\in\Hom(G,B,\mu)\ \big|\ |\{v\in V[G]\ |\ f(v)=t\}|\ge \frac{1}{2}n^{1-\beta}m\}.
\end{split}
\end{equation*}
Our first lemma is
\begin{lemma}\label{low_reduction_lemma}
$|\Omega_{\operatorname{low},1}|\le n^\beta|\Omega_0\cap\Omega_{\operatorname{low},2}|$.
\end{lemma}
We continue with the following definitions. For even $0\le x\le n-1$, let
\begin{equation*}
\begin{split}
W_x^0&=\{(z,w)\in V[G]\ \big|\ z\in \{x,x+1\}\text{ and } (z,w)\in V^\even\},\\
W_x^1&=\{(z,w)\in V[G]\ \big|\ z\in \{x+1,x+2\text{ mod }n\}\text{ and } (z,w)\in V^\odd\}.
\end{split}
\end{equation*}
We then say that $f\in\Hom(G,B,\mu)$ has a \emph{wall} at $x$ if $f$ is constant on $W_x^0$ and on $W_x^1$ (different constants on each set). We say that the wall is of height $h$ if $f$ equals $h$ on $W_x^0$. We call the wall an up-wall if $f(W_x^1)=f(W_x^0)+1$ and otherwise a down-wall. Let
\begin{equation*}
\begin{split}
W(f)&:=\{\text{even $0\le x\le n-1$}\ \big|\ \text{$f$ has a wall at $x$}\},\\
\Omegaw&:=\{f\in\Hom(G,B,\mu)\ \big|\ |W(f)|\le n^\gamma\}.
\end{split}
\end{equation*}
Our second (and main) lemma is
\begin{lemma}\label{Omega_0_w_lemma}
$|\Omega_0\cap \Omegaw|\le \frac{4(n^\gamma+4m)m2^{2m-1}}{n^{1-\beta-\gamma}}|\Hom(G,B,\mu)|$.
\end{lemma}

Next, we introduce a certain balancedness condition controlling the difference in the number of up-walls and down-walls of a function. For $f\in\Hom(G,B,\mu)$ let $s(f)\in\{-1,1\}^{W(f)}$ be defined by $s(f)_x = 1$ if the wall at $x$ is an up-wall and $s(f)_x=-1$ if it is a down-wall. Let
\begin{equation*}
\Omega_b:=\left\{f\in\Hom(G,B,\mu)\ \bigg|\ \big|\sum_{x\in W(f)} s(f)_x\big|> |W(f)|-\frac{n^{\gamma-\beta}}{8}\right\}.
\end{equation*}
\begin{lemma}\label{Omega_b_low_w_lemma}
There exists $n_0=n_0(\beta,\gamma)$ such that if $n\ge n_0$ then $|\Omega_b\cap \Omega_{\operatorname{low},2}\cap \Omegaw^c|\le 10n^{2\beta}|\Omega_b^c\cap \Omega_{\operatorname{low},4}\cap \Omegaw^c|$.
\end{lemma}

We continue with one final lemma.
\begin{lemma}\label{low_with_walls_lemma}
There exists $n_0=n_0(\beta,\gamma)$ and $C>0$ such that if $n\ge n_0$ we have $|\Omega_b^c\cap\Omega_{\operatorname{low},4}|\le \frac{C}{n^{(\gamma-3\beta)/2}}|\Hom(G,B,\mu)|$.
\end{lemma}

\begin{proof}[Proof of Theorem~\ref{linear_torus_thm}]
Putting the previous 4 lemmas together, we finally obtain, for $n\ge n_0(\beta,\gamma)$ for a sufficiently large $n_0(\beta,\gamma)$ and some $C,C',C''>0$,
\begin{equation*}
\begin{split}
|\Omega_{\operatorname{low},1}|&\overset{\text{Lemma~\ref{low_reduction_lemma}}}{\le} n^\beta|\Omega_{\operatorname{low},2}\cap\Omega_0|\le n^\beta\left(|\Omega_0\cap \Omegaw|+|\Omega_{\operatorname{low},2}\cap \Omegaw^c|\right)\le\\
&\le n^\beta|\Omega_0\cap \Omegaw|+n^\beta\left(|\Omega_b\cap\Omega_{\operatorname{low},2}\cap \Omegaw^c|+|\Omega_b^c\cap\Omega_{\operatorname{low},2}\cap \Omegaw^c|\right)\overset{\text{Lemma~\ref{Omega_b_low_w_lemma}}}{\le} \\
&\le n^\beta|\Omega_0\cap \Omegaw|+n^\beta(10n^{2\beta}+1)|\Omega_b^c\cap\Omega_{\operatorname{low},4}|\overset{\text{Lemmas~\ref{Omega_0_w_lemma} and \ref{low_with_walls_lemma}}}{\le} \\
&\le \left(\frac{4(n^\gamma+4m)m2^{2m-1}}{n^{1-2\beta-\gamma}}+\frac{C(10n^{2\beta}+1)}{n^{(\gamma-5\beta)/2}}\right)|\Hom(G,B,\mu)|\overset{\eqref{m_n_lambda_relation}}{\le}\\
&\le \left(\frac{4\lambda\log n(n^\gamma+4\lambda\log n)}{n^{1-2\beta-\gamma-2\lambda\log 2}}+\frac{C'}{n^{(\gamma-9\beta)/2}}\right)|\Hom(G,B,\mu)|\overset{\text{\eqref{gamma_beta_cond} and \eqref{beta_gamma_lambda_cond}}}{\le}\\
&\le C''n^{-\alpha'}|\Hom(G,B,\mu)|
\end{split}
\end{equation*}
for some $\alpha'=\alpha'(\beta,\gamma,\lambda)>0$. Hence if $f\unifin\Hom(G,B,\mu)$ then $\P(f\in\Omega_{\operatorname{low},1})\le C''n^{-\alpha'}$ proving the theorem with $\alpha=\min(\alpha', \beta)$. Note that the restriction that $n\ge n_0(\beta,\gamma)$ is implicitly imposed in the statement of the theorem since the bound \eqref{linear_tori_thm_bound} is meaningless if its right-hand side is larger than 1.
\end{proof}

\begin{proof}[Proof of Lemma~\ref{low_reduction_lemma}]
If $f\in\Omega_{\operatorname{low},1}$ then $f$ takes at most $n^\beta$ distinct values, all in $[-n^\beta+1, n^\beta-1]$. Since $|V^\even|=\frac{1}{2}nm$, it follows by the pigeonhole principle that $f$ takes some even value at least $\frac{1}{2}n^{1-\beta}m$ times. Thus,
\begin{equation*}
\Omega_{\operatorname{low},1}\subseteq \bigcup_{t\in[-n^\beta+1,n^{\beta}-1]\cap2\Z} \Omega_t.
\end{equation*}
Hence, since $|[-n^\beta+1,n^\beta-1]\cap2\Z|\le n^\beta$, the lemma will follow once we show that for all even $t\neq 0$, $|\Omega_t\cap\Omega_{\operatorname{low},1}|\le |\Omega_0\cap\Omega_{\operatorname{low},2}|$ (it is obvious for $t=0$). Fix an even $t\neq 0$. For $f\in\Hom(G,B,\mu)$, let
\begin{equation*}
A_t(f) = \text{connected component of $(0,\vec{0})$ in $V[G]\setminus \left\{v\in V[G]\ \big|\ f(v)=\frac{t}{2}\right\}$}.
\end{equation*}
We define $R_t:\Hom(G,B,\mu)\to\Hom(G,B,\mu)$ by
\begin{equation*}
R_t(f)(v) = \begin{cases}f(v)&v\in A_t(f)\\t-f(v)&v\notin A_t(f)\end{cases}
\end{equation*}
One can verify simply that for all $f\in\Hom(G,B,\mu)$, $R_t(f)\in\Hom(G,B,\mu)$ since $(0,\vec{0})\in A_t(f)$ and if $u,w\in V[G]$ satisfy $u\adj{G}w$, $u\in A_t(f)$ and $w\notin A_t(f)$ then necessarily $f(u)=\frac{t}{2}-1$ and $f(w)=\frac{t}{2}$. In addition, $R_t(f)(v)=0$ for all $v\in V[G]$ for which $f(v)=t$ (since such $v$ are never in $A_t(f)$) and hence $R_t(\Omega_t)\subseteq\Omega_0$. Furthermore, it is simple to verify that $\range(R_t(f))\le 2\range(f)$ for all $f\in\Hom(G,B,\mu)$ and hence $R_t(\Omega_t\cap\Omega_{\operatorname{low},1})\subseteq \Omega_0\cap\Omega_{\operatorname{low},2}$. Finally, it is straightforward to check that $A_t(f)=A_t(R_t(f))$ so that $R_t(R_t(f))=f$ for all $f\in\Hom(G,B,\mu)$, implying that $R_t$ is one-to-one. Hence $|\Omega_t\cap\Omega_{\operatorname{low},1}|\le |\Omega_0\cap\Omega_{\operatorname{low},2}|$ as required.
\end{proof}

\begin{proof}[Proof of Lemma~\ref{Omega_0_w_lemma}]
For an integer $0\le k\le n^{\gamma}$ and $f\in\Hom(G,B,\mu)$, let
\begin{equation*}
\begin{split}
W^0(f)&:=\{\text{even $0\le x\le n-1$}\ \big|\ \text{$f$ has an up-wall at $x$ of height 0}\},\\
\Omega^{0,k}_{\operatorname{w}}&:=\{f\in\Hom(G,B,\mu)\ \big|\ |W^0(f)|=k\}.
\end{split}
\end{equation*}
We clearly have $\Omegaw\subseteq \cup_{k=0}^{\lfloor n^\gamma\rfloor}\Omega^{0,k}_{\operatorname{w}}$ and hence it will be sufficient to show for each $0\le k\le n^{\gamma}$ that
\begin{equation}\label{Omega_w_k_bound}
|\Omega_0\cap\Omega^{0,k}_{\operatorname{w}}|\le \frac{2(k+4m)m2^{2m-1}}{n^{1-\beta}}|\Hom(G,B,\mu)|.
\end{equation}
Next, for $f\in\Hom(G,B,\mu)$ we let
\begin{equation*}
\tilde{W}^0(f):=\{\text{even $0\le x\le n-1$}\ \big|\ \text{there exists $v\in W_x^0$ such that $f(v)=0$}\}.
\end{equation*}
We then have $\tilde{W}^0(f)\supseteq W^0(f)$. We note that by the pigeon-hole principle, if $f\in \Omega_0$ then
\begin{equation}\label{many_zero_verts}
|\tilde{W}^0(f)|\ge \frac{1}{2}n^{1-\beta}
\end{equation}
where we used that $f$ can only take the value 0 on vertices of $V^{\even}$.
Points of $\tilde{W}^0(f)$ are potential ``building sites'' for walls using the transformation we will now define. First, for each even $0\le x\le n-1$ and each $v\in W_x^0$, let $s^v_x=(s^v_{x,1},\ldots, s^v_{x,2m})$ be some fixed permutation of $W^0_x\cup W_x^1$ with the properties that $s^v_{x,1}=v$ and for each $2\le i\le 2m$, $s^v_{x,i}$ is adjacent in $G$ to $s^v_{x,j}$ for some $1\le j<i$. Next, for a function $f\in \Hom(G,B,\mu), w\in V[G]$ and $t\in\Z$, define $P_{w,t}(f)$, the peak (or lake) of $f$ around $w$ from height $t$, by
\begin{equation*}
P_{w,t}(f):=\text{connected component of $w$ in $V[G]\setminus \{u\in V[G]\ |\ f(u)=t\}$}.
\end{equation*}
Then define the reflection (of the peak of $w$ around $t$) transformation $R_{w,t}$ (different from the one used in the proof of Lemma~\ref{low_reduction_lemma}) on the set of functions $f\in\Hom(G,B,\mu)$ for which $(0,\vec{0})\notin P_{w,t}(f)$ by
\begin{equation*}
R_{w,t}(f)(u) = \begin{cases}f(u)&u\notin P_{w,t}(f)\\2t-f(u)&u\in P_{w,t}(f)\end{cases}.
\end{equation*}
It is straightforward to verify that $R_{w,t}(f)\in\Hom(G,B,\mu)$ and $R_{w,t}(R_{w,t}(f))=f$ on this set of functions. Finally, let $\Omega_{x,0}:=\{f\in\Hom(G,B,\mu)\ |\ x\in\tilde{W}^0(f)\}$, fix some (arbitrary) total order on $V[G^-]$ and define the ``building transformation'' $B_x:\Omega_{x,0}\to\Hom(G,B,\mu)$ using the following algorithm:
\begin{enumerate}
\item Set $f_1:=f$ and define $v$ to be the vertex with minimal second coordinate among all $w\in W_x^0$ with $f(w)=0$. For $1\le i\le 2m$, set $w_i:=s^v_{x,i}$.
\item Iteratively for $2\le i\le 2m$ set
\begin{equation*}
f_i:=
\begin{cases}
f_{i-1}& (w_i\in W_x^0\text{ and }f_{i-1}(w_i)=0)\text{ or }\\
&(w_i\in W_x^1\text{ and }f_{i-1}(w_i)=1),\\
R_{w_i,1}(f_{i-1})& w_i\in W_x^0\text{ and }f_{i-1}(w_i)=2,\\
R_{w_i,0}(f_{i-1})& w_i\in W_x^1\text{ and }f_{i-1}(w_i)=-1.
\end{cases}
\end{equation*}
\item Set $B_x(f):=f_{2m}$.
\end{enumerate}
{\bf Claim:}\begin{enumerate}
\item $B_x(f)$ is well-defined for $f\in\Omega_{x,0}$.
\item $B_x(f)$ has an up-wall at $x$ of height $0$.
\item $B_x(f)(w)=f(w)$ for all $w\in V[G]$ such that $f(w)\in\{0,1\}$.
\end{enumerate}
The claim follows by showing that for all $1\le i\le 2m$ we have
\begin{enumerate}
\item[(a)] $f_i$ is well-defined for $f\in\Omega_{x,0}$.
\item[(b)] $f_i(w_i)=0$ if $w_i\in W_x^0$ and $f_i(w_i)=1$ if $w_i\in W_x^1$.
\item[(c)] For $i\ge 2$, $f_i(w)=f_{i-1}(w)$ for all $w\in V[G]$ such that $f_{i-1}(w)\in\{0,1\}$.
\end{enumerate}
For $i=1$, this follows from the fact that $f\in\Omega_{x,0}$ along with the fact that $w_1=s^v_{x,1}=v$. For $2\le i\le 2m$, it follows by induction on $i$ as follows. Fix $2\le i\le 2m$ and let $a\in\{0,1\}$ be such that $w_i\in W_x^a$. By definition of $s^v_x$, $w_i$ is adjacent in $G$ to $w_j$ for some $1\le j< i$. We necessarily have $w_j\in W_x^{1-a}$. By property (b) above for $j$ and property (c) above for all $j<k<i$ we see that $f_{i-1}(w_j)=1-a$. Hence $f_{i-1}(w_i)\in\{-a,2-a\}$. If $f_{i-1}(w_i)=a$ (noting that $a\in\{-a,2-a\}$) we have $f_i=f_{i-1}$ and (a), (b) and (c) follow for $i$. Otherwise, if $a=0$ and $f_{i-1}(w_i)=2$ then $P_{w_i,1}\cap\{w\ |\ f_{i-1}(w)\le 1\}=\emptyset$ and if $a=1$ and $f_{i-1}(w_i)=-1$ then $P_{w_i,0}\cap\{w\ |\ f_{i-1}(w)\ge 0\}=\emptyset$. In both cases, we deduce that (a), (b) and (c) above are satisfied for $i$.

Continuing, we will also use the fact that $B_x(f)$ is formed from $f$ by performing at most $2m-1$ reflections, each being either around 0 or around 1 (where by such reflections we mean applications of $R_{w,0}$ or $R_{w,1}$). This implies that
\begin{equation}\label{B_x_inverse_bound}
|B_x^{-1}(B_x(f))|\le m2^{2m-1}
\end{equation}
since in order to invert $B_x$, we need only know which $v\in W_x^0$ was chosen in step 1 of the definition of $B_x(f)$ and also for each of the following $2m-1$ steps, whether or not a reflection was performed.

By parts 2 and 3 of the above claim, we have that for any $f\in\Omega_{x,0}$,
\begin{equation}\label{walls_remain_after_building}
W^0(B_x(f))\supseteq (W^0(f)\cup\{x\}).
\end{equation}
In addition, we claim that
\begin{equation}\label{wall_number_after_building}
|W^0(B_x(f))|\le |W^0(f)|+4m.
\end{equation}
To see this, note that as mentioned above, we can reconstruct $f$ from $B_x(f)$ by performing at most $2m-1$ reflections around 0 and 1 (since $R_{w,t}$ is the inverse of itself). However, note that for any $g\in\Hom(G,B,\mu)$ and $w\in V[G]$, $P_{w,0}(g)$ can intersect at most two up-walls of height 0 (meaning that $P_{w,0}(g)\cap (W_x^0\cup W_x^1)$ can be non-empty for at most two values of $x\in W^0(g)$) since walls of height 0 act as a ``barrier''. Similarly $P_{w,1}(g)$ can intersect at most two up-walls of height 0. Hence, when reconstructing $f$ from $B_x(f)$ the number of up-walls can change by at most $2(2m-1)\le 4m$.

We finally arrive at the proof of \eqref{Omega_w_k_bound}. Fix an integer $0\le k\le n^\gamma$ and let $\A':=\{(f,x)\ |\ f\in\Omega_0\cap\Omega^{0,k}_{\operatorname{w}}, x\in \tilde{W}^0(f)\}$. Note that by \eqref{many_zero_verts},
\begin{equation}\label{A_prime_lower_bound}
|\A'|\ge \frac{1}{2}n^{1-\beta}|\Omega_0\cap\Omega^{0,k}_{\operatorname{w}}|.
\end{equation}
Define $T:\A'\to\Hom(G,B,\mu)$ by $T((f,x)):=B_x(f)$. We claim that for any $g\in T(\A')$ we have
\begin{equation}\label{T_inverse_size_bound}
|T^{-1}(g)|\le (k+4m)m2^{2m-1}.
\end{equation}
To see this, first note that by \eqref{walls_remain_after_building}, for any $(f,x)\in\A'$ such that $T((f,x))=g$ we have $x\in W^0(g)$. Then note that $|W^0(g)|\le k+4m$ by \eqref{wall_number_after_building} and the definition of $\Omega^{0,k}_{\operatorname{w}}$. Finally note that by \eqref{B_x_inverse_bound}, given $x\in W^0(g)$ there are at most $m2^{m-1}$ pairs $(f,x)\in\A'$ such that $B_x(f)=g$. These arguments imply \eqref{T_inverse_size_bound}. We deduce from \eqref{T_inverse_size_bound} that
\begin{equation*}
|T(\A')|\ge \frac{|\A'|}{(k+4m)m2^{2m-1}}.
\end{equation*}
Putting this bound together with \eqref{A_prime_lower_bound} we obtain
\begin{equation*}
\frac{|\Omega_0\cap\Omega^{0,k}_{\operatorname{w}}|}{|\Hom(G,B,\mu)|}\le \frac{|\Omega_0\cap\Omega^{0,k}_{\operatorname{w}}|}{|T(\A')|}\le \frac{2(k+4m)m2^{2m-1}}{n^{1-\beta}}
\end{equation*}
proving \eqref{Omega_w_k_bound}.\qedhere
\end{proof}
\begin{proof}[Proof of Lemma~\ref{Omega_b_low_w_lemma}]
Define $\Omega_1:=\Omega_b\cap \Omega_{\operatorname{low},2}\cap \Omegaw^c$ and $\Omega_2:=\Omega_b^c\cap \Omega_{\operatorname{low},4}\cap \Omegaw^c$.
Let $\ell:=\lceil\frac{n^{\gamma-\beta}}{8}\rceil$ and $I:=\{1+i\ell\ \big|\ i\in \left[0,\lceil 2n^\beta\rceil\right]\cap\Z\}$. Using \eqref{gamma_beta_cond} and the assumption that $n\ge n_0(\beta,\gamma)$ we have $\max I\le \frac{n^\gamma}{2}$ if $n_0(\beta,\gamma)$ is large enough. We also have $2n^\beta+1\le|I|\le 2n^\beta+2$. For $f\in\Omega_1$, let $k:=|W(f)|$, $x_1,\ldots, x_k$ be the elements of $W(f)$ sorted in increasing order and for $1\le i\le k$, let $h_i$ be the height of the wall at $x_i$.
Fixing an $f\in\Omega_1$ we see that $k>n^\gamma$ by definition of $\Omegaw$ implying that $k> 2\max I$. Hence, Since $f\in\Omega_{\operatorname{low},2}$ and $|I|\ge 2n^\beta+1$ there must exist distinct $i,j\in I$ such that $h_i=h_j$.
Letting $\Omega_{i,j}:=\{f\in\Omega_1\ |\ h_i=h_j\}$ we have shown that $\Omega_1\subseteq\cup_{\substack{i,j\in I\\i<j}}\Omega_{i,j}$. Hence the lemma will follow by establishing
\begin{equation}\label{Omega_i_j_Omega_2_sizes}
|\Omega_{i,j}|\le|\Omega_2|
\end{equation}
for all $i,j\in I$ satisfying $i<j$. Fix such $i,j$ and $f\in \Omega_{i,j}$. We define a new function $T^{i,j}(f)$ by reflecting the region between the walls at $x_i$ and $x_j$ around height $h_i$, that is,
\begin{equation*}
 T^{i,j}(f)((x,y)) := \begin{cases}
                      f((x,y))& x\le x_i\text{ or }x>x_j\\
                      2h_i - f((x,y))& x_i<x\le x_j
                     \end{cases}.
\end{equation*}
It is straightforward to verify that $T^{i,j}(f)\in\Hom(G,B,\mu)$ since $h_i=h_j$, that $W(T^{i,j}(f))=W(f)$ and that $s(T^{i,j}(f))(x_p)$ equals $-s(f)(x_p)$ if $i\le p<j$ and equals $s(f)(x_p)$ otherwise. Informally, $T^{i,j}$ ``flips'' $j-i$ of the walls of $f$. Since $j-i$ satisfies
\begin{equation*}
\frac{n^{\gamma-\beta}}{8}\le \ell\le j-i\le \max I\le \frac{1}{2}{k}
\end{equation*}
and $f\in\Omega_b$, it follows that $T^{i,j}(f)\in\Omega_b^c$. Checking also that $\range(T^{i,j}(f))\le2\range(f)$ we deduce that $T^{i,j}(f)\in\Omega_2$. Finally, noting that $T^{i,j}$ is one-to-one on $\Omega_{i,j}$, we arrive at \eqref{Omega_i_j_Omega_2_sizes}.
\end{proof}

For the proof of Lemma~\ref{low_with_walls_lemma}, we need the following standard claim about simple random walk.

{\bf Claim:} There exists $C>0$ such that for all integer $k,s$ and $t$ satisfying that $k-s$ is even and $k\ge |s|+2$ we have that if $X_1,\ldots, X_k\in\{-1,1\}$ are IID with $\P(X_1=1)=\frac{1}{2}$ then
\begin{equation*}
 \P\left(\sum_{i=1}^{\lfloor k/2\rfloor} X_i=t\ \Big|\ \sum_{i=1}^k X_i=s\right)\le \frac{C}{\sqrt{k-|s|}}.
\end{equation*}

\begin{proof}[Proof of Lemma~\ref{low_with_walls_lemma}]
We start by enlarging the class of functions we consider beyond $\Hom(G,B,\mu)$. We let $\widetilde{\Hom}(G,B,\mu)$ be all functions $f:V[G]\to\Z$ which satisfy $f(0,\vec{0})=0$ (recalling that in this section $B=\{(0,\vec{0})\}$ and $\mu((0,\vec{0}))=0$) and satisfy $|f(v)-f(w)|=1$ for all $v,w\in V[G]$ except possibly when $v\in W_0^0$ and $w\in W_{n-2}^1$ or when $v\in W_{n-2}^1$ and $w\in W_0^0$. In other words, $\widetilde{\Hom}(G,B,\mu)=\Hom(\tilde{G},B,\mu)$ where $\tilde{G}$ is the same graph as $G$ but with the edges between vertices of $W_0^0$ and $W_{n-2}^1$ removed. We define $W(f)$ and $s(f)$ for functions $f\in\widetilde{\Hom}(G,B,\mu)$ in exactly the same way as for functions in $\Hom(G,B,\mu)$.

Given $f\in\widetilde{\Hom}(G,B,\mu)$ and $x\in W(f)$ we define a new function $S_x(f)$ by shifting the wall of $f$ at $x$ from an up-wall to a down-wall and vice versa and correspondingly shifting the whole function $f$ to the ``right'' of $x$, as follows
\begin{equation*}
 S_x(f)(v):=\begin{cases}
                 f(v)&v\in W_y^0\text{ for some even $y\le x$ or }\\
                     &v\in W_y^1\text{ for some even $y<x$}\\
                 f((z,w))-2s(f)_x&\text{otherwise}
                \end{cases}.
\end{equation*}
We readily verify that $S_x(f)\in\widetilde\Hom(G,B,\mu)$, $W(S_x(f))=W(f)$, $S_x(S_x(f))=f$ and if $y\in W(f)$ then $s(S_x(f))_y$ equals $s(f)_y$ if $y<x$ and equals $-s(f)_y$ if $y\ge x$. In addition, we check that if $x,y\in W(f)$ then $S_x(S_y(f))=S_y(S_x(f))$. We finally check that if $f\in\Hom(G,B,\mu)$ and we have distinct $x_1,\ldots, x_\ell\in W(f)$ for some $\ell$ then $(S_{x_1}\circ\cdots\circ S_{x_\ell})(f)\in\Hom(G,B,\mu)$ iff $\sum_{i=1}^\ell s(f)_{x_i}=0$.

We define an equivalence relation $\sim$ on $\Hom(G,B,\mu)$ by $f\sim g$ iff $g=(S_{x_1}\circ\cdots\circ S_{x_\ell})(f)$ for some $\ell$ and distinct $x_1,\ldots, x_\ell\in W(f)$. Denoting the equivalence class of $f$ by $[f]$, we have by the previous paragraph that $[f]$ is in bijection with $\left\{s_1,\ldots, s_{|W(f)|}\in\{-1,1\}\ \big|\ \sum_{i=1}^{|W(f)|} s_i=\sum_{x\in W(f)} s(f)_x\right\}$ via the correspondence $s_i=s(f)_{y_i}$, where $(y_i)_{i=1}^{|W(f)|}$ is $W(f)$ sorted in increasing order. We wish to show that for some $C>0$, $|\Omega_b^c\cap\Omega_{\operatorname{low},4}|\le Cn^{(3\beta-\gamma)/2}|\Hom(G,B,\mu)|$. To this end, it is sufficient to show that for any $f\in \Omega_b^c$ we have
\begin{equation}\label{equiv_class_intersect_low}
 |[f]\cap\Omega_{\operatorname{low},4}|\le Cn^{(3\beta-\gamma)/2}|[f]|.
\end{equation}
Fix $f\in \Omega_b^c$ and let $k:=|W(f)|$ and $y_1,\ldots, y_k$ be the elements of $W(f)$ in increasing order. Fix $v:=(y_{\lfloor k/2\rfloor}+1, 0)$. Define $h:=f(v)-\sum_{i=1}^{\lfloor k/2\rfloor} s(f)_{y_i}$. Then it is straightforward to see that for each $g\in[f]$ we have
\begin{equation*}
 g(v) = h + \sum_{i=1}^{\lfloor k/2\rfloor}s(g)_{y_i}.
\end{equation*}
Let $X_1,\ldots, X_k$ be IID random variables with $\P(X_i=1)=\frac{1}{2}$ and set $s:=\sum_{i=1}^k s(f)_{y_i}$. Let $g$ be sampled uniformly at random from $[f]$. Using the bijection above we see that $g(v) = h + Z$ where the random variable $Z$ is distributed as $\sum_{i=1}^{\lfloor k/2\rfloor} X_i$ conditioned that $\sum_{i=1}^{k} X_i = s$. Using now that $f\in \Omega_b^c$ we have that $|s|\le k-\frac{n^{\gamma-\beta}}{8}$. Hence, recalling \eqref{gamma_beta_cond} and our assumption that $n\ge n_0(\beta, \gamma)$ we see that $k-|s|\ge 2$ if $n_0(\beta,\gamma)$ is large enough. Thus it follows from the Claim above that for any $t$,
\begin{equation*}
 \P(g(v)=t)\le\frac{C'}{n^{(\gamma-\beta)/2}}
\end{equation*}
for some $C'>0$. Hence, $\P(g\in\Omega_{\operatorname{low},4})\le \frac{C}{n^{(\gamma-3\beta)/2}}$ for some $C>0$, proving \eqref{equiv_class_intersect_low} and the lemma.
\end{proof}
\end{section}

\begin{section}{Open Questions}\label{open_questions_sec}
In the following questions, by the standard observables for a random function $f:V[G]\to\Z$ (for some graph $G$), we mean $\var(f(v))$ for generic vertices $v$ and $\E\range(f)$.
\begin{enumerate}
\item Two dimensions: When $G$ is the $n\times n$ torus (with, say, the one-point BC $(B,\mu)$) and $f\unifin\Hom(G,B,\mu)$, what is the order of magnitude of our standard observables?
Does $f$ converge weakly to the Gaussian free field?

\item Low dimensions: What is the smallest dimension $d$ for which the random height function is still typically flat (as in Theorem~\ref{homogeneous_torus_thm}, say)? Is it for all $d\ge 3$ (as Figure~\ref{2d_3d_100_fig} hints)?

\begin{figure}[t!]
\centering
{\includegraphics[width=\textwidth,viewport=25 130 990 640,clip]{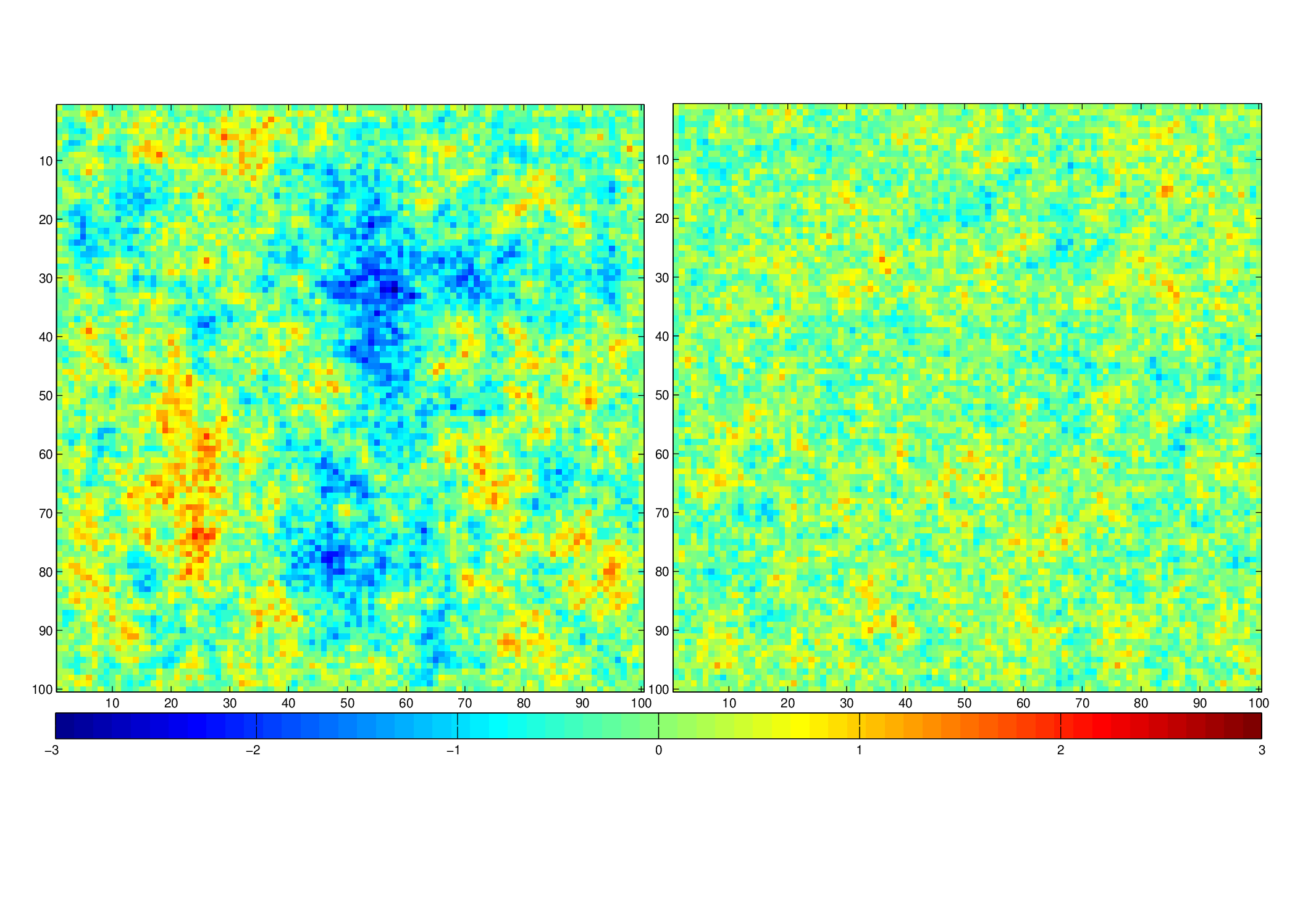}}
{\it
\caption{Samples of a uniformly random Lipschitz function taking real values which differ by at most one between adjacent vertices. The left picture shows a sample on the 100 x 100 torus and the right picture shows the middle slice (at height 50) of a sample on the $100\times100\times100$ torus, both conditioned to have boundary values in the $[-\frac{1}{2},\frac{1}{2}]$ interval. Sampled using coupling from the past \cite{PW96}.\label{2d_3d_100_cont_fig}}
}
\end{figure}

\item $M$-Lipschitz functions: For a graph $G$ and $M\in\N$, consider the model of functions $f:V[G]\to\Z$ satisfying $|f(v)-f(w)|\le M$ subject to some boundary conditions. The case $M=1$ is the case of Lipschitz functions considered in this paper. If $G=\Z_n^d$ and $f$ is sampled uniformly from such functions (say, with a one-point BC), what is the order of magnitude of our standard observables?
If one takes $M=M(d)$ large enough and considers high dimensions $d$, do these quantities behave differently (in terms of $n$) than for the Lipschitz functions considered in this paper? How do these quantities behave in dimension 2? Figure~\ref{2d_3d_100_cont_fig} shows samples of the ``limiting'' height function model: when the function $f$ is sampled uniformly from all $f:V[G]\to\R$ (that is, $\Z$ is replaced by $\R$) satisfying given boundary conditions and $|f(v)-f(w)|\le 1$ whenever $v$ is adjacent to $w$ in $G$.

\item Entropy repulsed surface: Let $G=\Z_n^d$ and $f\unifin\Hom(G,B,\mu)$ for, say, a one-point BC. Condition that $f$ is everywhere non-negative. What is the order of magnitude of our standard observables?

\item Sloped surfaces: Let $G$ be a cube in $\Z^d$ with side length $n$ (the same as $\Z_n^d$, but with non-periodic boundary) and $f\unifin\Hom(G,B,\mu)$ for boundary conditions $(B,\mu)$ which impose a slope to $f$. For example, $B$ can be the boundary defined in \eqref{zero_BC_def} and $\mu(b)$ can be defined by the closest even integer to $\alpha b_1$, where $\alpha\in(0,1)$ and $b_1$ is the first coordinate of $b$. What is the order of magnitude of the fluctuations of $f$ from the expected sloped surface?

\item Uniform 3-coloring and anti-ferromagnetic 3-state Potts models: As explained in Section~\ref{3_coloring_ice_section}, when $G=\Z_n^d$, $(B,\mu)$ is the zero BC (say) and $f\unifin\Hom(G,B,\mu)$, the model is equivalent to the uniform 3-coloring model (anti-ferromagnetic 3-state Potts model at zero temperature) with zero BC and thus we could deduce that a such a random 3-coloring will typically be nearly constant on the even sub-lattice. For which boundary conditions does this phenomenon hold (in particular, what happens for a one-point BC)? does it persist for the Potts model with small positive values of the temperature?

\item Infrared bound: Can the technique of the infrared bound be applied to the homomorphism model to obtain a simpler derivation of concentration results? For example, can one use this technique to show that the variance of the height at a generic vertex of a random homomorphism on $\Z_n^d$ (with the one-point BC, say) is bounded uniformly in $n$, when $d\ge 3$? Related questions are mentioned as an open problem in the survey on the subject by Marek Biskup \cite[Problem 8.3]{B09}.

\item Non-periodic boundary conditions: All of our results have been proved for tori $G$. Do these results extend to boxes in $\Z^d$ (with non-periodic boundary)? As explained in Section~\ref{3_coloring_ice_section}, it is of interest to make this extension since the model on such boxes (with certain boundary conditions) is equivalent to the uniform 3-coloring model. However, our methods of proof rely on the periodicity, for example, in our definition of the shift transformation and the fact that it is invertible given the location of the level set (see Figure~\ref{shift_transform_fig}, Section~\ref{shift_transformation_sec} and \eqref{inverse_shift_eq}).

\item General tori: We have shown that in high-dimensions, random homomorphism and Lipschitz height functions are typically flat on non-linear tori and typically rough on linear tori. However, not all tori fall under our definitions of non-linear and linear tori (\eqref{non_linearity_cond} and \eqref{linear_cond}). What is the typical behavior of random homomorphism and Lipschitz height functions on tori which are neither non-linear, nor linear?

\item Odd cutsets: How different are the odd cutsets introduced in this paper from ordinary cutsets?
For example, define $\MCut_L$ to be all minimal edge cutsets in
$\Z^d$ separating the origin from infinity and having exactly $L$
edges and define $\OMCut_L$ to be the subset of these which are odd
(see Section~\ref{preliminaries_sec} for more precise definitions).
For large $d$ and $L$, it is shown in \cite{BB07} (and in
\cite{LM98}) that $\exp\left(\frac{c\log d}{d} L\right)\le
|\MCut_L|\le \exp\left(\frac{C\log d}{d} L\right)$ for some $C,c>0$.
Is $|\OMCut_L|$ of the same order of magnitude or is it only of
order $\exp\left(\frac{C}{d}L\right)$? What is the scaling limit of
odd cutsets? Following \cite{S04}, it seems reasonable that the
scaling limit of a uniformly sampled cutset from $\MCut_L$ is super
Brownian motion. However, if the cutset is uniformly sampled from
$\OMCut_L$, it may well be the case that the limit is different,
with the random cutset typically containing a macroscopic cube in
its interior.

\end{enumerate}

\end{section}

\end{document}